\documentclass[english,runningheads,11pt]{llncs}
\usepackage{tabularx,booktabs,multirow,delarray,array}
\usepackage{graphicx,amssymb,amsmath,amssymb}
\usepackage{latexsym}

\usepackage[in]{fullpage}

\def\etal{\textsl{et~al. }}
\def\calP{\mathcal{P}}
\def\calM{\mathcal{M}}
\def\calR{\mathcal{R}}

\def\calF{\mathcal{F}}
\def\calQ{\mathcal{Q}}
\def\calD{\mathcal{D}}
\def\calV{\mathcal{V}}
\def\bay{bay(\overline{cd})}
\def\canal{canal(x,y)}
\def\st{$s$-$t$}
\newcommand{\Tri}{\mbox{$T\!r\!i$}}

\newenvironment{proof}{\noindent {\textbf{Proof:}}\rm}{\hfill $\Box$\rm}
\newtheorem{observation}{Observation}

\begin{document}

\title{Two-Point $L_1$ Shortest Path Queries in the Plane\thanks{A preliminary version appeared in the 30th Annual
Symposium on Computational Geometry (SoCG 2014).}}

\author{Danny Z. Chen\inst{1}\thanks{D.Z.~Chen's research was supported in part by NSF
under Grant CCF-1217906.}
\and Rajasekhar Inkulu\inst{2}\thanks{R.~Inkulu's research was supported in part by IITG startup grant.}
\and
Haitao Wang\inst{3}\thanks{Corresponding author. H.~Wang's research was supported in part by NSF under Grant CCF-1317143.}
}

 \institute{
 Department of Computer Science and Engineering\\
 University of Notre Dame, Notre Dame, IN 46556, USA\\
 \email{dchen@nd.edu}
 \and
Department of Computer Science and Engineering \\
 Indian Institute of Technology Guwahati, Guwahati 781039, Assam, India \\
\email{rinkulu@iitg.ac.in}\\
\and
  Department of Computer Science\\
  Utah State University, Logan, UT 84322, USA\\
  \email{haitao.wang@usu.edu}
}

\maketitle

\pagenumbering{arabic}
\setcounter{page}{1}

\begin{abstract}
Let $\calP$ be a set of $h$ pairwise-disjoint polygonal
obstacles with a total of $n$ vertices in the plane.  We consider the problem of
building a data structure that can quickly compute an $L_1$ shortest obstacle-avoiding
path between any two query points $s$ and $t$.
Previously, a data structure of
size $O(n^2\log n)$ was constructed in $O(n^2\log^2 n)$ time that
answers each two-point query in $O(\log^2 n+k)$ time, i.e., the
shortest path length is reported in $O(\log^2 n)$ time and an
actual path is reported in additional $O(k)$ time, where $k$ is the
number of edges of the output path. In this paper, we build a new data structure
of size $O(n+h^2\cdot \log h \cdot 4^{\sqrt{\log h}})$ in
$O(n+h^2\cdot \log^{2} h \cdot 4^{\sqrt{\log h}})$ time that
answers each query in $O(\log n+k)$ time. Note that
$n+h^2\cdot \log^{2} h \cdot 4^{\sqrt{\log h}}=O(n+h^{2+\epsilon})$
for any constant $\epsilon>0$. (In contrast, for the Euclidean version of this two-point
query problem, the best known algorithm uses $O(n^{11})$ space to achieve an $O(\log n+k)$
query time.) In addition, we construct a data structure of
size $O(n+h^2\log^2 h)$ in $O(n+h^2\log^2 h)$ time that answers each query in
$O(\log n+\log^2h+k)$ time, and a data structure of size $O(nh\log h)$
in $O(nh\log h+h^2\log^2 h)$ time that answers each query in
$O(\log n\log h +k)$ time.  Further, we extend our techniques to the
weighted rectilinear version in which the ``obstacles" of $\calP$ are rectilinear regions
with ``weights" and allow $L_1$ paths to travel through them with weighted costs.
Previously,  a data structure of
size $O(n^2\log^2 n)$ was built in $O(n^2\log^2 n)$ time that
answers each query in $O(\log^2 n+k)$ time.
Our new algorithm answers each query in $O(\log n+k)$ time with a
data structure of size $O(n^2\cdot \log n\cdot 4^{\sqrt{\log n}})$ that
is built in $O(n^2\cdot \log^{2} n\cdot 4^{\sqrt{\log n}})$ time
(note that $n^2\cdot \log^{2} n\cdot 4^{\sqrt{\log n}}=
O(n^{2+\epsilon})$ for any constant $\epsilon>0$).
\end{abstract}

\section{Introduction}
\label{sec:intro}
Let $\calP$ be a set of $h$ pairwise-disjoint polygonal obstacles in the
plane with a total of $n$ vertices. We consider two-point shortest obstacle-avoiding
path queries for which the path lengths are measured in the $L_1$ metric.
The plane minus the interior of
the obstacles is called the {\em free space}. Our goal is to build a
data structure to quickly compute an $L_1$ shortest path
in the free space between any two query points $s$ and $t$.
Previously, Chen \etal \cite{ref:ChenSh00} constructed a data structure of
size $O(n^2\log n)$ in $O(n^2\log^2 n)$ time that computes the
length of the $L_1$ shortest \st\ path in $O(\log^2 n)$ time and an
actual path in additional $O(k)$ time,
where $k$ is the number of edges of the output path. Throughout this paper,
unless otherwise stated,
when we say that the query time of a data structure is $O(f(n,h))$ (which
may be a function of both $n$ and $h$), we
mean that the shortest path length can be reported in
$O(f(n,h))$ time and an actual path can be found in additional time
linear in the number of edges of the output path. Hence, the query
time of the data structure in \cite{ref:ChenSh00} is $O(\log^2 n)$.

In this paper, we build a new data structure of size
$O(n+h^2\cdot \log h \cdot 4^{\sqrt{\log h}})$ in
$O(n+h^2\cdot \log^{2} h \cdot 4^{\sqrt{\log h}})$ time, with $O(\log n)$
query time.  Note that $n+h^2\cdot \log^{2} h \cdot 4^{\sqrt{\log
h}}=O(n+h^{2+\epsilon})$ for any constant $\epsilon>0$. Hence, comparing with
the results in \cite{ref:ChenSh00}, we reduce the query
time by a logarithmic factor, and use less preprocessing time and
space when $h$ is small, e.g., $h=O(n^\delta)$ for any constant $\delta<1$.
In addition, we can also build a data structure of size $O(n+h^2\log^2h)$
in $O(n+h^2\log^2 h)$ time, with an $O(\log n+\log^2 h)$ query time,
and another data structure of size $O(nh\log h)$
in $O(nh\log h+h^2\log^2 h)$ time, with an $O(\log n\log h)$ query
time.

Further, we extend our techniques to the {\em weighted rectilinear
version} in which each ``obstacle" $P\in \calP$ is a region with a nonnegative
weight $w(P)$ and the edges of the obstacles in $\calP$ are all axis-parallel;
a path intersecting the interior of $P$ is charged a cost
depending on $w(P)$. For this problem, Chen \etal \cite{ref:ChenSh00}
constructed a data structure of
size $O(n^2\log^2 n)$ in $O(n^2\log^2 n)$ time that answers each
two-point shortest path query in $O(\log^2 n)$ time.
We build a new data structure of size $O(n^2\cdot \log n \cdot
4^{\sqrt{\log n}})$ in $O(n^2\cdot \log^{2} n \cdot 4^{\sqrt{\log
n}})$ time that answers each query in $O(\log n)$ time.
Note that $n^2\cdot \log^{2} n \cdot 4^{\sqrt{\log
n}}=O(n^{2+\epsilon})$ for any constant $\epsilon>0$.

\subsection{Related Work}

The problems of computing shortest paths among obstacles in the plane
have been studied extensively (e.g.,
\cite{ref:ChenCo13,ref:ChenSh00,ref:ChenA11ESA,ref:ChenCo12arXiv,ref:ChenCo13SoCG,ref:ChenL113STACS,ref:ElGindyOr94,ref:GuibasOp89,ref:GuibasLi87,ref:HershbergerA91,ref:HershbergerAn88,ref:HershbergerCo94,ref:HershbergerAn99,ref:HershbergerA13,ref:InkuluPl09,ref:InkuluA10,ref:KapoorAn97,ref:LeeEu84,ref:LeeSh91,ref:MitchellAn89,ref:MitchellA91,ref:MitchellL192,ref:MitchellSh96}).
There are three main types of such problems:
{\em finding a single shortest $s$-$t$ path} (both $s$ and $t$ are
given as part of the input and the goal is to find a single shortest \st\ path),
{\em single-source shortest path queries} ($s$ is given as part of the input
and the goal is to build a data structure to
answer shortest path queries for any query point $t$), and
{\em two-point shortest path queries} (as defined and considered in this
paper). The distance metrics can be the Euclidean (i.e., $L_2$) or $L_1$.
Refer to \cite{ref:MitchellGe00} for a comprehensive
survey on this topic.

For the simple polygon case, in which $\calP$ is a single simple polygon,
all three types of problems have been solved optimally
\cite{ref:GuibasOp89,ref:GuibasLi87,ref:HershbergerA91,ref:HershbergerCo94,ref:LeeEu84},
in both the Euclidean and $L_1$ metrics. Specifically, an $O(n)$-size data structure
can be built in $O(n)$ time that answers each
two-point Euclidean shortest path query in $O(\log n)$ time
\cite{ref:GuibasOp89,ref:HershbergerA91}.
Since in a simple polygon a Euclidean shortest path is also an $L_1$ shortest path
\cite{ref:HershbergerCo94}, the results in \cite{ref:GuibasOp89,ref:HershbergerA91}
hold for the $L_1$ metric as well.

The polygonal domain case (or ``a polygon with holes''), in which $\calP$ has $h$ obstacles as
defined above, is more difficult. For the Euclidean
metric, Hershberger and Suri \cite{ref:HershbergerAn99} built a {\em single source
shortest path map} of size $O(n\log n)$ in $O(n\log n)$ time that answers each
query in $O(\log n)$ time.
For the $L_1$ metric, Mitchell
\cite{ref:MitchellAn89,ref:MitchellL192} built an $O(n)$-size single source shortest
path map in $O(n\log n)$ time that answers each
query in $O(\log n)$ time. Later, Chen and Wang
\cite{ref:ChenA11ESA,ref:ChenCo12arXiv,ref:ChenL113STACS}
built an $L_1$ single source shortest path map of size $O(n)$ in $O(n+h\log h)$ time, with an
$O(\log n)$ query time, for a triangulated free
space (the current best triangulation algorithm takes
$O(n+h\log^{1+\epsilon}h)$ time for any constant $\epsilon>0$
\cite{ref:Bar-YehudaTr94}).
For two-point $L_1$ shortest path queries, Chen \etal \cite{ref:ChenSh00}
gave the previously best solution, as
mentioned above; for a special case where the obstacles are rectangles, ElGindy and Mitra \cite{ref:ElGindyOr94} gave an $O(n^2)$ size data structure that supports $O(\log n)$ time queries.
For two-point queries in the Euclidean metric,
Chiang and Mitchell \cite{ref:ChiangTw99}
constructed a data structure of size $O(n^{11})$ that answers each
query in $O(\log n)$ time, and alternatively, a data
structure of size $O(n^{10}\log n)$ with an $O(\log^2 n)$ query time; other data structures
with trade-off between preprocessing and query time were also given in
\cite{ref:ChiangTw99}. If the query points $s$ and $t$ are both restricted
to the boundaries of the obstacles of $\calP$, Bae and Okamato
\cite{ref:BaeQu12} built a data structure of size $O(n^5 {\rm poly}(\log n))$
that answers each query in $O(\log n)$ time, where
${\rm poly}(\log n)$ is a polylogarithmic factor.
Efficient algorithms were also given for the case when
the obstacles have curved boundaries
\cite{ref:ChenCo13,ref:ChenCo13SoCG,ref:ChewPl85,ref:HershbergerAn88,ref:HershbergerA13}.

For the weighted region case, in which
the ``obstacles" allow paths to pass through their
interior with weighted costs, Mitchell and Papadimitriou \cite{ref:MitchellTh91} gave an
algorithm that finds a weighted Euclidean shortest path in a time of $O(n^8)$
times a factor related to the precision of the problem instance.
For the weighted rectilinear case, Lee \etal \cite{ref:LeeSh91}
presented two algorithms for finding a weighted $L_1$
shortest path, and Chen \etal \cite{ref:ChenSh00} gave an improved algorithm with
$O(n\log ^{3/2}n)$ time and $O(n\log n)$ space. Chen
\etal \cite{ref:ChenSh00} also presented a data structure for
two-point weighted $L_1$ shortest path queries among
weighted rectilinear obstacles, as mentioned above.

\subsection{Our Approaches}

Our first main idea is to propose an enhanced graph model based on the
scheme in \cite{ref:ChenSh00,ref:ClarksonRe87,ref:ClarksonRe88}, to reduce the query
time from $O(\log^2 n)$ to $O(\log n)$. In \cite{ref:ChenSh00,ref:ClarksonRe87,ref:ClarksonRe88},
to build a graph, a total of $n$ vertical lines (called ``cut-lines'') are created recursively in $O(\log n)$ levels.
Then, each obstacle vertex $v$ is projected to $O(\log n)$ cut-lines (one cut-line per level) to create ``Steiner points'' if $v$ is horizontally visible to such cut-lines.
For any two query points $s$ and $t$, to report an $L_1$ shortest \st\ path, the algorithm in \cite{ref:ChenSh00} finds $O(\log n)$ Steiner points (called ``gateways'') on $O(\log n)$ cut-lines for each of $s$ and $t$, such that there must be a shortest \st\ path containing a gateway of $s$ and a gateway of $t$. Consequently, a shortest path is obtained in $O(\log^2 n)$ time using the $O(\log n)$ gateways of $s$ and $t$.

We propose an enhanced
graph $G_E$ by adding more Steiner points onto the cut-lines such that
we need only $O(\sqrt{\log n})$ gateways for any query points,
and consequently, computing the shortest path length takes $O(\log n)$ time. More
specifically, for each obstacle vertex, instead of projecting it to a single vertical
cut-line at each level, we project it to $O(2^{\sqrt{\log n}})$ cut-lines in
every $O(\sqrt{\log n})$
consecutive levels (thus creating $O(2^{\sqrt{\log n}})$ Steiner points); in fact,
these cut-lines form a binary tree structure of height $O(\sqrt{\log n})$ and they are carefully chosen to ensure that $O(\sqrt{\log n})$ gateways are sufficient for any query point.
Hence, the size
of the graph $G_E$ is $O(n\sqrt{\log n}2^{\sqrt{\log n}})$.

To improve the data structure construction so that its time and space bounds
depend linearly on $n$, we utilize the extended corridor
structure \cite{ref:ChenA11ESA,ref:ChenCo12arXiv,ref:ChenL113STACS},
which partitions the free space of $\calP$ into an ``ocean'' $\calM$, and
multiple ``bays'' and ``canals''. We build a graph $G_E(\calM)$
of size $O(h\sqrt{\log h} 2^{\sqrt{\log h}})$ on $\calM$ similar to $G_E$,
such that if both query points are in $\calM$, then the query can be answered in
$O(\log n)$ time. It remains to deal with the general case when at least one query
point is not in $\calM$. This is a major difficulty in our problem and our algorithm for
this case is another of our main contributions. Below, we use a bay as an example to
illustrate our main idea for this algorithm.

For two query points $s$ and $t$, suppose $s$ is in a bay $B$ and $t$ is outside $B$.
Since $B$ is a simple polygon, any shortest \st\ path must cross the ``gate'' $g$ of $B$,
which is a single edge shared by $B$ and $\calM$. We prove that there exists a shortest
\st\ path that must contain one of three special points $z(s)$,
$z_1(s)$, and $z_2(s)$, where $z(s)$ is in $B$ and the other
two points are on $g$ (and thus in $\calM$). For the case
when a shortest \st\ path contains either $z_1(s)$ or $z_2(s)$,
we can use the graph $G_E(\calM)$ to find such a shortest
path. For the other case, we build another graph $G_E(g)$ based on the
horizontal projections of the vertices of $G_E(\calM)$ on $g$, and use $G_E(g)$ to find
such a shortest path (along with a set of interesting observations) by a merge of
$G_E(g)$ and $G_E(\calM)$. Intuitively, $G_E(g)$ plays the role of connecting the shortest
path structure inside $B$ with those in $\calM$.

The case when a query point is in a canal can be handled similarly in spirit,
although it is more complicated because each canal has two gates.


The rest of the paper is organized as follows. In Section
\ref{sec:pre}, we introduce some notations and sketch the previous results
that will be needed by our algorithms. In Section
\ref{sec:newgraph}, we
propose our enhanced graph $G_E$ that helps reduce the query time to $O(\log n)$. In
Section \ref{sec:obstacle}, we further reduce
the preprocessing time and space by using the extended corridor
structure. In Section \ref{sec:weighted}, we
extend our techniques in Section \ref{sec:newgraph} to the weighted
rectilinear case.

Henceforth, unless otherwise stated, ``shortest paths'' always refer
to $L_1$ shortest paths and ``distances'' and ``lengths'' always refer to $L_1$
distances and lengths. To distinguish from graphs, the vertices/edges of
$\calP$ are always referred to as obstacle vertices/edges,
and graph vertices are referred to as ``nodes''. For
simplicity of discussion, we make a general position assumption that
no two obstacle vertices have the same $x$- or
$y$-coordinate except for the weighted rectilinear case.

\section{Preliminaries}
\label{sec:pre}

A path in the plane is {\em $x$-monotone} (resp., {\em $y$-monotone})
if its intersection with any
vertical (resp., horizontal) line is either empty or connected. A
path is {\em $xy$-monotone} if it is both {\em $x$-monotone} and {\em
$y$-monotone}. It is well-known that any $xy$-monotone path is an
$L_1$ shortest path.

A point $p$ is {\em visible} to another point $q$ if the line segment
$\overline{pq}$ entirely is in the free space.
A point $p$ is {\it horizontally} {\it visible} to a line $l$ if
there is a point $q$ on $l$ such that $\overline{pq}$
is horizontal and is in the free space.
For a line $l$ and a point $p$, the point $q\in l$ is the {\em horizontal}
{\it projection} of $p$ on $l$ if $\overline{pq}$ is horizontal, and we denote it by $p_h(l)=q$.
Let $\partial\calP$ denote the boundaries
of all obstacles in $\calP$. For a point $p$ in the free space of
$\calP$, if we shoot a horizontal ray from $p$ to the left, the first point on
$\partial\calP$ hit by the ray is called the {\em leftward projection} of $p$ on
$\partial\calP$, denoted by $p^l$; similarly, we define the
{\em rightward, upward}, and {\em downward} projections of $p$ on
$\partial\calP$, denoted by $p^r$, $p^u$, and $p^d$, respectively.

We sketch the graph in \cite{ref:ChenSh00}, denoted by $G_{old}$,
for answering two-point queries, which was originally proposed in
\cite{ref:ClarksonRe87,ref:ClarksonRe88} for computing a single
shortest path. To define $G_{old}$, two types of
{\em Steiner points} are specified, as follows.
For each obstacle vertex $p$, its
four projections on $\partial\calP$, i.e., $p^l,p^r,p^u$, and $p^d$, are {\em type-1}
Steiner points. Clearly, there are $O(n)$ type-1 Steiner points in total.

The {\em type-2 Steiner points} are on {\em cut-lines}. In
order to facilitate an explanation on our new graph model in Section
\ref{sec:newgraph}, we organize the cut-lines in a binary tree structure,
called the {\em cut-line tree} and denoted by $T(\calP)$. The tree $T(\calP)$ is defined as follows.
For each node $u$ of $T(\calP)$, a set $V(u)$ of obstacle vertices
and a cut-line $l(u)$ are associated with $u$, where $l(u)$ is a
vertical line through the median of the $x$-coordinates of the
obstacle vertices in $V(u)$. For the root $r$ of $T(\calP)$,
$V(r)$ is the set of all obstacle vertices of $\calP$. For the left
(resp., right) child $v$ of $u$, $V(v)$ consists of the obstacle
vertices of $V(u)$ on the left (resp., right) of $l(u)$.
Since the number of vertices of $\calP$ is $n$, the height of
$T(\calP)$ is $O(\log n)$. For every node $u$ of
$T(\calP)$, for each vertex $p\in V(u)$, if $p$ is horizontally visible to
$l(u)$, then the point $p_h(l(u))$, i.e., the horizontal projection of $p$ on
$l(u)$, is a type-2 Steiner point. Since each obstacle vertex defines
a type-2 Steiner point on at most one cut-line at each level of $T(\calP)$,
there are $O(n\log n)$ type-2 Steiner points.

The node set of $G_{old}$ consists of all obstacle vertices of $\calP$
and all Steiner points thus defined.

The edges of $G_{old}$ are defined as follows. First, for every obstacle vertex $p$,
there is an edge $\overline{pq}$ in $G_{old}$ for each $q\in
\{p^l,p^r,p^u,p^d\}$. Second, for every obstacle edge $e$ of $\calP$, $e$ may contain
multiple type-1 Steiner points, and these Steiner points and the two endpoints
of $e$ are the nodes of $G_{old}$ on $e$;
the segment connecting each pair of consecutive graph nodes on $e$ defines an
edge in $G_{old}$. Third, for each cut-line $l$, any two
consecutive type-2 Steiner points on $l$ define an edge in $G_{old}$ if
these two points are visible to each other.
Finally, for each obstacle vertex $p$, if $p$ defines a type-2
Steiner point $p'$ on a cut-line, then $\overline{pp'}$ defines an
edge in $G_{old}$.  Clearly, $G_{old}$ has $O(n\log n)$ nodes and
$O(n\log n)$ edges.

%

It was shown in \cite{ref:ClarksonRe87,ref:ClarksonRe88} that
$G_{old}$ contains a shortest path between any two obstacle vertices.
Chen \etal \cite{ref:ChenSh00} used $G_{old}$ to answer two-point
queries by ``inserting'' the query points $s$ and $t$ into $G_{old}$ so that shortest
\st\ paths are ``controlled'' by only $O(\log n)$ nodes of $G_{old}$,
called ``gateways''. The gateways of $s$ are defined as follows.
Intuitively, the gateways of $s$ are those nodes of $G_{old}$ that would be
adjacent to $s$ if we had built $G_{old}$ by treating $s$ as an obstacle vertex.
Let $V_g(s,G_{old})$ be the set of gateways of $s$, which we further
partition into two subsets $V^1_g(s,G_{old})$ and
$V^2_g(s,G_{old})$. We first define $V^1_g(s,G_{old})$, whose size is
$O(1)$. For each $q\in \{s^l,s^r,s^u,s^d\}$, let $v_1$
and $v_2$ be the two graph nodes adjacent to $q$ on the obstacle edge
containing $q$; then $v_1$ and $v_2$ are in $V^1_g(s,G_{old})$, and the
paths $\overline{sq}\cup\overline{qv_1}$ and
$\overline{sq}\cup\overline{qv_2}$ are the {\em gateway edges} from
$s$ to $v_1$ and $v_2$, respectively. Next, we define
$V^2_g(s,G_{old})$, recursively, on the cut-line tree $T(\calP)$.
Let $v$ be the root of $T(\calP)$.
Suppose $s$ is
horizontally visible to the cut-line $l(v)$. Let $q$ be the Steiner point
on $l(v)$ immediately above (resp., below) the projection point $s_h(l(v))$;
if $q$ is visible to $s_h(l(v))$, then $q$ is in $V^2_g(s,G_{old})$
and the path $\overline{ss_h(l(v))}\cup \overline{s_h(l(v))q}$ is the gateway
edge from $s$ to $q$. We also call $l(v)$ a {\em projection cut-line}
of $s$ if $s$ is horizontally visible to $l(v)$.
We proceed to the left (resp., right) child of $v$ in $T(\calP)$ if $s$
is to the left (resp., right) of $l(v)$.
We continue in this way until reaching a leaf of
$T(\calP)$. Therefore, $V^2_g(s,G_{old})$ contains $O(\log n)$
type-2 Steiner points on $O(\log n)$ projection cut-lines.

The above defines the gateway set $V_g(s,G_{old})$, and each
gateway $q \in V_g(s,G_{old})$ is associated with a gateway edge between $s$ and $q$. Henceforth,
when we say ``a path from $s$ contains a gateway $q$'', we
implicitly mean that the path contains the corresponding gateway edge as well.
The above also defines $O(\log n)$ projection cut-lines for $s$, which will
be used later in Section \ref{sec:newgraph}.
It was shown in \cite{ref:ChenSh00} that for any obstacle vertex $v$,
there is a shortest $s$-$v$ path using $G_{old}$ that contains a gateway of $s$.

Similarly, we define the gateway set $V_g(t,G_{old})$  for $t$.
Assume that there is a shortest \st\ path containing an obstacle
vertex. Then, there must be a shortest \st\ path that contains a gateway
$v_s\in V_g(s,G_{old})$, a gateway $v_t\in V_g(t,G_{old})$, and a
shortest path from $v_s$ to $v_t$ in the graph $G_{old}$ \cite{ref:ChenSh00}. Based on
this result, a {\em gateway graph} $G_g(s,t)$ is built for the query
on $s$ and $t$, as follows. The node
set of $G_g(s,t)$ is $\{s,t\}\cup V_g(s,G_{old})\cup V_g(t,G_{old})$.
Its edge set consists of all gateway edges and the edges $(v_s,v_t)$
for each $v_s\in V_g(s,G_{old})$ and each $v_t\in V_g(t,G_{old})$,
where the weight of  $(v_s,v_t)$ is the length of a shortest path
from $v_s$ to $v_t$ in $G_{old}$. Hence, $G_g(s,t)$ has $O(\log n)$ nodes and
$O(\log^2 n)$ edges, and if we know the weights of all edges $(v_s,v_t)$,
then a shortest \st\ path in $G_g(s,t)$ can be found in
$O(\log^2 n)$ time. To obtain the weights of all edges $(v_s,v_t)$, we
compute a single source shortest path tree in $G_{old}$ from each node of $G_{old}$
in the preprocessing. Then, the weight of each such edge $(v_s,v_t)$
is obtained in $O(1)$ time. Further, suppose we find a shortest
\st\ path in $G_g(s,t)$ that contains a gateway $v_s\in
V_g(s,G_{old})$ and a gateway $v_t\in V_g(t,G_{old})$; then we can
report an actual shortest \st\ path in time linear to the
number of edges of the output path by using the shortest path tree from $v_s$
in $G_{old}$ (which has been computed in the preprocessing).

As discussed in \cite{ref:ChenSh00},
it is possible that no shortest \st\ path contains any
obstacle vertex.
For example, consider
a projection point $s^r$ of $s$ and a projection point $t^d$ of $t$. If
$\overline{ss^r}$ intersects $\overline{tt^d}$, say at a point $q$, then
$\overline{sq}\cup \overline{qt}$ is a shortest \st\ path; otherwise, if $s^r$ and
$t^d$ are both on the same obstacle edge, then $\overline{ss^r}\cup
\overline{s^rt^d}\cup\overline{t^dt}$ is a shortest \st\ path. We call
such shortest \st\ paths {\em trivial shortest paths}. Similarly, trivial shortest
\st\ paths can also be defined by other projection points in $\{s^l,s^r,s^u,s^d\}$
and $\{t^l,t^r,t^u,t^d\}$.
It was shown in \cite{ref:ChenSh00} that if there is no trivial
shortest \st\ path, then there exists a shortest \st\ path that
contains an obstacle vertex. If we know $\{s^l,s^r,s^u,s^d\}$
and $\{t^l,t^r,t^u,t^d\}$, then we can determine whether there exists
a trivial shortest \st\ path in $O(1)$ time. For any query points $s$ and $t$,
their projection points can be computed easily in $O(\log n)$ time by using
the horizontal and vertical visibility decompositions of $\calP$, as shown in
\cite{ref:ChenSh00}.

\section{Reducing the Query Time Based on an Enhanced Graph}
\label{sec:newgraph}

In this section, we propose an ``enhanced graph'' $G_E$ that allows us to reduce
the query time to $O(\log n)$, although
$G_E$ has a larger size than $G_{old}$. We
first define $G_E$, and then show how to answer two-point queries by using $G_E$.

\subsection{The Enhanced Graph $G_E$}

On the nodes of $G_E$, first, every node of $G_{old}$ is also a node
in $G_E$. In addition, $G_E$
contains the following {\em type-3} Steiner points as nodes.
To define the type-3 Steiner points, we introduce the concepts of ``levels'' and
``super-levels'' on  the cut-line tree $T(\calP)$ defined in Section
\ref{sec:pre}.
$T(\calP)$ has $O(\log n)$ levels. We
define the level numbers recursively: The root $v$
is at the first level, and its {\em level number} is denoted by $ln(v)$ $=$ $1$;
for any node $v$ of $T(\calP)$, if $u$ is a child of $v$, then $ln(u)=ln(v)+1$.
We further partition the $O(\log n)$ levels of $T(\calP)$ into $O(\sqrt{\log
n})$ {\em super-levels}: For any $i$, $1\leq i\leq O(\sqrt{\log
n})$, the $i$-th super-level contains the levels from
$(i-1)\cdot \sqrt{\log n}+1$ to $i\cdot \sqrt{\log n}$.

\begin{figure}[t]
\begin{minipage}[t]{\linewidth}
\begin{center}
\includegraphics[totalheight=2.0in]{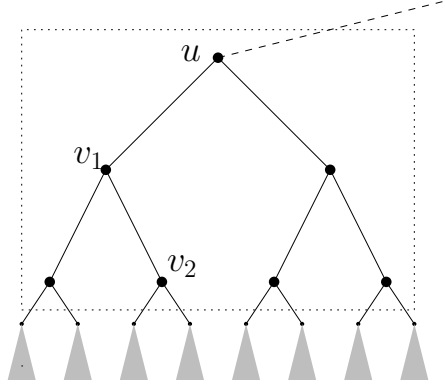}
\caption{\footnotesize Illustrating $T_u(\calP)$, i.e., the portion of the tree in the dotted box, where $\sqrt{\log n}=3$. }
\label{fig:type3}
\end{center}
\end{minipage}
\vspace*{-0.15in}
\end{figure}

Consider the $i$-th super-level.
Let $u$ be any node at the highest level (i.e., the level with the smallest level number) of this super-level. Let
$T_u(\calP)$ denote the subtree of $T(\calP)$ rooted at $u$ without including
any node outside the $i$-th super-level (e.g., see Fig.~\ref{fig:type3} and its corresponding cut-lines and level numbers in Fig.~\ref{fig:levelno}). Since $T_u(\calP)$ has
$O(\sqrt{\log n})$ levels, $T_u(\calP)$ has
$O(2^{\sqrt{\log n}})$ nodes. Recall that $u$ is associated with a
subset $V(u)$ of obstacle vertices and a vertical cut-line $l(u)$, and for any vertex $p$
in $V(u)$, if $p$ is horizontally visible to
$l(u)$, then its projection point $p_h(l(u))$ is a type-2 Steiner
point. Each point $p\in V(u)$ defines the following type-3 Steiner points. For
each node $v$ in $T_u(\calP)$, if $p$ is horizontally visible to
$l(v)$, then its projection point $p_h(l(v))$ is a type-3 Steiner
point (e.g., see Fig.~\ref{fig:levelno}; note that if $p\in V(v)$, then the Steiner point is also a type-2 Steiner point).
Hence, $p$ defines $O(2^{\sqrt{\log n}})$ type-3 Steiner
points in the $i$-th super-level of $T(\calP)$. Let $S(p)$ be the set of all
type-2 and type-3 Steiner points on the cut-lines of the subtree
$T_u(\calP)$ induced by $p$, and let $S(p)$ also contain $p$. In the order of the
points in $S(p)$ from left to right, we put an edge in
$G_E$ connecting every two consecutive points in $S(p)$ (e.g., see Fig.~\ref{fig:levelno}).
Since the total number of obstacle vertices in $V(u)$ for all nodes $u$
at the same level of $T(\calP)$ is $n$, the number of type-3
Steiner points thus defined in each super-level is $O(n2^{\sqrt{\log
n}})$, and the total number of type-3 Steiner points on all cut-lines
in $T(\calP)$ is $O(n\sqrt{\log n}2^{\sqrt{\log n}})$. The number of edges
thus added to $G_E$ is also $O(n\sqrt{\log n}2^{\sqrt{\log n}})$.

Hence, the total number of nodes in $G_E$ is $O(n\sqrt{\log
n}2^{\sqrt{\log n}})$, which is dominated by the number of type-3
Steiner points.
We have also defined above some edges in $G_E$. The rest of
edges in $G_E$ are defined similarly as in $G_{old}$. Specifically,
first, as in $G_{old}$, for every obstacle vertex $p$,
there is an edge $\overline{pq}$ in $G_E$ for each $q\in
\{p^l,p^r,p^u,p^d\}$. Second, as in $G_{old}$,
for each obstacle edge $e$, $e$ may contain
multiple type-1 Steiner points;
the segment connecting each pair of consecutive graph nodes on $e$ defines an
edge in $G_E$. Third, for each cut-line $l$, every pair of
consecutive Steiner points (type-2 or type-3)
on $l$ defines an edge in $G_E$ if
these two points are visible to each other.
Clearly, the total number of edges in $G_E$ is $O(n\sqrt{\log
n}2^{\sqrt{\log n}})$.

\begin{figure}[t]
\begin{minipage}[t]{\linewidth}
\begin{center}
\includegraphics[totalheight=2.0in]{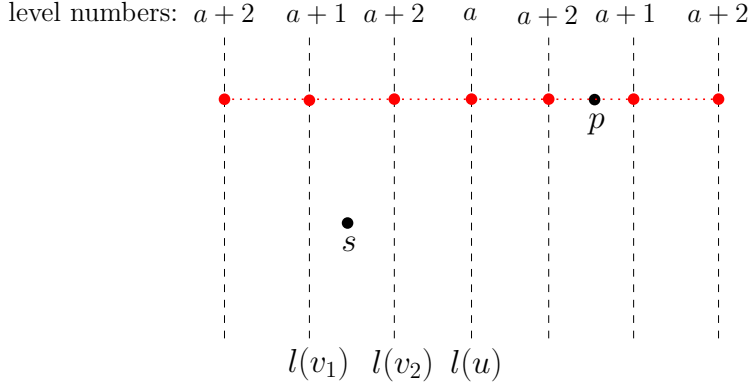}
\caption{\footnotesize Illustrating the cut-lines and level numbers of the subtree $T_u(\calP)$ in Fig.~\ref{fig:type3}, where $a$ is the level number $ln(u)$ of the node $u$. $p$ is an obstacle vertex. If $p$ is visible to all cut-lines, then the red points are type-2 and type-3 Steiner points defined by $p$ and the (red) dotted segments are the corresponding graph edges. }
\label{fig:levelno}
\end{center}
\end{minipage}
\vspace*{-0.15in}
\end{figure}

This finishes the definition of our enhanced graph $G_E$, which has
$O(n\sqrt{\log n}2^{\sqrt{\log n}})$ nodes and $O(n\sqrt{\log n}2^{\sqrt{\log n}})$
edges. The following lemma gives an algorithm for computing $G_E$.

\begin{lemma}\label{lem:10}
The enhanced graph $G_E$ can be constructed
in $O(n\log^{3/2}n2^{\sqrt{\log n}})$ time.
\end{lemma}
\begin{proof}
First of all, all type-1 Steiner points are
computed easily in $O(n\log n)$ time, e.g., by using the vertical and horizontal visibility decompositions of $\calP$. The edges of $G_E$ connecting the obstacle vertices and their
corresponding type-1 Steiner points can also be computed. For each obstacle
edge $e$, we sort all graph nodes on $e$ and then
compute the edges of $G_E$ connecting the consecutive nodes on $e$. Since
there are $O(n)$ type-1 Steiner points, computing these
edges takes $O(n\log n)$ time.

Next, we compute both the type-2 and type-3 Steiner points and their adjacent
edges. For this, we need to use the two projection
points $p^l$ and $p^r$ for each obstacle vertex $p$ of $\calP$, which have been computed as type-1 Steiner points.
Consider an obstacle vertex $p$ in $V(u)$ for a
node $u$ at the highest level of a super-level. For each node $v$ in
$T_u(\calP)$, we need to determine whether $p$ is horizontally
visible to $l(v)$, which can be done in $O(1)$ time since $p^l$
and $p^r$ are already known. We also need to have a sorted order of all
cut-lines in $T_u(\calP)$ from left to right, and this ordered list
can be obtained by an in-order traversal of $T_u(\calP)$ in linear
time. Therefore, the edges of $G_E$ connecting the Steiner points defined by $p$ on
consecutive cut-lines in this super-level can be computed in time linear to the number of nodes in $T_u(\calP)$.
Since there are $O(n\sqrt{\log n}2^{\sqrt{\log n}})$ type-2 and type-3
Steiner points, computing all such edges takes $O(n\sqrt{\log
n}2^{\sqrt{\log n}})$ time.

It remains to compute the graph edges on all cut-lines connecting consecutive Steiner
points (if they are visible to each other).
This step is done in $O(n\log^{3/2}n2^{\sqrt{\log n}})$ time by a
sweeping algorithm, as follows. For each cut-line $l$, we sort
the Steiner points on $l$ by their $y$-coordinates, and
determine whether every two consecutive Steiner points on $l$ are visible to
each other. For this, we sweep a vertical line $L$ from left to right.
During the sweeping, we use a balanced binary search tree $T$ to maintain
the maximal intervals of $L$ that are in the free space of $\calP$ (there are $O(n)$ such intervals).
At each obstacle vertex, we
update $T$ in $O(\log n)$ time. At each (vertical) cut-line $l$,
for every two consecutive Steiner points, we determine whether they
are visible to each other in $O(\log n)$ time by checking whether they
are in the same maximal interval maintained by $T$. Since
there are $O(n\sqrt{\log n}2^{\sqrt{\log n}})$ pairs of consecutive
Steiner points on all cut-lines, computing all edges of $G_E$ on the cut-lines takes
totally $O(n\log^{3/2}n2^{\sqrt{\log n}})$ time. Another
approach for computing these edges in $O(n\log^{3/2}n2^{\sqrt{\log
n}})$ time is to perform vertical ray-shootings from all Steiner
points (we omit the details).

In summary, the enhanced graph $G_E$ can be computed in
$O(n\log^{3/2}n2^{\sqrt{\log n}})$ time.
\end{proof}

\subsection{Reducing the Query Time}

We use the enhanced graph $G_E$ to reduce the query time to $O(\log n)$.
Consider two query points $s$ and $t$. One of our key ideas is: We
define a new set of gateways for $s$, denoted by $V_g(s,G_E)$, which
contains $O(\sqrt{\log n})$ nodes of $G_E$, such that
for any obstacle vertex $p$ of $\calP$, there
exists a shortest path from $s$ to $p$ through a gateway of
$V_g(s,G_E)$. The set $V_g(s,G_E)$ can be divided into two subsets $V^1_g(s,G_E)$ and
$V^2_g(s,G_E)$, where $V^1_g(s,G_E)$ (of size $O(1)$) is exactly the same as
$V^1_g(s,G_{old})$ defined on $G_{old}$ in Section \ref{sec:pre}. Below,
we define the subset $V^2_g(s,G_E)$.

Recall that $s$ has $O(\log n)$ projection cut-lines, as defined in
Section \ref{sec:pre}. By definition, $s$ is horizontally visible to all its
projection cut-lines. Since $G_E$ has more Steiner points than $G_{odd}$, the intuition is that we do not have to include gateways in each projection cut-line of $s$. More specifically, we only need to include gateways in two projection cut-lines in each super-level (one to the left of $s$ and the other to the right of $s$). The details are given below.

 We define the {\em relevant projection
cut-lines} of $s$, as follows. Let $S$ be the set of projection cut-lines of
$s$ to the right of $s$.  Consider a cut-line $l\in S$
and suppose $l$ is associated with a node $u$ in the $i$-th super-level
of the cut-line tree $T(\calP)$ for some $i$. Then $l$ is a {\em relevant
projection cut-line} of $s$ if $ln(u)> ln(v)$ (i.e., their level
numbers) for every node $v$ with $v\neq u$ in the
$i$-th super-level of $T(\calP)$
such that the cut-line $l(v)$ of $v$ is also in $S$. In other words,
$l(u)$ is a relevant projection cut-line of $s$ if $u$ has the largest distance
in $T(\calP)$ from the root
among all nodes $v$ in the $i$-th super-level of $T(\calP)$ whose
cut-lines $l(v)$ are in $S$. For example, in Fig.~\ref{fig:type3} and Fig.~\ref{fig:levelno}, suppose $s$ is between the cut-lines $l(v_1)$ and $l(v_2)$ and both $l(u)$ and $l(v_2)$ are horizontally visible to $s$; then among the cut-lines of all nodes in $T_u(\calP)$, only $l(v_2)$ and $l(u)$ are in $S$, but only $l(v_2)$ is the relevant projection cut-line of $s$.  The relevant projection cut-lines
of $s$ to the left of $s$ are defined similarly.  Since $s$ has $O(\log
n)$ projection cut-lines and any two of them are at different levels
of $T(\calP)$, the number of relevant projection cut-lines of $s$ is
$O(\sqrt{\log n})$, i.e., at most two from each super-level of
$T(\calP)$ (one to the left of $s$ and the other to the right of $s$). For each relevant projection cut-line $l$ of $s$, the Steiner point $p$
(if any) immediately above (resp., below) the projection point $s_h(l)$ of $s$ on $l$ is in
$V^2_g(s,G_E)$ if $p$ is visible to $s_h(l)$. Thus,
$|V^2_g(s,G_E)|=O(\sqrt{\log n})$.

$V_g(s,G_E)$ thus defined is of size
$O(\sqrt{\log n})$. We also define the gateway edge for
each gateway of $V_g(s,G_E)$ and $s$ in the same way as in Section \ref{sec:pre}.
Below, when we say a shortest path from $s$ containing a gateway, we
mean the path containing the corresponding gateway edge as well.

\begin{lemma}\label{lem:20}
For any obstacle vertex $p$ of $\calP$, there exists a shortest path
from $s$ to $p$ using $G_E$ that contains a gateway of $s$ in $V_g(s,G_E)$.
\end{lemma}
\begin{proof}
Recall that $V_g(s,G_{old})$ is the gateway set of $s$ defined on
$G_{old}$ in Section \ref{sec:pre}, and by \cite{ref:ChenSh00}, there exists a
shortest path $\pi(s,p)$ from $s$ to $p$ using $G_{old}$
that contains a point $a\in V_g(s,G_{old})$.

By the definition of $G_E$, if any edge $e$ of $G_{old}$ connecting two nodes $u$ and $v$ is not an edge of $G_E$, then $e$
can be viewed as being ``divided" into many edges in $G_E$ such that the concatenation
of these edges is a path from $u$ to $v$ in $G_E$ with the same length
as $e$. Hence, $\pi(s,p)$ is still a shortest path
along $G_E$.
For any point $a \in V_g(s,G_{old})$ that is on a shortest $s$-$p$ path, we call
it a {\it via point}.
If any via point $a$ is in $V^1_g(s,G_{old})$, then $a$ is in $V_g(s,G_E)$
since $V^1_g(s,G_E)=V^1_g(s,G_{old})$, and we are done.
Otherwise, all via points must be in $V^2_g(s,G_{old})$. If any such via point $a\in V^2_g(s,G_{old})$ is also in
$V^2_g(s,G_E)$, then we are done as well. It remains to prove for the
case that for every via point $a$, $a\in V^2_g(s,G_{old})$ and $a\not\in
V^2_g(s,G_E)$ hold. Recall that every node of $G_{old}$, including each via point
$a$, is also a node of $G_E$.
Below, we find an $xy$-monotone path from $s$ to such a via point $a$ along $G_E$ that
contains a gateway $b \in V^2_g(s,G_E)$. Since any $xy$-monotone path
is a shortest path, this gives a shortest $s$-$p$ path (through $a$)
containing a gateway $b$ of $s$ in $V_g(s,G_E)$, thus proving the lemma.

Without loss of generality, we assume that $a$ is to the right of $s$ and
above $s$ (i.e., $a$ is to the northeast of $s$, see
Fig.~\ref{fig:gatewayproof}). Suppose $a$ is on the cut-line $l(v)$ of a node $v$ in
the $i$-th super-level of $T(\calP)$.
If $l(v)$ is a relevant cut-line of $s$, then there must be a gateway $b$ of $s$ in $V^2_g(s,G_E)$
lying in the vertical segment $\overline{s_h(l(v))a}$ on $l(v)$ (possibly $b=a$), and thus
we are done.  Otherwise, $l(v)$ is not a relevant cut-line of $s$, and
there exists a relevant cut-line $l(v')$ of $s$
to the right of $s$ such that $v'$ is in the $i$-th super-level of $T(\calP)$
and $ln(v')>ln(v)$.
Next, we show that the sought gateway $b$ lies on $l(v')$.

It was shown in \cite{ref:ChenSh00} (Lemma 3.4) that the level numbers of the
projection cut-lines of $s$ to the right of $s$,
in the left-to-right order, are decreasing.
This observation can also be seen easily by considering the projection cut-lines of $T(\calP)$ in a top-down manner. Hence,
$l(v')$ is to the left of $l(v)$ (see Fig.~\ref{fig:gatewayproof}). Let $q$ be the obstacle vertex that
defines the Steiner point $a$ on $l(v)$. By our definition of
Steiner points, $q$ must be in $V(u)$ for the
node $u$ that is the highest ancestor of $v$ (and $v'$) in the $i$-th
super-level. Therefore, if $q$ is horizontally visible to
$l(v')$, then $q$ also defines a Steiner point on $l(v')$. We now show
that $q$ is horizontally visible to $l(v')$, and for this,
it suffices to prove
that $a$ is horizontally visible to $l(v')$ since $q$ is horizontally visible to $a$.
Because $a \in V^2_g(s,G_{old})$ and no via point is in $V^1_g(s,G_{old})$,
it was shown in \cite{ref:ChenSh00} that
$a$ must be horizontally visible to the vertical line through $s$. Since
$l(v')$ is between $s$ and $a\in l(v)$, $a$ is also horizontally
visible to $l(v')$.

\begin{figure}[t]
\begin{minipage}[t]{\linewidth}
\begin{center}
\includegraphics[totalheight=1.2in]{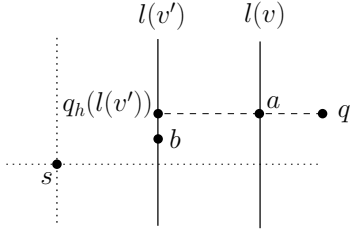}
\caption{\footnotesize Illustrating the proof of Lemma \ref{lem:20}:
$q$ is the obstacle vertex that defines the Steiner point $a$; $l(v')$
is between $s$ and $l(v)$.}
\label{fig:gatewayproof}
\end{center}
\end{minipage}
\vspace*{-0.15in}
\end{figure}

Thus, $q$ defines a Steiner point on $l(v')$, i.e., the point
$q_h(l(v'))$ (see Fig.~\ref{fig:gatewayproof}).
By the definition of $V^2_g(s,G_E)$, the
lowest Steiner point $b$ on $l(v')$ above $s$ must be a gateway in
$V^2_g(s,G_E)$. Note that $b$ may or may not be $q_h(l(v'))$, but $b$
cannot be higher than $q_h(l(v'))$. Thus, the concatenation of the
gateway edge from $s$ to $b$, $\overline{bq_h(l(v'))}$, and
$\overline{q_h(l(v'))a}$, which is an $xy$-monotone path from $s$ to
$a$ using $G_E$, contains the gateway $b$ of $V^2_g(s,G_E)$.
The lemma thus follows.
\end{proof}

Similarly, we define the gateway set $V_g(t,G_E)$ for $t$ in $G_E$.
The similar result for $t$ as Lemma \ref{lem:20} for $s$ also holds. Thus, we have
the following corollary.

\begin{corollary}\label{cor:10}
If there exists a shortest \st\ path through an obstacle vertex of
$\calP$, then there exists a shortest \st\ path through a gateway of
$s$ in $V_g(s,G_E)$ and a gateway of $t$ in $V_g(t,G_E)$.
\end{corollary}

Next, we give an algorithm for computing the two gateway sets $V_g(s,G_E)$ and
$V_g(t,G_E)$.

\begin{lemma}\label{lem:30}
With a preprocessing of $O(n\log^{3/2} n2^{\sqrt{\log n}})$ time
and $O(n\sqrt{\log n}2^{\sqrt{\log n}})$ space,
we can compute the gateway sets $V_g(s,G_E)$ and
$V_g(t,G_E)$  in $O(\log n)$ time for any query points $s$ and $t$.
\end{lemma}
\begin{proof}
We only discuss the case for computing $V_g(s,G_E)$ since $V_g(t,G_E)$ can be
computed similarly.

To compute $V^1_g(s,G_E)$, it suffices to determine the four
projection points $\{s^l,s^r,s^u,s^d\}$ of $s$ on $\partial\calP$, which can be computed in
$O(\log n)$ time by using the horizontal and vertical
visibility decompositions of $\calP$. These two visibility decompositions can be built in
$O(n\log n)$ time by standard sweeping algorithms.
After that, we also need to build a point location
data structure \cite{ref:EdelsbrunnerOp86,ref:KirkpatrickOp83} on each of the
two decompositions in additional $O(n)$ time.

To compute $V^2_g(s,G_E)$, it might be possible to modify the approach in
\cite{ref:ChenSh00}. However, to explain the approach in
\cite{ref:ChenSh00}, we may have to review a number of observations given
in \cite{ref:ChenSh00}. To avoid a tedious discussion, we propose
the following algorithm that is simple.

We first obtain the set $S$ of all relevant projection cut-lines of $s$. This can be
done in $O(\log n)$ time by following the cut-line tree $T(\calP)$
from the root and using $s^l$ and $s^r$ to determine the horizontal visibility of $s$.
Note that the cut-lines of $S$ are at some
nodes on a path from the root to a leaf. To obtain
$V^2_g(s,G_E)$, for each cut-line $l\in S$, we need to: (1) find the
Steiner point $p$ on $l$ immediately above (resp., below) $s_h(l)$, and (2)
determine whether $p$ is visible to $s_h(l)$.

Consider a cut-line $l\in S$. Let $v_1(l)$ and $v_2(l)$ be the two
gateways of $V^2_g(s,G_E)$ on $l$ (if any) such that $v_1(l)$ is above
$v_2(l)$. That is, $v_1(l)$ (resp., $v_2(l)$) is the Steiner
point on $l$ immediately above (resp., below) $s_h(l)$ and visible to $s_h(l)$.
If we maintain a sorted list of all Steiner
points on $l$, then $v_1(l)$ and $v_2(l)$ can be found by binary
search  on the sorted list. However, there are two issues with this
approach. First, if we do binary search on each cut-line of $S$, since
$|S|=O(\sqrt{\log n})$, it takes $O(\log^{3/2}n)$ time
on all cut-lines of $S$. Second, even if we find $v_1(l)$ and
$v_2(l)$, we still need to check whether $s_h(l)$ is visible to them. To
resolve these two issues, we take the following approach.


For every Steiner point $p$ on the cut-line $l$, suppose we associate with $p$ its upward and
downward projection points $p_u$ and $p_d$ on $\partial\calP$. Then once we
find the Steiner point $q$ on $l$ immediately above (resp., below) $s_h(l)$, we can
determine easily whether $q$ is visible to $s_h(l)$ using $q_u$ and $q_d$; if
$q$ is visible to $s_h(l)$, then $v_1(l)=q$ (resp., $v_2(l)=q$), or else $v_1(l)$
(resp., $v_2(l)$) does not exist.
%
For any Steiner point $p$ on $l$, $p^l$ and $p^r$ can be found in $O(\log n)$ time by using the vertical visibility decomposition of $\calP$.
Since there are $O(n\sqrt{\log n}2^{\sqrt{\log n}})$ Steiner points $p$ on all cut-lines of $T(\calP)$,
their projection points $p^u$ and $p^d$
can be computed in totally $O(n\log^{3/2} n2^{\sqrt{\log n}})$ time.

Next, for each cut-line $l$, we sort all Steiner points on $l$. With this, one can compute all gateways of
$V_g^2(s,G_E)$ in $O(\log^2 n)$ time by doing binary search on each relevant projection cut-line of $s$.
To reduce the query time to $O(\log n)$, we make use of the fact that all relevant projection cut-lines of $s$ are at the nodes on a path of $T(\calP)$ from the root to a leaf.
We build a fractional cascading structure \cite{ref:ChazelleFr86}
on the sorted lists of Steiner points on all cut-lines along $T(\calP)$, such that the searches on all cut-lines at the nodes on any path of $T(\calP)$ from the root to a leaf take $O(\log n)$ time.
Hence, all gateways of $V_g^2(s,G_E)$ can be computed in $O(\log n)$ time. Since the total number of Steiner points
in the sorted lists of all cut-lines of $T(\calP)$ is $O(n\sqrt{\log n}2^{\sqrt{\log n}})$, the fractional cascading structure can be built in $O(n\sqrt{\log n}2^{\sqrt{\log n}})$ space and  $O(n\log^{3/2} n2^{\sqrt{\log n}})$ time.
The lemma thus follows.
\end{proof}

\begin{theorem}\label{theo:10}
We can build a data structure of size $O(n^2\cdot \log n \cdot
4^{\sqrt{\log n}})$ in $O(n^2\cdot \log^2 n \cdot
4^{\sqrt{\log n}})$ time that can
answer each two-point $L_1$ shortest path query in $O(\log n)$ time (i.e., for
any two query points $s$ and $t$, the length of a shortest \st\ path can be found in $O(\log n)$
time and an actual path can be reported in additional time linear to the number of edges of the output path).
\end{theorem}
\begin{proof}
In the preprocessing, we first build the graph $G_E$. Then, for each node $v$ of $G_E$, we compute a shortest path tree in $G_E$ from $v$. We also maintain a shortest path length table such that
for any two nodes $u$ and $v$, the shortest $u$-$v$ path length in $G_E$ can be obtained in $O(1)$ time. Since $G_E$ is of a size $O(n\sqrt{\log n}2^{\sqrt{\log n}})$, computing and
maintaining all these shortest
path trees in $G_E$ take $O(n^2\log n 4^{\sqrt{\log n}})$ space and
$O(n^2 \log^2 n
4^{\sqrt{\log n}})$ time.
%
%
%
We also do the preprocessing for Lemma \ref{lem:30}.

Given any two query points $s$ and $t$, we first check whether there
is a trivial shortest \st\ path, as discussed in Section
\ref{sec:pre}, in $O(\log n)$ time by using the algorithm
in \cite{ref:ChenSh00} (with an $O(n\log n)$ time preprocessing). If
there is a trivial shortest \st\ path, then we are done.
Otherwise, there must be a shortest \st\ path that contains an
obstacle vertex of $\calP$. Then, we first compute the gateway sets
$V_g(s,G_E)$ and $V_g(t,G_E)$ in $O(\log n)$ time by Lemma
\ref{lem:30}. Finally, we determine the shortest $s$-$t$ path length
by using the gateway graph as discussed in Section \ref{sec:pre}, in
$O(\log n)$ time, since there are $O(\sqrt{\log n})$ gateways and thus
the gateway graph has $O(\sqrt{\log n})$ nodes and $O(\log n)$ edges.

We can also report an actual shortest \st\ path in additional time linear to the number of edges of the output path by using the shortest path trees of $G_E$.
%
This proves the theorem.
\end{proof}


\section{Reducing the Time and Space Bounds of the Preprocessing}
\label{sec:obstacle}

In this section, we improve the preprocessing
 in Theorem \ref{theo:10} to $O(n+h^2\cdot \log h \cdot
4^{\sqrt{\log h}})$ space and $O(n+h^2\cdot \log^2 h \cdot 4^{\sqrt{\log h}})$ time,
while maintaining the $O(\log n)$ query time. For this, we shall make use
of the extended corridor data structure
\cite{ref:ChenA11ESA,ref:ChenCo12arXiv,ref:ChenL113STACS,ref:KapoorAn97}, and more importantly, explore
a number of new observations, which may be interesting in their own right.

The corridor structure has been used to solve shortest path problems (e.g., \cite{ref:InkuluPl09,ref:KapoorEf88,ref:KapoorAn97}), and new concepts like ``ocean'', ``bays'', and ``canals'' have been introduced \cite{ref:ChenA11ESA,ref:ChenCo12arXiv,ref:ChenCo12ICALP,ref:ChenCo13SoCG,ref:ChenL113STACS,ref:ChenVi13WADS}, which we refer to as the ``extended corridor structure''.
This structure is a subdivision of the free space on which algorithms for specific problems
rely.  While the extended corridor structure itself is relatively simple, the main difficulty is to design
efficient algorithms to exploit it. In some sense, the role played by the extended corridor structure is similar to that of triangulations for many geometric algorithms.
We briefly review the extended corridor structure in Section \ref{subsec:extended}, since our
presentation uses many notations introduced in it.

\subsection{The Extended Corridor Structure}
\label{subsec:extended}

For simplicity of discussion, we assume that the obstacles of $\calP$ are all contained in a
rectangle $\calR$.
Let $\calF$ denote the free space in $\calR$, and
$\Tri(\calF)$ denote a triangulation of $\calF$ (see Fig.~\ref{fig:triangulation}). The line segments of $\Tri(\calF)$ that are not obstacle edges are referred to as {\em diagonals}.

Let $G(\calF)$ denote the dual graph of $\Tri(\calF)$,
i.e., each node of $G(\calF)$ corresponds to a
triangle of $\Tri(\calF)$ and each edge connects two nodes
corresponding to two triangles sharing a diagonal of $\Tri(\calF)$.
Based on $G(\calF)$, we
compute a planar 3-regular graph, denoted by $G^3$ (the degree of every node in $G^3$ is three),
possibly with loops and multi-edges,
as follows. First, we remove each degree-one node from $G(\calF)$
along with its incident edge; repeat this process until no
degree-one node remains in the graph. Second, remove every degree-two node from
$G(\calF)$ and replace its two incident edges by a single edge;
repeat this process until no degree-two node remains. The
resulted graph is $G^3$ (see Fig.~\ref{fig:triangulation}), which has
$O(h)$ faces, nodes, and  edges \cite{ref:KapoorAn97}. Every node of
$G^3$ corresponds to a triangle in $\Tri(\calF)$, called a
{\em junction triangle} (see Fig.~\ref{fig:triangulation}).
The removal of the nodes for all junction triangles from $G^3$ results in $O(h)$
{\em corridors}, each of which corresponds to an edge of $G^3$.

The boundary of each corridor $C$ consists of four parts (see
Fig.~\ref{fig:corridor}): (1) A boundary portion of an obstacle
$P_i\in \calP$, from a point $a$ to a point $b$; (2) a diagonal of a
junction triangle from $b$ to a point $e$ on an obstacle
$P_j\in \calP$ ($P_i=P_j$ is possible); (3) a boundary portion of
the obstacle $P_j$ from $e$ to a point $f$; (4) a diagonal of a
junction triangle from $f$ to $a$.
The corridor $C$ is a simple polygon.
Let $\pi(a,b)$ (resp., $\pi(e,f)$) be the Euclidean shortest path from $a$ to $b$
(resp., $e$ to $f$) in $C$. The region $H_C$ bounded by
$\pi(a,b), \pi(e,f)$, $\overline{be}$, and
$\overline{fa}$ is called an {\em hourglass}, which is {\em open} if
$\pi(a,b)\cap \pi(e,f)=\emptyset$ and {\em closed} otherwise (see
Fig.~\ref{fig:corridor}). If $H_C$ is open, then both $\pi(a,b)$ and
$\pi(e,f)$ are convex chains and are called the {\em sides} of
$H_C$; otherwise, $H_C$ consists of two ``funnels" and a path
$\pi_C=\pi(a,b)\cap \pi(e,f)$ joining the two apices of the two
funnels, and $\pi_C$ is called the {\em corridor path} of $C$.
The two funnel apices (e.g., $x$ and $y$ in Fig.~\ref{fig:corridor})
are called {\em corridor path terminals}.
Each side of a funnel is also a convex chain.

\begin{figure}[t]
\begin{minipage}[t]{0.47\linewidth}
\begin{center}
\includegraphics[totalheight=1.2in]{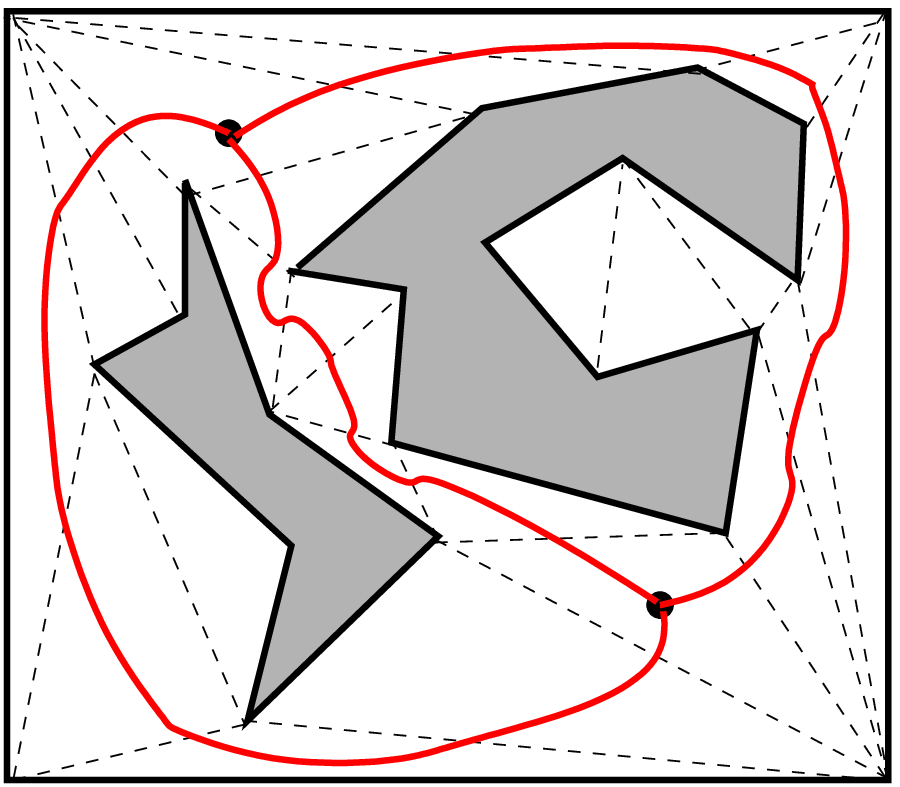}
\caption{\footnotesize \cite{ref:ChenCo12arXiv,ref:ChenCo12ICALP} Illustrating a triangulation of the free
space among two obstacles and the corridors (indicated by red solid curves).
There are two junction triangles marked by a large dot inside
each of them, connected by three solid (red) curves. Removing the two
junction triangles results in three corridors.}
\label{fig:triangulation}
\end{center}
\end{minipage}
\hspace*{0.04in}
\begin{minipage}[t]{0.52\linewidth}
\begin{center}
\includegraphics[totalheight=1.2in]{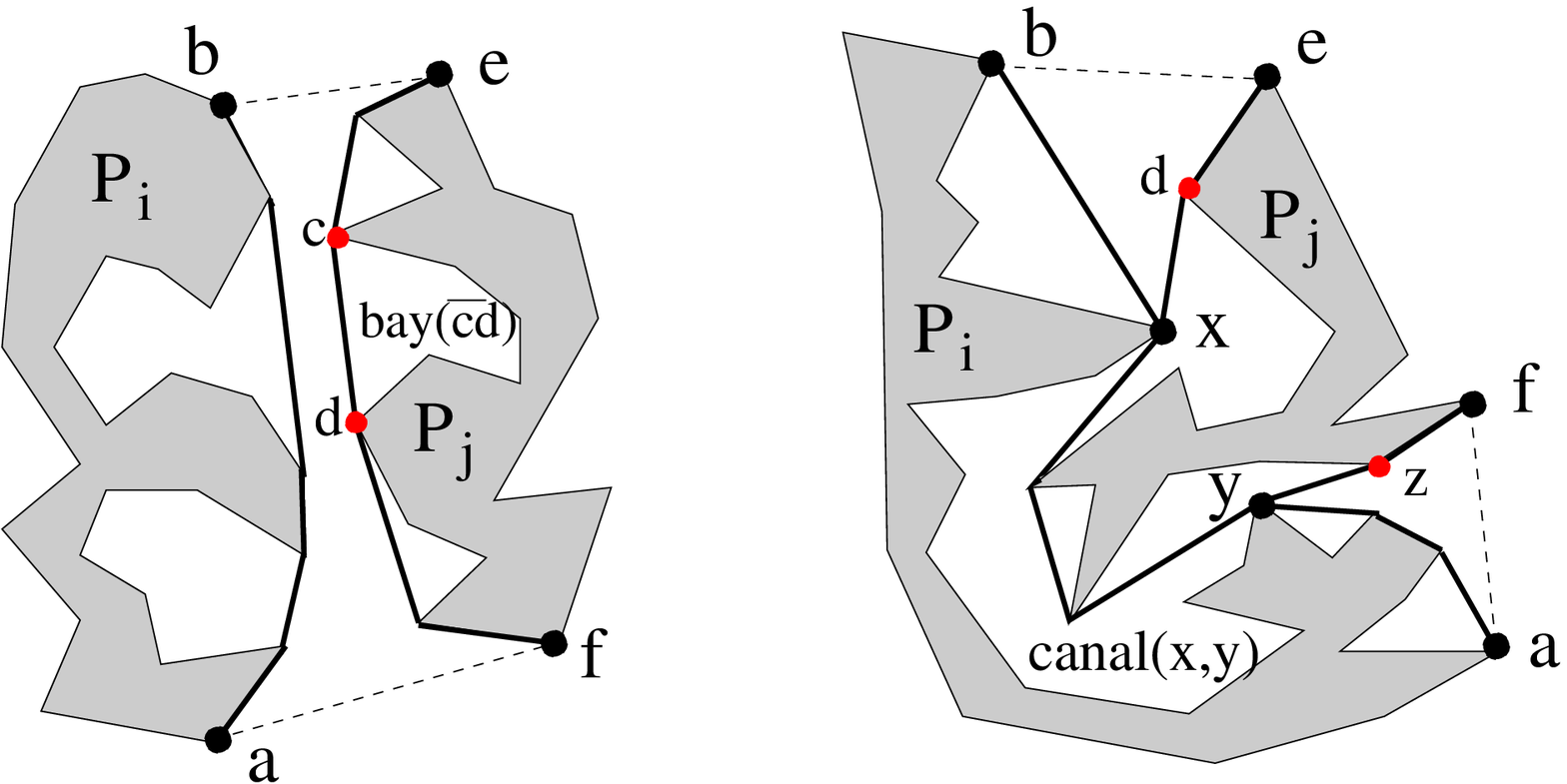}
\caption{\footnotesize \cite{ref:ChenCo12arXiv,ref:ChenCo12ICALP} Illustrating an open hourglass (left) and a
closed hourglass (right) with a corridor path connecting the apices
$x$ and $y$ of the two funnels. The dashed segments are diagonals.
The paths $\pi(a,b)$ and $\pi(e,f)$ are shown with thick solid
curves. A bay $bay(\overline{cd})$ with gate $\overline{cd}$ (left)
and a canal $\canal$ with gates $\overline{xd}$ and $\overline{yz}$
(right) are also indicated.} \label{fig:corridor}
\end{center}
\end{minipage}
\vspace*{-0.15in}
\end{figure}

Let $\calM$ be the union of the $O(h)$ junction triangles, open
hourglasses, and funnels.  Then $\calM\subseteq\calF$.
We call $\calM$ the {\em ocean}.
Since the sides of open hourglasses and funnels
are all convex, the boundary
$\partial\calM$ of $\calM$ consists of $O(h)$ convex chains with
a total of $O(n)$ vertices; also, there are $O(h)$ reflex vertices on $\partial\calM$,
which are corridor path terminals.
We further partition the free space $\calF\setminus \calM$ into regions called {\em
bays} and {\em canals}, as follows.

Consider the hourglass $H_C$ of a corridor $C$.
%
%
%
If $H_C$ is open, then $H_C$
has two sides. Let $S_1(H_C)$ be one side of $H_C$.
The obstacle vertices on $S_1(H_C)$ all lie on the same
obstacle, say $P\in\calP$. Let $c$ and $d$ be any two consecutive
vertices on $S_1(H_C)$ such that $\overline{cd}$ is
not an edge of $P$ (e.g., see the left figure in
Fig.~\ref{fig:corridor}, with $P=P_j$). The free region enclosed
by $\overline{cd}$ and the boundary portion of $P$ between $c$ and
$d$ is called a {\em bay}, denoted by
$bay(\overline{cd})$. We call
$\overline{cd}$ the {\em gate} of $bay(\overline{cd})$, which is
an edge shared by $\bay$ and $\calM$.
If $H_C$ is closed, let $x$ and $y$ be the two apices
of its two funnels. Consider two consecutive vertices $c$ and $d$ on
a side of any funnel such that $\overline{cd}$ is not an obstacle
edge. If neither $c$ nor $d$ is a funnel apex, then $c$ and $d$ must
lie on the same obstacle and the segment $\overline{cd}$ also
defines a bay with that obstacle. However, if $c$ or $d$ is a funnel
apex (say, $c=x$), then $c$ and $d$ may lie on different obstacles.
If they lie on the same obstacle, then they also define a bay;
otherwise, we call $\overline{xd}$ the {\em canal gate} at $x=c$
(see Fig.~\ref{fig:corridor}). Similarly, there is a canal gate
at the other funnel apex $y$, say $\overline{yz}$. Let $P_i$ and
$P_j$ be the two obstacles bounding the hourglass $H_C$. The
region enclosed by $P_i$, $P_j$,
$\overline{xd}$, and $\overline{yz}$ that contains the corridor path
of $H_C$ is called a {\em canal}, denoted by $canal(x,y)$.


Every bay or canal is a simple polygon.
The ocean, bays, and canals together
constitute the free space $\calF$. While the
total number of all bays is $O(n)$, the total number of all canals is
$O(h)$.





\subsection{Queries in the Ocean $\calM$}

For any two points $s$ and $t$ in the ocean $\calM$, it has been
proved that there exists an $L_1$ shortest \st\ path in the free space of the union of $\calM$
and all corridor paths \cite{ref:ChenA11ESA,ref:ChenCo12arXiv,ref:ChenL113STACS}.
Let $\calM'$ be the union of $\calM$
and all corridor paths. Thus, if $s$ and $t$ are both in $\calM$, then there is a shortest \st\ path in $\calM'$.

In this subsection, we will first construct a graph $G_E(\calM)$ of
size $O(h\cdot \sqrt{\log h} \cdot 2^{\sqrt{\log h}})$ on
$\calM$, in a similar fashion as $G_E$ in
Section \ref{sec:newgraph}. Using the graph $G_E(\calM)$ and with
additional $O(n)$ space, for any query points $s$ and $t$ in $\calM$, the
shortest path query can be answered in $O(\log n)$ time.

Let $\calQ=\calR\setminus\calM$. Note that $\partial\calQ$ is $\partial\calM$.
Hence, $\partial\calQ$ consists of $O(h)$ convex chains with totally $O(n)$ vertices,
and $\partial\calQ$ also contains $O(h)$ reflex vertices that are
corridor path terminals. Since $\calP$ has $h$ obstacles,
$\calQ$ contains at most $h$ connected components and each
obstacle of $\calP$ is contained in a component of $\calQ$.
For any point $q$ in $\calM$, in this subsection, let $q^l$,
$q^r,q^u$, and $q^d$ denote the leftward, rightward,
upward, and downward projection points of $q$ on $\partial\calQ$, respectively.

An obstacle vertex $p$ on $\partial\calQ$ is said to be {\em extreme} if both its
incident edges on $\partial\calQ$ are on the same side of the vertical or horizontal line through $p$. Let $V_e(\calQ)$ denote the set of all extreme vertices and corridor
path terminals of $\calQ$. Since $\partial\calQ$ consists of $O(h)$
convex chains and $O(h)$ reflex vertices that are corridor path
terminals, $|V_e(\calQ)|=O(h)$.  We could build a graph
on $V_e(\calQ)$ with respect to $\calQ$ in a similar way as we built
$G_E$ on the obstacle vertices of $\calP$ in Section \ref{sec:newgraph},
and then use this graph to answer queries when both query points are in
$\calM$. However, in order to handle the general queries (in
Section \ref{subsec:general}) for which at least one query point is not in
$\calM$, we need to consider more points for building the graph.
Specifically, let $\calV(\calQ)=\{p^l,p^r, p^u, p^d \ |\ p\in V_e(\calQ)\} \cup
V_e(\calQ)$, i.e., in addition to $V_e(\calQ)$,
$\calV(\calQ)$ also contains the four projections of all points in
$V_e(\calQ)$ on $\partial\calQ$.
Since $|V_e(\calQ)|=O(h)$, $|\calV(\calQ)|=O(h)$.

For each connected component $Q$ of
$\calQ$, let $\calV(Q)$ denote the set of points of $\calV(\calQ)$ on $Q$.
Consider any two points $a$ and $b$ of $\calV(Q)$ that are consecutive on the
boundary $\partial Q$ of $Q$. By the definition of $a$ and
$b$, the boundary portion of $\partial Q$ between $a$ and $b$ that
contains no other points of $\calV(Q)$ must be an $xy$-monotone
path (similar results were also given in
\cite{ref:ChenA11ESA,ref:ChenCo12arXiv,ref:ChenL113STACS,ref:InkuluPl09}), and we
call it an {\em elementary curve} of $\partial Q$. Hence, for any two
points on an elementary curve, the portion of the curve between the
two points is a shortest path between the two points.

Our goal is to build a graph, denoted by $G_E(\calM)$, on $\calV(\calQ)$
with respect to $\calQ$
in a similar way as we built $G_E$ in Section \ref{sec:newgraph}, and use it to
answer queries. To argue the correctness of our approach, we
also define a graph $G_{old}(\calM)$ on $\calV(\calQ)$ and $\calQ$ in a
similar way as $G_{old}$ on $\calP$. Again,
$G_{old}(\calM)$ is only for showing the correctness of our approach
based on $G_E(\calM)$ (recall that we use $G_{old}$ to show the
correctness of using $G_E$). Below, we define $G_E(\calM)$ and
$G_{old}(\calM)$ simultaneously.


We first define their node sets. Each point of $\calV(\calQ)$ defines a
node in both graphs. In addition, $G_{old}(\calM)$ has type-1 and
type-2 Steiner points as nodes; $G_E(\calM)$ has type-1,
type-2, and type-3 Steiner points as nodes. Such Steiner points are
defined using $\calV(\calQ)$ in a similar way as before, but with respect to $\partial\calQ$.
Specifically, for
each point $p \in \calV(\calQ)$, its four projections
$p^l,p^r,p^u$, and $p^d$ on $\partial \calQ$
are type-1 Steiner points. Let $T(\calM)$ be the cut-line tree
defined on the points of $\calV(\calQ)$, similar to $T(\calP)$.
Each node $u$ of $T(\calM)$ is associated
with a subset $V(u)\subseteq \calV(\calQ)$ and a vertical cut-line
$l(u)$ through the median of the $x$-coordinates of the points in $V(u)$.
Since $|\calV(\calQ)|=O(h)$, $T(\calM)$ has
$O(\log h)$ levels and $O(\sqrt{\log h})$ super-levels.
For every node $u\in T(\calM)$, for each point $p\in V(u)$,
if $p$ is horizontally visible to $l(u)$, then the projection of $p$
on $l(u)$ is a type-2 Steiner point.
Also, there are $O(h\sqrt{\log h}2^{\sqrt{\log h}})$ type-3 Steiner
points on the cut-lines of $T(\calM)$, which are defined in
a similar way as in Section \ref{sec:newgraph}, and we omit the details.

The edge sets of the two graphs are defined similarly as
those in $G_{old}$ and $G_E$. We only point out the differences here.
One big difference is that for each corridor path, since its two terminals
define two nodes in both $G_{old}(\calM)$ and $G_E(\calM)$, $G_E(\calM)$ has
an edge connecting these two nodes in both graphs whose weight
is the length of the corridor path.
Another subtle difference is as follows.
In $G_{old}$ and $G_E$, for each obstacle edge $e$ of $\calP$, both
graphs have an edge connecting each pair of consecutive graph nodes
on $e$. In contrast, here we consider
each individual elementary curve of $\calQ$ instead of each
individual edge of $\calQ$ because not every vertex of
$\calQ$ defines a node in $G_{old}(\calM)$ and $G_E(\calM)$.
Specifically, consider each elementary curve $\beta$ of $\calQ$. Note that
the two endpoints of $\beta$ must be in $\calV(Q)$ and thus define two nodes
in both graphs. For each pair of consecutive graph nodes along $\beta$, we put an edge
in both $G_{old}(\calM)$ and $G_E(\calM)$ whose weight is the
length of the portion of $\beta$ between these two points.
We then have the following lemma.

\begin{lemma}\label{lem:40}
For any two points $u$ and $v$ in $\calV(\calQ)$,
a shortest path from $u$ to $v$ in $G_{old}(\calM)$ (resp., $G_E(\calM)$)
corresponds to a shortest path from $u$ to $v$ in the plane.
\end{lemma}
\begin{proof}
We first show that a shortest path from $u$ to $v$ in
$G_{old}(\calM)$ corresponds to a shortest
path from $u$ to $v$ in the plane, and then show a shortest path from $u$ to $v$ in
$G_E(\calM)$ corresponds to a shortest path from $u$ to $v$ in
$G_{old}(\calM)$. This will prove the lemma.

To show a shortest path from $u$ to $v$ in
$G_{old}(\calM)$ corresponds to a shortest
path from $u$ to $v$ in the plane, we will build a new graph $G$ and
prove the following: (1) a shortest path from $u$ to $v$ in $G$
corresponds to a shortest path from $u$ to $v$ in $G_{old}(\calM)$,
and (2) a shortest path from $u$ to $v$ in $G$ corresponds to a
shortest path from $u$ to $v$ in the plane. Below, to define the
graph $G$, we first review some observations that have been discovered
before.

\begin{figure}[t]
\begin{minipage}[t]{\linewidth}
\begin{center}
\includegraphics[totalheight=1.2in]{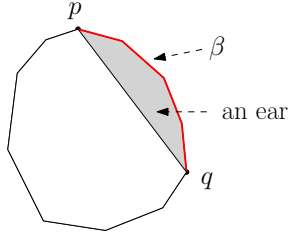}
\caption{\footnotesize Illustrating an ear bounded by $\overline{pq}$ and an elementary curve $\beta$.}
\label{fig:core}
\end{center}
\end{minipage}
\vspace*{-0.15in}
\end{figure}

Let $Q$ be any connected component of $\calQ$. Consider an elementary curve $\beta$ of
$Q$ with endpoints $p$ and $q$. By
the definition of elementary curves, the line segment $\overline{pq}$
must be inside $Q$ (similar results were given in
\cite{ref:ChenA11ESA,ref:ChenCo12arXiv,ref:ChenL113STACS}). We call
the region enclosed by $\beta$ and $\overline{pq}$ an {\em ear} of $Q$,
$\overline{pq}$ the {\em base} of the ear, and $\beta$ the elementary curve of the ear.
It is possible that $\beta$ is $\overline{pq}$, in which case
the ear is $\overline{pq}$. It is easy to see that the bases of all
elementary curves of $Q$ do not intersect except at their endpoints
\cite{ref:ChenA11ESA,ref:ChenCo12arXiv,ref:ChenL113STACS}. Hence, if
we connect the bases of its elementary curves, we obtain a simple
polygon that is contained in $Q$; we call this simple polygon the
{\em core} of $Q$, denoted by $Q_{core}$. Clearly, the union of
$Q_{core}$ and all the ears of $Q$ is $Q$. Denote by $\calQ_{core}$ the set of
cores of all components of $\calQ$. Note
that the vertex set of $\calQ_{core}$ is $\calV(\calQ)$ and the edges of
$\calQ_{core}$ are the bases of all ears of $\calQ$.
Thus, $\calQ_{core}$ has $O(h)$ vertices and edges.
By the results in
\cite{ref:ChenA11ESA,ref:ChenCo12arXiv,ref:ChenL113STACS}, for any
two points in $\calM$, in particular, any two vertices $u$ and $v$ in $\calV(\calQ)$,
there is a shortest $u$-$v$ path in the plane that avoids all
cores of $\calQ_{core}$ and possibly contains corridor paths. More
specifically, there exists a shortest path $\pi(u,v)$ from $u$ to $v$
that contains a sequence of vertices of $\calV(\calQ)$,
$p_1,p_2,\ldots,p_k$, in this order, with $u=p_1$ and $v=p_k$, such that for
any two consecutive vertices $p_i$ and $p_{i+1}$, $1\leq i\leq k-1$, if
$p_i$ and $p_{i+1}$ are terminals of the same corridor path, then the
entire corridor path is contained in $\pi(u,v)$, or else $\pi(u,v)$
contains the line segment $\overline{p_ip_{i+1}}$ which does not
intersect the interior of any core in $\calQ_{core}$.


We build a graph $G=G_{old}(\calQ_{core})$ on $\calV(\calQ)$ with
respect to the cores of $\calQ$, in the same way as $G_{old}$ on $\calP$ in
\cite{ref:ChenSh00,ref:ClarksonRe87,ref:ClarksonRe88}, with the only
difference that if two nodes of $G$ are terminals of the same corridor
path, then there is an extra edge in $G$ connecting these two nodes whose
weight is the length of the corridor path. Note that
$u$ and $v$ define two nodes in $G$. Based on the above discussion,
we claim that the shortest path $\pi(u,v)$ defined above must correspond to a
shortest path from $u$ to $v$ in $G$. Indeed, for any $i$, $1\leq
i\leq k-1$, if $p_i$ and $p_{i+1}$ are terminals of the same corridor
path, then recall that $\pi(u,v)$ contains the entire corridor path and
there is an edge in $G$ connecting $p_i$ and $p_{i+1}$ whose weight
is the length of that corridor path; otherwise, $\pi(u,v)$
contains the segment $\overline{p_ip_{i+1}}$ and by the proof in
\cite{ref:ClarksonRe87,ref:ClarksonRe88}, there must be a path in
$G$ whose length is equal to that of
$\overline{p_ip_{i+1}}$ since $p_i$ is visible to $p_{i+1}$ with
respect to the cores of $\calQ$. This proves that there is a shortest
$u$-$v$ path in $G$ whose length is equal to that of $\pi(u,v)$.

Next, we prove that a shortest $u$-$v$ path in $G$ must
correspond to a shortest $u$-$v$ path in $G_{old}(\calM)$.
To make the paper self-contained we give some details below; for complete details, please refer to \cite{ref:InkuluPl09,ref:ChenA11ESA,ref:ChenCo12arXiv}.
Both $G_{old}(\calM)$ and $G$ are built on $\calV(\calQ)$
in the same way, with
the only difference that $G_{old}(\calM)$ is built
with respect to $\calQ$ while $G$ is built
with respect to $\calQ_{core}$. A useful fact is that for
any two points $a$ and $b$ on any elementary curve $\beta$, the length of the portion of $\beta$ between
$a$ and $b$ is equal to that of the segment $\overline{ab}$ because $\beta$ is $xy$-monotone.
Note that the space outside $\calQ_{core}$ is the union
of the space outside $\calQ$ and all ears of $\calQ$.
Since both graphs have extra edges to connect corridor
path terminals, to prove that a shortest $u$-$v$ path in $G$
corresponds to a shortest $u$-$v$ path in $G_{old}(\calM)$,
based on the analysis in \cite{ref:ClarksonRe87,ref:ClarksonRe88}, we
only need to show the following: For any two vertices $a$ and $b$ of
$\calV(\calQ)$ visible to each other with respect to $\calQ_{core}$ such
that no other vertices of $\calV(\calQ)$ than $a$ and $b$ are in
the axis-parallel rectangle $R(a,b)$ that has $\overline{ab}$ as a diagonal,
there must be an $xy$-monotone path between $a$ and $b$ in
$G_{old}(\calM)$. Note that $a$ may not be visible to $b$ with respect
to $\calQ$.

By the construction of the graph $G$
\cite{ref:ClarksonRe87,ref:ClarksonRe88}, there must be an $xy$-monotone path from $a$ to $b$ in $G$, for which there are two possible cases. Below, we
prove in each case there is also an $xy$-monotone path from $a$
to $b$ in $G_{old}(\calM)$. Without loss of generality, we assume $b$ is
to the northeast of $a$.

\begin{enumerate}
\item
{\bf Case 1}.
If any core of $\calQ_{core}$ intersects the interior of the
rectangle $R(a,b)$, then as shown in
\cite{ref:ClarksonRe87,ref:ClarksonRe88}, either the rightward projection
of $a$ on $\partial\calQ_{core}$ and the downward projection of $b$ on
$\partial\calQ_{core}$ are both on the same edge of $\partial\calQ_{core}$
that intersects $R(a,b)$ (e.g., see Fig.~\ref{fig:case1}),
or the upward projection
of $a$ on $\partial\calQ_{core}$ and the leftward projection of $b$ on
$\partial\calQ_{core}$ are both on the same edge of $\partial\calQ_{core}$ that intersects
$R(a,b)$.  Here, we assume that the former case
occurs. Let $a_1$ be the rightward projection
of $a$ on $\partial\calQ_{core}$ and $b_1$ be the downward projection of $b$ on
$\partial\calQ_{core}$, and $\overline{a_2b_2}$ be the edge of $\calQ_{core}$ that
contains both $a_1$ and $b_1$.
By the construction of $G$, there is an
$xy$-monotone path from $a$ to $b$ consisting of
$\overline{aa_1}\cup\overline{a_1b_1}\cup\overline{b_1b}$. Below, we show
that there is also an $xy$-monotone path from $a$ to $b$ in
$G_{old}(\calM)$.

Let $ear(\overline{a_2b_2})$ be the ear of $\calQ$ whose base is
$\overline{a_2b_2}$. Let $\beta$ be the elementary curve of
$ear(\overline{a_2b_2})$. Since no vertex of $\calV(\calQ)-\{a,b\}$ is in $R(a,b)$
and all extreme points of $\calQ$ are in $\calV(\calQ)$, the rightward
projection of $a$ on $\partial\calQ$ and the downward projection of
$b$ on $\partial\calQ$ must be both on $\beta$ (e.g., see
Fig.~\ref{fig:case1new}); we denote these two
projection points by $a_3$ and $b_3$, respectively. By the
construction of $G_{old}(\calM)$, there must be an $xy$-monotone path
from $a$ to $b$ in $G_{old}(\calM)$ that is a concatenation of
$\overline{aa_3}$, the portion of $\beta$ between $a_3$ and $b_3$, and
$\overline{b_3b}$ (note that $a_3$ is a type-1 Steiner point defined by $a$ and $b_3$ is a type-1
Steiner point defined by $b$ in $G_{old}(\calM)$).

\begin{figure}[t]
\begin{minipage}[t]{0.49\linewidth}
\begin{center}
\includegraphics[totalheight=1.2in]{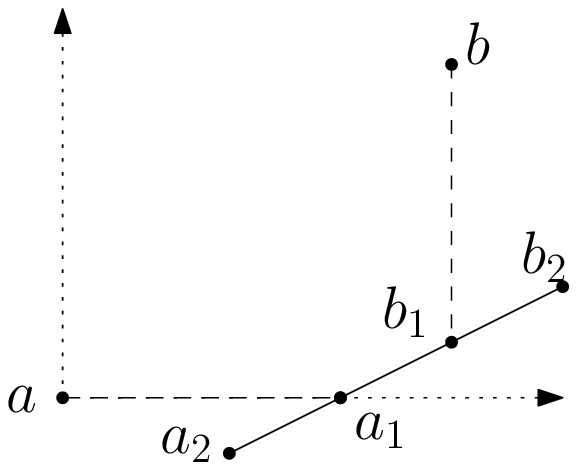}
\caption{\footnotesize Illustrating the proof of Lemma \ref{lem:40}:
$a_1$ is the rightward projection of $a$ on $\partial\calQ_{core}$ and $b_1$
is the downward projection of $b$ on $\partial\calQ_{core}$.}
\label{fig:case1}
\end{center}
\end{minipage}
\hspace*{0.04in}
\begin{minipage}[t]{0.49\linewidth}
\begin{center}
\includegraphics[totalheight=1.2in]{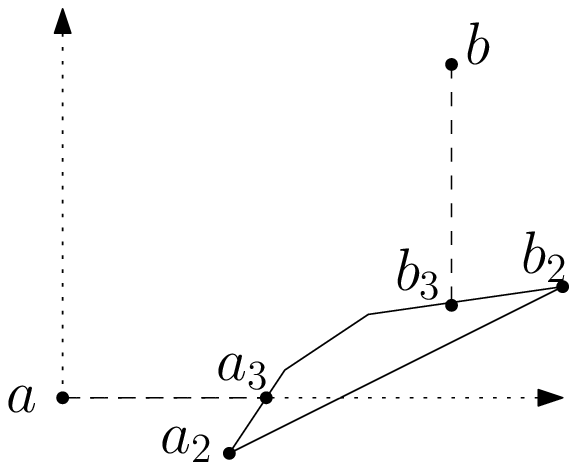}
\caption{\footnotesize Illustrating the proof of Lemma \ref{lem:40}:
$a_3$ is the rightward projection of $a$ on $\partial\calQ$ and $b_3$
is the downward projection of $b$ on $\partial\calQ$. Both $a_3$ and
$b_3$ must be on the same elementary curve $\beta$.}
\label{fig:case1new}
\end{center}
\end{minipage}
\vspace*{-0.15in}
\end{figure}

\item
{\bf Case 2}.
If no core of $\calQ_{core}$ intersects the interior of the rectangle
$R(a,b)$, then by the construction of $G$, there must be
a cut-line $l$ between $a$ and $b$ such that on $l$, $a$ defines a Steiner point
$a_h(l)$ and $b$ defines a Steiner point $b_h(l)$ (e.g., see
Fig.~\ref{fig:case2}). Thus, there is an
$xy$-monotone path from $a$ to $b$ in $G$ consisting of
$\overline{aa_h(l)}\cup
\overline{a_h(l)b_h(l)}\cup\overline{b_h(l)b}$. Below, we show that
there is also an $xy$-monotone path from $a$ to $b$ in
$G_{old}(\calM)$.

\begin{figure}[h]
\begin{minipage}[t]{0.49\linewidth}
\begin{center}
\includegraphics[totalheight=1.2in]{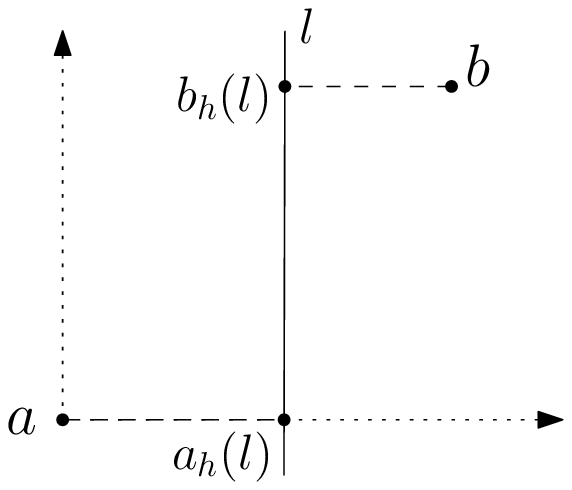}
\caption{\footnotesize Illustrating the proof of Lemma \ref{lem:40}:
$a_h(l)$ is the rightward projection of $a$ on $l$ and $b_h(l)$
is the leftward projection of $b$ on $l$.}
\label{fig:case2}
\end{center}
\end{minipage}
\hspace*{0.04in}
\begin{minipage}[t]{0.49\linewidth}
\begin{center}
\includegraphics[totalheight=1.2in]{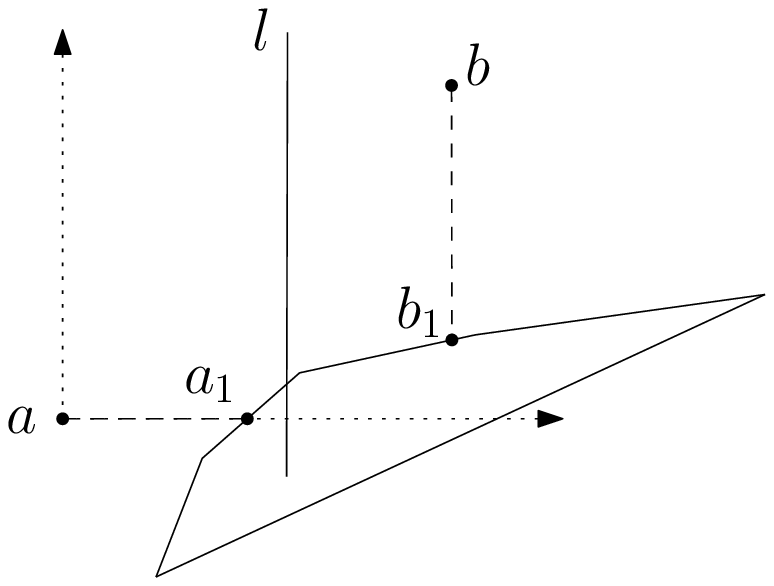}
\caption{\footnotesize Illustrating the proof of Lemma \ref{lem:40}:
$a_1$ is the rightward projection of $a$ on $\partial\calQ$ and $b_1$
is the downward projection of $b$ on $\partial\calQ$. Both $a_1$ and
$b_1$ must be on the same elementary curve $\beta$.}
\label{fig:case2new}
\end{center}
\end{minipage}
\vspace*{-0.15in}
\end{figure}

Since both $G$ and $G_{old}(\calM)$ are built on $\calV(\calQ)$, they have the same cut-line tree.
Hence, the cut-line $l$ still exists in $G_{old}(\calM)$. If
both $a$ and $b$ are horizontally visible to $l$, then they still
define Steiner points on $l$ and consequently there is also
an $xy$-monotone path from $a$ to $b$ in $G_{old}(\calM)$. Otherwise, we
assume that $a$ is not horizontally visible to $l$. Let $a_1$ be the
rightward projection of $a$ on $\partial\calQ$ (see
Fig.~\ref{fig:case2new}). Hence, $a_1$ must be
between $l$ and $a$. Let $\beta$ be the elementary curve that contains
$a_1$. Thus, $\beta$ intersects the lower edge of $R(a,b)$ at $a_1$.
Since $R(a,b)$ does not contain any point of $\calV(\calQ)-\{a,b\}$,
the two endpoints of $\beta$ are not in
$R(a,b)$ and thus the downward projection of $b$ on $\partial\calQ$,
denoted by $b_1$, must be on $\beta$ as well. By the
construction of $G_{old}(\calM)$, there must be an $xy$-monotone path
from $a$ to $b$ in $G_{old}(\calM)$ that is the concatenation of
$\overline{aa_1}$, the portion of $\beta$ between $a_1$ and $b_1$, and
$\overline{b_1b}$.
\end{enumerate}

The above arguments prove that a shortest path from $u$ to $v$ in
$G_{old}(\calM)$ corresponds to a shortest path from $u$ to $v$ in the
plane.

It remains to show
that a shortest $u$-$v$ path in $G_{old}(\calM)$ corresponds
to a shortest $u$-$v$ path in $G_E(\calM)$. This can be
seen easily since for any edge $e=\overline{pq}$ in $G_{old}(\calM)$,
if $e$ is not
in $G_E(\calM)$, then $e$ is ``divided" into many edges in $G_E(\calM)$ such
that their concatenation is a path from $p$ to $q$.

The lemma thus follows.
\end{proof}

The next lemma gives an algorithm for computing the graph $G_E(\calM)$.

\begin{lemma}\label{lem:50}
The graph $G_E(\calM)$ can be computed in
$O(n+h\log^{3/2}h2^{\sqrt{\log h}})$ time.
\end{lemma}
\begin{proof}
The algorithm for constructing $G_E(\calM)$ is similar to that for $G_E$ in Lemma \ref{lem:10}.
As a preprocessing, the free space $\calF$ can be triangulated in
$O(n+h\log^{1+\epsilon}h)$ time for any constant $\epsilon>0$
\cite{ref:Bar-YehudaTr94}, after which computing the extended corridor
structure, in particular, takes $O(n+h\log h)$ time
\cite{ref:ChenA11ESA,ref:ChenCo12arXiv,ref:ChenL113STACS}.
Consequently, we obtain
$\calQ$ and the vertex set $\calV(\calQ)$. All corridor paths are also available.

First, we compute the four projections of each point of
$\calV(\calQ)$ on $\partial\calQ$ as type-1 Steiner points, which can be
done after we compute the vertical and horizontal visibility
decompositions of $\calQ$ in $O(n+h\log^{1+\epsilon}h)$ time
\cite{ref:Bar-YehudaTr94}. The graph edges for connecting each point
of $\calV(\calQ)$ to its four projection points on $\partial\calQ$ can
be obtained as well.

Next, we compute the type-2 and type-3 Steiner points and the corresponding graph edges
connecting these Steiner points. Since $|\calV(\calQ)|=O(h)$, the
cut-line tree $T(\calM)$ can be computed in $O(h\log h)$ time.
Then, we determine the Steiner points on the cut-lines by traversing the tree $T(\calM)$ from top to bottom in a similar way as in Lemma \ref{lem:10}. Since we have obtained the four projection points for each point of $\calV(\calQ)$, computing all Steiner
points on the cut-lines takes $O(h\sqrt{\log h}2^{\sqrt{\log h}})$
time. Their corresponding edges can be computed in $O(h\log^{3/2} h2^{\sqrt{\log h}}\log n)$ time.

It remains to compute the graph edges of $G_E(\calM)$ connecting consecutive
graph nodes on each elementary curve of $\calQ$ and the graph edges connecting
every two consecutive Steiner points (if they are visible to each other)
on each cut-line.

On each connected component $Q$ of $\calQ$, we could compute a sorted
list of all Steiner points and the points of $\calV(Q)$
by sorting all these points and all obstacle
vertices of $Q$ along $\partial Q$. But that would take $O(n\log n)$ time
in total because there are $O(n)$ obstacle vertices on all components
of $\calQ$. To do better, we take the following approach. For
each elementary curve $\beta$, we sort all Steiner points on $\beta$ by either their
$x$-coordinates or $y$-coordinates. Since $\beta$ is $xy$-monotone,
such an order is also an order along $\beta$. Then, we merge
the Steiner points thus ordered with the obstacle vertices on $\beta$,
in linear time. Since there are $O(h)$ Steiner points on
$\partial\calQ$, it takes totally $O(n+h\log h)$ time to sort the
Steiner points and obstacle vertices on all elementary curves of
$\calQ$. After that, the edges of $G_E(\calM)$ on all elementary
curves can be computed immediately.

We now compute the graph edges on the cut-lines connecting consecutive
Steiner points. We first sort all Steiner points on each cut-line. This
sorting takes $O(h\log^{3/2} h2^{\sqrt{\log n}})$ time for all
cut-lines. For each pair of consecutive Steiner points $p$ and $q$ on
every cut-line, we determine whether $p$ is visible to $q$ by
checking whether the upward projections of $p$ and $q$ on
$\partial\calQ$ are equal, and these upward projections can be
performed in $O(\log n)$ time using the vertical visibility
decomposition of $\calQ$. Hence, the graph edges on all cut-lines are
computed in $O(h\sqrt{\log h}2^{\sqrt{\log h}}\cdot\log n)$ time.

In summary, we can compute the graph $G_E(\calM)$ in $O(n+h\sqrt{\log
h}2^{\sqrt{\log h}}\cdot\log n)$ time. Note that $n+h\sqrt{\log
h}2^{\sqrt{\log h}}\cdot\log n=O(n+h\log^{3/2} h2^{\sqrt{\log h}})$. The
lemma thus follows.
\end{proof}

Consider any two query points $s$ and $t$ in the ocean
$\calM$. We define the gateway sets
$V_g(s,G_E(\calM))$ for $s$ and $V_g(t,G_E(\calM))$ for $t$ on $G_E(\calM)$,
as follows. We only discuss $V_g(s,G_E(\calM))$;
$V_g(t,G_E(\calM))$ is similar. The definition of
$V_g(s,G_E(\calM))$ is very similar to that of $V_g(s,G_E)$, with only
slight differences. Specifically, $V_g(s,G_E(\calM))$ has two subsets
$V^1_g(s,G_E(\calM))$ and $V^2_g(s,G_E(\calM))$.
$V^2_g(s,G_E(\calM))$ is defined in the same way as
$V^2_g(s,G_E)$, and thus $|V^2_g(s,G_E(\calM))|=O(\sqrt{\log h})$.
$V^1_g(s,G_E(\calM))$ is defined with respect to the elementary
curves of $\calQ$, as follows. Let $q$ be the rightward projection point of $s$
on $\partial\calQ$. Suppose $q$ is on the elementary curve $\beta$ and
$p_1$ and $p_2$ are the two
nodes of $G_E(\calM)$ on $\beta$ adjacent to $q$. Then $p_1$ and $p_2$
are in $V^1_g(s,G_E(\calM))$, and for each $p\in
\{p_1,p_2\}$, we define a gateway edge from $s$ to $p$ consisting of
$\overline{sq}$ and the portion of $\beta$ between $q$ and $p$. Similarly,
for each of the leftward, upward, and downward projections of $s$ on
$\partial\calQ$, there are at most two gateways in $V^1_g(s,G_E(\calM))$.

The next lemma shows that the gateways of $V_g(s,G_E(\calM))$
``control'' the shortest paths from $s$ to all points of $\calV(\calQ)$.

\begin{lemma}\label{lem:60}
For any point $p$ of $\calV(\calQ)$, there exists a shortest path from
$s$ to $p$ using $G_E(\calM)$ that contains a gateway of $s$ in
$V_g(s,G_E(\calM))$.
\end{lemma}
\begin{proof}
We define a gateway set $V_g(s,G_{old}(\calM))$ for $s$ on the graph
$G_{old}(\calM)$, as follows. The set $V_g(s,G_{old}(\calM))$
has two subsets $V^1_g(s,G_{old}(\calM))$ and
$V^2_g(s,G_{old}(\calM))$. The first subset $V^1_g(s,G_{old}(\calM))$
is exactly the same as $V^1_g(s,G_E(\calM))$, and the second subset
$V^2_g(s,G_{old}(\calM))$ contains gateways on the cut-lines of
$T(\calM)$, which are defined similarly as $V^2_g(s,G_{old})$ on
$G_{old}$ and $T(\calP)$, discussed in Section \ref{sec:pre}. Note
that the gateways in $V_g(s,G_{old}(\calM))$ are exactly those nodes
of $G_{old}(\calM)$ that are adjacent to $s$ if we ``insert'' $s$
into the graph $G_{old}(\calM)$ (similar arguments were used for
$V_g(s,G_{old})$ in \cite{ref:ChenSh00}). Hence, there exists a
shortest path from $s$ to $p$ using $G_{old}(\calM)$ that contains a
gateway of $s$ in $V_g(s,G_{old}(\calM))$.

Since the graph $G_E(\calM)$ is defined analogously as $G_E$ and
$G_{old}(\calM)$ is defined analogously as $G_{old}$, by using a similar
analysis as in the proof of Lemma \ref{lem:20}, we can show that there
exists a shortest path from $s$ to $p$ using $G_E(\calM)$ that contains a
gateway of $s$ in $V_g(s,G_E(\calM))$. We omit the details. The
lemma thus follows.
\end{proof}

Similar results also hold for the gateway set $V_g(t,G_E(\calM))$ of $t$.
We have the following corollary.

\begin{corollary}\label{cor:20}
If there exists a shortest \st\ path through a point of
$\calV(\calQ)$, then there exists a shortest \st\ path through a gateway of
$s$ in $V_g(s,G_E(\calM))$ and a gateway of $t$ in $V_g(t,G_E(\calM))$.
\end{corollary}

The following lemma gives an algorithm for computing the gateways.

\begin{lemma}\label{lem:70}
With a preprocessing of $O(n+h\cdot \log^{3/2} h \cdot 2^{\sqrt{\log h}})$ time and
$O(n+h\cdot \sqrt{\log h} \cdot 2^{\sqrt{\log h}})$ space,
the gateway sets $V_g(s,G_E(\calM))$
and $V_g(t,G_E(\calM))$ can be computed in
$O(\log n)$ time for any two query points $s$ and $t$ in $\calM$.
\end{lemma}
%
%
%
%
%
\begin{proof}
The algorithm is similar to that for Lemma \ref{lem:30}; we only point
out the differences. We discuss our algorithm only for computing
$V_g(s,G_E(\calM))$; the case for $V_g(t,G_E(\calM))$ is similar.

To compute $V^1_g(s,G_E(\calM))$, we build the horizontal and vertical
visibility decompositions of $\calQ$.
Then, the four projections of $s$ on $\partial\calQ$ can be determined
in $O(\log n)$ time. Consider
any such projection $p$ of $s$. Suppose $p$ is on an elementary curve $\beta$. We need to
determine the two nodes of $G_E(\calM)$ on $\beta$ adjacent to $p$,
which are gateways of $V^1_g(s,G_E(\calM))$.
We maintain a sorted list of all nodes of $G_E(\calM)$ on $\beta$, and
do binary search to find these two gateways of $s$ on $\beta$ in
this sorted list by using only the $y$-coordinates (or the
$x$-coordinates) of the nodes since $\beta$ is
$xy$-monotone. Also, since $\beta$ is $xy$-monotone, for any two points
$q$ and $q'$ on $\beta$, the length of the portion of $\beta$ between $q$ and
$q'$ is equal to the length of $\overline{qq'}$. Hence, after these
two gateways of $s$ on $\beta$ are found, the lengths of the two gateway
edges from $s$ to them can be
computed in constant time. Since $V^1_g(s,G_E(\calM))$ has $O(1)$
%
%
gateways, $V^1_g(s,G_E(\calM))$ can be computed in $O(\log n)$ time.

To compute $V^2_g(s,G_E(\calM))$, we take the same approach as for Lemma \ref{lem:30}.
In the preprocessing, for every cut-line $l$, we maintain a sorted list of all Steiner points
on $l$, and associate with each such Steiner point its upward and downward projections
on $\partial\calQ$.
Computing these projections for each Steiner point takes $O(\log n)$ time.
Then we build a fractional cascading data structure \cite{ref:ChazelleFr86} for
the sorted lists of Steiner points on all cut-lines along the cut-line tree
$T(\calM)$. Using this fractional cascading data structure, the gateway set
$V^2_g(s,G_E(\calM))$ can be computed in $O(\log h)$ time.

The preprocessing takes totally $O(n+h\sqrt{\log h}2^{\sqrt{\log h}}\log n)$ time and
$O(n+h\sqrt{\log h}2^{\sqrt{\log h}})$ space.
Note that
$n+h\sqrt{\log h}2^{\sqrt{\log h}}\log n=O(n+h\log^{3/2}h2^{\sqrt{\log
h}})$.
The lemma thus follows.
\end{proof}

We summarize our algorithm in Lemma
\ref{lem:80} below for the case when both query points are in $\calM$.

\begin{lemma}\label{lem:80}
With a preprocessing of $O(n+h^2\log^2 h 4^{\sqrt{\log h}})$ time and
$O(n+h^2\log h4^{\sqrt{\log h}})$ space,
each two-point query can be answered in $O(\log n)$
time for any two query points in the ocean $\calM$.
\end{lemma}
\begin{proof}
In the preprocessing, we build the graph $G_E(\calM)$, and for each node
$v$ of $G_E(\calM)$, compute a shortest path tree in $G_E(\calM)$
from $v$. We maintain a shortest path length table such that for
any two nodes $u$ and $v$ in $G_E(\calM)$, the shortest path length
between $u$ and $v$ can  be found in $O(1)$ time. Since
$G_E(\calM)$ has $O(h\sqrt{\log h}2^{\sqrt{\log h}})$ nodes and edges, computing and maintaining
all shortest path trees in $G_E(\calM)$ take $O(h^2\log h4^{\sqrt{\log h}})$ space and
$O(h^2\log^2 h4^{\sqrt{\log h}})$ time.
%
%
%
%

To report an actual shortest path in the plane in time linear to the number of edges of
the output path, we need to maintain additional information. Consider an elementary curve
$\beta$ of $\calQ$. Let $u$ and $v$ be two consecutive nodes of $G_E(\calM)$ on $\beta$.
By our definition of $G_E(\calM)$, there is an edge $(u,v)$
in $G_E(\calM)$. If the edge $(u,v)$ is contained in our output path, we
need to report all obstacle vertices and edges of $\beta$ between $u$ and
$v$. For this, on each elementary curve $\beta$, we explicitly
maintain a list of obstacle edge between each pair of consecutive nodes of
$G_E(\calM)$ along $\beta$. Since the total number of nodes of $G_E(\calM)$
on all elementary curves is $O(h)$ and the total number of obstacle
vertices of $\calQ$ is $O(n)$, maintaining such {\em edge lists} for
all elementary curves takes $O(n)$ space.

In addition, we also perform the preprocessing for Lemma \ref{lem:70}.

The overall preprocessing takes $O(n+h^2\log^{2} h 4^{\sqrt{\log h}})$ time and
$O(n+h^2\log h4^{\sqrt{\log h}})$ space.

Now consider any two query points $s$ and $t$ in $\calM$. As
for Theorem \ref{theo:10}, we first check whether there
exists a trivial shortest \st\ path. But trivial shortest paths
here are defined with respect to the elementary curves of $\calQ$ instead of the obstacle
edges of $\calP$. For example, consider $s^r$ (i.e., the rightward projection of
$s$ on $\partial\calQ$) and $t^d$. If $\overline{ss^r}$
intersects $\overline{tt^d}$, then there is a trivial shortest \st\
path $\overline{sq}\cup \overline{qt}$, where $q=\overline{ss^r}\cap
\overline{tt^d}$; otherwise, if $s^r$ and $t^d$ are both on the same
elementary curve $\beta$ of $\calQ$, then there is a trivial shortest
\st\ path which is the concatenation of $\overline{ss^r}$, the
portion of $\beta$ between $s^r$ and $t^d$, and $\overline{t^dt}$.
Similarly, trivial shortest \st\ paths are also defined by other
projections of $s$ and $t$ on $\partial\calQ$.

We can determine whether there exists a trivial shortest \st\ path in
$O(\log n)$ time by using the vertical and horizontal decompositions of
$\calQ$ to compute the four projection points of $s$ and
$t$ on $\partial\calQ$. If yes, we find such a shortest path in
additional time linear to the number of edges of the output path. Note
that for the case, e.g., when $s^r$ and $t^d$ are both on the same elementary curve
$\beta$, the output path may not be of $O(1)$ size since there
may be multiple obstacle vertices on the portion of $\beta$ between $s^r$
and $t^d$; but we can still output such a path in linear time by using the
edge lists we maintain on each elementary curve.
Below, we assume there is no trivial shortest \st\ path.

By using the cores of $\calQ$ in the proof of Lemma \ref{lem:40} and
a similar analysis as in \cite{ref:ChenSh00}, we can show that
there must be a shortest \st\
path that contains at least one point of $\calV(\calQ)$. By
Corollary \ref{cor:20}, there exists a shortest \st\ path through a
gateway of $s$ and a gateway of $t$ in $G_E(\calM)$. Using Lemma \ref{lem:70}, we
compute the two gateway sets $V_g(s,G_E(\calM))$ and
$V_g(t,G_E(\calM))$. By building a gateway graph for $s$ and $t$ as
in Theorem \ref{theo:10}, we can compute the length of a
shortest \st\ path in $O(\log h)$ time since
$|V_g(s,G_E(\calM))|=O(\sqrt{\log h})$,
$|V_g(t,G_E(\calM))|=O(\sqrt{\log h})$, and thus
the gateway graph has $O(\sqrt{\log h})$ nodes and $O(\log h)$ edges.
An actual path can then be reported in additional time linear
to the number of edges of the output path, by using the shortest path
trees of $G_E(\calM)$ and the edge lists maintained on the elementary
curves, as discussed above.
The lemma thus follows.
\end{proof}

%
%
%

\subsection{The General Queries}
\label{subsec:general}

In this section, we show how to handle the general queries in which at least one query point
is not in $\calM$. Without loss of generality, we assume that $s$ is in
a bay or a canal, denoted by $B$. We first focus on the case when
$B$ is a bay. The case when $B$ is a canal can be handled by similar
techniques although it is a little more complicated since each canal
has two gates. The point $t$ can be in $B$, $\calM$,
or another bay or canal, and we discuss these three cases below.
Let $g$ denote the gate of $B$.

%
%
%
As an overview of our approach, we characterize the different possible ways that a
shortest \st\ path may cross the gate $g$, show how to find
such a possible path for each way, and finally compute all possible ``candidate"
paths and select the one with the smallest path length as our solution.

\subsubsection{The Query Point $t$ is in $B$}

When the query point $t$ is in $B$, we have the following lemma.

\begin{lemma}\label{lem:90}
If $B$ is a bay and $t\in B$, then there exists a shortest \st\ path in $B$.
\end{lemma}
\begin{proof}
Let $\pi$ be any shortest \st\ path in the plane. If $\pi$ is in $B$, then we are
done. Otherwise, $\pi$ must intersect the only gate $g$ of $B$; further, since both $s$ and
$t$ are in $B$, if $\pi$ exits from $B$ (through $g$), then it must enter $B$ again
(through $g$ as well).  Let $p$ be the first point on $g$ encountered as
going from $s$ to $t$ along $\pi$ and let $q$ be the last such
point on $g$. Let $\pi'$ be the \st\ path obtained by replacing the portion of $\pi$
between $p$ and $q$ by $\overline{pq}\subseteq g$. Note that $\pi'$ is in $B$.
Since $\overline{pq}$ is a
shortest path from $p$ to $q$, $\pi'$ is also a shortest \st\ path.
The lemma thus follows.
\end{proof}

To handle the case of $t\in B$, in the preprocessing,
we build a data structure for two-point Euclidean shortest path queries
in $B$, denoted by $\calD(B)$,
in $O(|B|)$ time and space \cite{ref:GuibasOp89}.
Since a Euclidean shortest path in any simple polygon is also an $L_1$
shortest path and $B$ is a simple polygon, for $t\in B$, we can use $\calD(B)$ to
answer the shortest \st\ path query in $B$ in $O(\log n)$ time.

\subsubsection{The Query Point $t$ is in $\calM$}

If the query point $t$ is in $\calM$, then a shortest \st\ path must cross the gate $g$ of
$B$. A main difficulty for answering the general queries is to deal with this case.
More specifically, we already have a graph $G_E(\calM)$ on $\calM$,
and our goal is to design a mechanism to connect the bay $B$ with
$G_E(\calM)$ through the gate $g$, so that it can
capture the shortest path information in the union of $B$ and $\calM'$ (recall that $\calM'$ is the union of $\calM$ and all corridor paths).


We begin with some
observations on how a shortest \st\ path may cross $g$. Without
loss of generality, we assume that $g$ has a positive slope and
the interior of $B$ on $g$ is above $g$.
Let $a_1$ and $a_2$ be the two endpoints of $g$ such that $a_1$
is higher than $a_2$ (see Fig.~\ref{fig:funnel2}).
Let $\pi(s,a_1)$ (resp., $\pi(s,a_2)$) be the Euclidean shortest path in $B$ from $s$ to
$a_1$ (resp., $a_2$).
Let $z'$ be the farthest point from $s$ on $\pi(s,a_1)\cap \pi(s,a_2)$
(possibly $z'=s$). Let $\pi(z',a_1)$ (resp., $\pi(z',a_2)$)
be the subpath of $\pi(s,a_1)$ (resp., $\pi(s,a_2)$)
between $z'$ and $a_1$ (resp., $a_2$).
It is well known that both $\pi(z',a_1)$ and
$\pi(z',a_2)$ are convex chains \cite{ref:GuibasLi87,ref:LeeEu84}, and
the region enclosed by $\pi(z',a_1)$, $\pi(z',a_2)$, and $g$ in $B$ is a
``funnel'' with $z'$ as the {\em apex} and $g$ as the {\em base} (see
Fig.~\ref{fig:funnel2}).  Let $F$
denote this funnel and $\partial F$ denote its boundary.

We define four special points $z'_1,z'_2,z_1$, and $z_2$ (see Fig.~\ref{fig:funnel2}).
Suppose we move along $\pi(z',a_1)$ from $z'$; let $z_1'$ be the first point on $\pi(z',a_1)$
we encounter that is {\em horizontally} visible to $g=\overline{a_1a_2}$. Similarly, as
moving along $\pi(z',a_2)$ from $z'$, let $z_2'$ be the first point on $\pi(z',a_2)$ encountered
that is {\em vertically} visible to $g$.
Note that in some cases $z_1'$ (resp., $z'_2$) can be $z'$, $a_1$, or $a_2$.
Let $z_1$ be the horizontal projection
of $z'_1$ on $g$ and $z_2$ be the vertical projection of
$z'_2$ on $g$  (see Fig.~\ref{fig:funnel2}).

\begin{figure}[t]
\begin{minipage}[t]{\linewidth}
\begin{center}
\includegraphics[totalheight=1.7in]{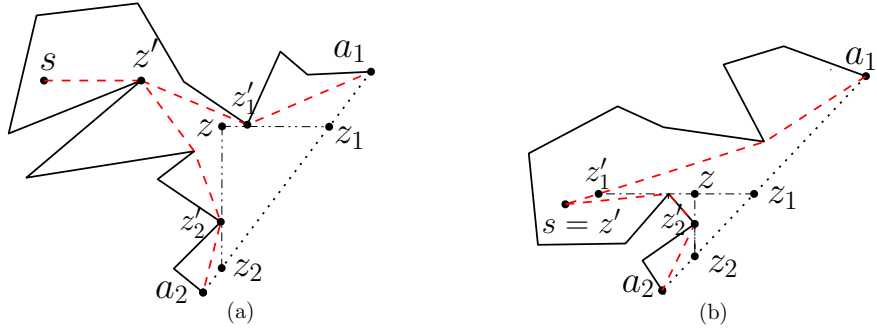}
\caption{\footnotesize Illustrating the definitions of
$z',z'_1,z'_2,z,z_1$, and $z_2$. In (a), $\overline{z_1z_1'}$ is tangent to $\pi(z',a_1)$ (at $z_1'$); in (b), $\overline{z_1z_1'}$ is tangent to $\pi(z',a_2)$.}
\label{fig:funnel2}
\end{center}
\end{minipage}
\vspace*{-0.15in}
\end{figure}

The points $z_1$ and $z_2$ are particularly useful.
We first have the following observation.

\begin{observation}\label{obser:10}
The point $z_1$ is above $z_2$, i.e., the $y$-coordinate of $z_1$ is no smaller than that of $z_2$.
\end{observation}
\begin{proof}
If $z'$ is either $a_1$ or $a_2$, then by their definitions, we have $z_1=z_2=z_1'=z_2'=z'$ and
the observation trivially holds. Suppose $z'$ is neither $a_1$ nor $a_2$. If $z_1=a_1$, then
the observation also holds since $a_1$ is the highest point on $g$. We assume $z_1\neq a_1$,
which implies $z_1'\neq a_1$.

Let $\pi(z_1',a_1)$ be the portion of $\pi(z',a_1)$ between $z_1'$ and $a_1$. Note that the
``pseudo-triangular'' region enclosed by $\overline{a_1z_1}$, $\overline{z_1z_1'}$, and $\pi(z_1',a_1)$ does not contain any point of $\partial B$ in its interior.
For any point $p$ in the interior of $\overline{a_1z_1}$,
since $\pi(z_1',a_1)$ is convex and $\overline{z_1z_1'}$ is horizontal, $p$ must be vertically
visible to $\pi(z_1',a_1)$, say, at a point $q\in \pi(z_1',a_1)$. Clearly, $q$ is not $z_1'$.
Hence, the line containing $\overline{pq}$ cannot be tangent to $\pi(z',a_1)$ at $q$,
implying that $q$ is not $z_2'$. Therefore, the point $p$ must be strictly above $z_2$. Since $p$ is an arbitrary point in the interior of $\overline{a_1z_1}$, $z_1$ must be above $z_2$. The observation thus follows.
\end{proof}

\begin{lemma}\label{lem:100}
For any point $p\in
\overline{a_1z_1}$, there is a shortest path from $s$ to $p$ that
contains $z_1$; likewise,
for any point $p\in \overline{a_2z_2}$, there is a shortest
path from $s$ to $p$ that contains $z_2$.
\end{lemma}
\begin{proof}
We only prove the case of $p\in \overline{a_1z_1}$ since the other case of
$p\in \overline{a_2z_2}$ is symmetric. It suffices to show that there exists a shortest
path from $z'$ to $p\in \overline{a_1z_1}$ that contains $z_1$.

Recall that $z_1$ is the horizontal projection of $z_1'$ on $g$. Let $\pi(z_1',a_1)$ be the portion of $\pi(z',a_1)$ between $z_1'$ and $a_1$.
Consider the ``pseudo-triangular'' region $R$ enclosed by $\overline{a_1z_1}$,
$\overline{z_1z_1'}$, and $\pi(z_1',a_1)$. Since $\pi(z_1',a_1)$ is convex, every point on
$\pi(z_1',a_1)$ is horizontally visible to $g$.

We claim that there exists a shortest path $\pi$ from $z'$ to $p$ that intersects $\overline{z_1'z_1}$. Indeed, if $z'=z_1'$, then the claim is trivially true. Otherwise, since $z_1'$ is the first point on $\pi(z',a_1)$ that is horizontally visible to $g$ if we go from $z'$ to $a_1$ along $\pi(z',a_1)$, $z'$ cannot be horizontally visible to $g$, and thus, $z'$ is not in $R$. Note that $\overline{z_1z_1'}$ partitions the funnel $F$ into two parts, one of which is $R$. Also,
the funnel $F$ contains a shortest path $\pi$ from $z'$ to $p$. Since $p\in R$ and $z'\not\in R$, the path $\pi$ must intersect $\overline{z_1'z_1}$. The claim is proved.

Suppose $\pi$ intersects $\overline{z_1'z_1}$ at a point $q$. Since $\overline{qz_1}\cup \overline{z_1p}$ is $xy$-monotone (and thus is a shortest path), we can obtain another shortest path from $z'$ to $p$ that contains $z_1$ by replacing the portion of $\pi$ between $q$ and $p$ by $\overline{qz_1}\cup \overline{z_1p}$. The lemma thus follows.
\end{proof}

For the case of $t\in \calM$,
Lemma \ref{lem:100} implies the following: If a shortest \st\ path
crosses $g$ at a point on $\overline{a_1z_1}$ (resp.,
$\overline{a_2z_2}$), then there must be a shortest \st\ path that is a
concatenation of a shortest path $\pi(s,z_1)$ (resp., $\pi(s,z_2)$) from $s$ to $z_1$ (resp., $z_2$) in
$B$ and a shortest path $\pi(z_1,t)$ from $z_1$ (resp., $\pi(z_2,t)$ from $z_2$) to $t$ in $\calM'$.
The path $\pi(s,z_1)$ can be found using the data structure $\calD(B)$
and $\pi(z_1,t)$ can be found by Lemma \ref{lem:80} since both $z_1$ and $t$ are in $\calM$.
Hence, such a shortest \st\ path query is answered in $O(\log n)$ time, provided that
we can find $z_1$ and $z_2$ in $O(\log n)$ time (as to be shown in Lemma \ref{lem:170}).

In the following, we assume every shortest \st\ path crosses the interior of $\overline{z_1z_2}$, and in other words, no shortest \st\ paths cross $\overline{a_1z_1}\cup \overline{a_2z_2}$.

Let $z$ denote the
intersection of the horizontal line containing $\overline{z_1z_1'}$ and the
vertical line containing $\overline{z_2z_2'}$ (see Fig.~\ref{fig:funnel2}).
The point $z$ is useful as shown by the next lemma.

\begin{lemma}\label{lem:110}
The point $z$ is in the funnel $F$, and for any point $p\in
\overline{z_1z_2}$, there is a shortest path
from $s$ to $p$ that contains $z$.
\end{lemma}
\begin{proof}
We first prove $z\in F$. For this, it suffices to prove that the interior of the triangle
$\triangle zz_1z_2$ does not contain any point on the boundary of $F$.
Let $R$ denote the interior of $\triangle zz_1z_2$.

Assume to the contrary that $R$ intersects $\partial F$. Let $q$ be any point in
$R\cap \partial F$ that is horizontally visible to $\overline{z_1z_2}$. Such a point $q$
always exists if $R\cap\partial F\neq \emptyset$. Note that $q$ is on either $\pi(z',a_1)$
or $\pi(z',a_2)$.
Without loss of generality, assume $q$ is on $\pi(z',a_1)$.
Observe that $\pi(z_1',a_1)$ is $xy$-monotone since $z_1'$ is horizontally visible to $g$.
Because $q$ is also horizontally visible to $g$, by the definition of $z_1'$, $q$ must be
on $\pi(z_1',a_1)$. Since $q$ is in $R$, $q$ must be strictly below $z_1'$. Since $a_1$ is no
lower than $z_1$, $a_1$ is also no lower than $z_1'$. Thus, when following the path
$\pi(z_1',a_1)$ from $z_1'$ to $a_1$, we have to strictly go down (through $q$) and then go up
(to $a_1$), which contradicts with that the fact the path $\pi(z_1',a_1)$ is
$xy$-monotone. Hence, $R$ cannot contain any point on $\partial F$ and $z$ must be in $F$.

Consider any point $p\in \overline{z_1z_2}$. Below we prove that there is a shortest path from
$s$ to $p$ containing $z$. It suffices to show that there exists a shortest path from $z'$
to $p$ containing $z$. If $z_1=z_2$, then $z=z_1=z_2=p$ and we are done. Below we assume
$z_1\neq z_2$, which implies $z_1\neq a_2$ since otherwise $z_1=z_2$ by Observation
\ref{obser:10}; similarly, $z_2\neq a_1$. Note that $z_1\neq z_2$ also implies
$z'\not\in \{a_1,a_2\}$.

Let $\pi(z',p)$ be a shortest path in $F$ from $z'$ to $p$. Let $l_1$ be the horizontal line containing $\overline{z_1z_1'}$ and $l_2$ be the vertical line containing $\overline{z_2z_2'}$.

In the following, we first prove that $l_1\cap F$ is a line segment and it must intersect the path $\pi(z',p)$.
Consider the line segment $\overline{z_1z_1'}$. Depending on whether $l_1$ is tangent
to $\pi(z',a_1)$ at $z_1'$, there are two possible cases (e.g., see Fig.~\ref{fig:funnel2}).

\begin{enumerate}
\item
If $l_1$ is tangent to $\pi(z',a_1)$ at $z_1'$ (see Fig.~\ref{fig:funnel2}(a)), then we extend
$\overline{z_1z_1'}$ horizontally leftwards until it hits $\partial F$, say, at a point $z_1''$. Since $\pi(z',a_1)$ is convex, $z'$ is above the line $l_1$ and $z_1''$ is on $\pi(z',a_2)$. Since $\pi(z',a_2)$ is also convex and $z'$ is above $l_1$, we obtain $l_1\cap F=\overline{z_1z_1''}$.

Observe that $\overline{z_1'z_1''}$ partitions $F$ into two sub-polygons such that $z'$ and
$p$ are in different sub-polygons. Hence, the path $\pi(z',p)$ must intersect
$\overline{z_1'z_1''}\subseteq l_1\cap F$, which is a line segment.

\item
If $l_1$ is not tangent to $\pi(z',a_1)$ at $z_1'$, then depending on whether $z_1'=z'$, there are two subcases.

\begin{enumerate}
\item
If $z_1'=z'$, then due to the convexity of $\pi(z',a_1)$ and $\pi(z',a_2)$, we have $l_1\cap F=\overline{z_1z_1'}$. Since $z'=z_1'$, it is trivially true that $\pi(z',p)$ intersects $l_1\cap F=\overline{z_1z_1'}$.

\item
If $z_1'\not=z'$ (see Fig.~\ref{fig:funnel2}(b)), then we claim that $\overline{z_1z_1'}$ must be tangent to $\pi(z',a_2)$ at a point, say, $z_1''$. Suppose to the contrary that this is not
the case. Then, since $z_1'\neq z'$, $z_1\neq a_2$, and $l_1$ is not tangent to $\pi(z',a_1)$ at $z_1'$, we can move $l_1$ downwards by an infinitesimal value such that the new $l_1$ intersects
$g$ at a point $z_3$ and intersects $\pi(z',a_1)$ at a point $z_3'$ such that $z_3'$ is horizontally visible to $z_3$. Clearly, $z_3'$ is on $\pi(z',a_1)$ between $z'$ and $z_1'$.
But this contradicts with the definition of $z_1'$, i.e., $z_1'$ is the first point on $\pi(z',a_1)$ horizontally visible to $g$ if we go from $z'$ to $a_1$ along $\pi(z',a_1)$. The claim is thus proved.

By the above claim and the convexity of $\pi(z',a_2)$, $z'$ is below $l_1$. Also by the
convexity of $\pi(z',a_1)$, we have $l_1\cap F=\overline{z_1z_1'}$. Further, observe that $\overline{z_1'z_1''}$ partitions $F$ into two sub-polygons such that $z'$ and $p$ are in different sub-polygons. Hence, the path $\pi(z',p)$ must intersect $\overline{z_1'z_1''}\subseteq l_1\cap F$.

Therefore, $l_1\cap F$ is a line segment that intersects $\pi(z',p)$.
\end{enumerate}
\end{enumerate}

The above arguments prove that $l_1\cap F$ is a line segment that intersects the path $\pi(z',p)$, say, at a point $q_1$. By using a
similar analysis, we can also show that $l_2\cap F$ is a line segment that intersects $\pi(z',p)$, say, at a point $q_2$. Note that this implies that $z$ is on the intersection of the segment $l_1\cap F$ and the segment $l_2\cap F$.
Since $\overline{q_1z}\cup \overline{zq_2}$ is $xy$-monotone (and thus is a shortest path),
if we replace the subpath of $\pi(z',p)$ between $q_1$ and $q_2$
by $\overline{q_1z}\cup \overline{zq_2}$ to obtain another path
$\pi'(z',p)$ from $z'$ to $p$, then $\pi'(z',p)$ is still a shortest path. Since $\pi'(z',p)$ contains $z$, the lemma follows.
\end{proof}

If there is a shortest \st\ path crossing $g$ at a point on $\overline{z_1z_2}$,
then by Lemma \ref{lem:110}, there is a shortest \st\ path that is a
concatenation of a shortest path from $s$ to $z$ in $B$ and a
shortest path from $z$ to $t$ (which crosses $g$). A shortest $s$-$z$ path in
$B$ can be found by using the data structure $\calD(B)$ in $O(\log n)$ time,
provided that we can compute $z$ in $O(\log n)$
time. It remains to show how to compute a shortest $z$-$t$ path that
crosses $g$ at a point on $\overline{z_1z_2}$. Note that
such a shortest $z$-$t$ path either does or does not cross a point in
$\calV(g)\cap \overline{z_1z_2}$, where $\calV(g)$ is the set of
points of $\calV(\calQ)$ lying on $g$
($\calV(g)=\emptyset$ is possible).
For the former case (when $\calV(g)\neq \emptyset$ holds), we shall
build a graph $G_E(g)$ inside $B$ and merge it with the graph $G_E(\calM)$ on $\calM$ so that
the merged graph allows to find a shortest $z$-$t$ path crossing a point in
$\calV(g)\cap \overline{z_1z_2}$. Next, we introduce the graph $G_E(g)$.

Let $h_g=|\calV(g)|$. The graph $G_E(g)$ is defined on the points
of $\calV(g)$ in a similar manner as $G_E$ in Section \ref{sec:newgraph}. One big difference is that $G_E(g)$ is built inside $B$ and uses vertical {\em cut-segments} in $B$
instead of cut-lines. Also, no type-1 Steiner point is needed for $G_E(g)$.
Specifically, we define a {\em cut-segment tree} $T(g)$ as follows. The root $u$ of $T(g)$ is associated with a point set $V(u)=\calV(g)$. Each node $u$ of $T(g)$ is also associated with a vertical {\em cut-segment}
$l(u)$, defined as follows. Let $p$ be the point of $V(u)$ that has the median $x$-coordinate among all
points in $V(u)$. Note that $p$ is on $g$. We extend a vertical line
segment from $p$ upwards into the interior of $B$ until it hits $\partial B$; this segment is the cut-segment $l(u)$.
The left (resp., right) child of $u$ is defined recursively on the points of $V(u)$ to the left (resp., right) of $l(u)$.

Clearly, $T(g)$ has $O(\log h_g)$ levels and $O(\sqrt{\log h_g})$
super-levels. We define the type-2 and type-3 Steiner points on the cut-segments of $T(g)$
in the same way as in Section \ref{sec:newgraph}.
Consider a super-level and let $u$ be any node at the highest level of this super-level. For
every $p\in V(u)$, for each cut-segment $l$ in the subtree $T_u(g)$ of $T(g)$ in the same super-level,
if $p$ is horizontally visible to $l$, then the horizontal projection $p_h(l)$ of $p$ on $l$ is defined as a Steiner point on $l$; we order the Steiner points defined by $p$ from left to right,
and put an edge in $G_E(g)$ connecting every two such consecutive Steiner points.
Hence, there are $O(h_g\sqrt{\log h_g}2^{\sqrt{\log h_g}})$ Steiner points on all cut-segments
of $T(g)$. The above process also defines $O(h_g\sqrt{\log h_g}2^{\sqrt{\log h_g}})$ edges in $G_E(g)$.

The node set of $G_E(g)$ consists of all points of $\calV(g)$ and all Steiner points on the
cut-segments of $T(g)$. In addition to the graph edges defined above, for each cut-segment $l$,
a graph edge connects every two consecutive graph nodes on $l$
(note that here every two such graph nodes are visible to each other).
Clearly, $G_E(g)$ has $O(h_g \sqrt{\log h_g}2^{\sqrt{\log h_g}})$ nodes and  $O(h_g \sqrt{\log h_g}2^{\sqrt{\log h_g}})$ edges.

Let $n_B$ denote the number of obstacle vertices of the bay $B$.
Note that $\calV(g)$ is sorted along $g$.

\begin{lemma}\label{lem:120}
The graph $G_E(g)$ can be constructed in
$O(n_B+h_g\cdot \log^{3/2} h_g \cdot 2^{\sqrt{\log h_g}})$ time.
\end{lemma}
\begin{proof}
To compute the cut-segments of $T(g)$, for each point $p\in \calV(g)$, we need to
compute the first point on the boundary of $B$ hit by extending a
vertical line segment from $p$ upwards.
For this, we first compute the vertically visible region of $B$ from the segment $g$
using the linear time algorithms in \cite{ref:JoeCo87,ref:LeeVi83}, and then
find all such cut-segments from the points of $\calV(g)$, in $O(n_B+h_g)$ time.
The cut-segment tree $T(g)$ can then be computed in
$O(h_g\log h_g)$ time.

To compute the Steiner points on the cut-segments, for each point
$p\in \calV(g)$, we find the first point $p_h(B)$ on the boundary of $B$
horizontally visible from $p$.
The points $p_h(B)$ for all $p\in \calV(g)$ can be computed in totally
$O(n_B+h_g)$ time by using the algorithms in \cite{ref:JoeCo87,ref:LeeVi83}.

Next, we compute the Steiner points on the cut-segments of $T(g)$.
Determining whether a point $p\in \calV(g)$ is
horizontally visible to a cut-segment $l$ (and if yes, put a
corresponding Steiner point on $l$) takes $O(1)$ time using
$p_h(B)$, as follows. We first check whether the $y$-coordinate of $p$ is between the
$y$-coordinate of the lower endpoint of $l$ and that of the upper
endpoint of $l$; if yes, we check whether $l$ is between $p$ and
$p_h(B)$ (if yes, then $p$ is horizontally visible to $l$);
otherwise, $p$ is not horizontally visible to $l$.
Thus, all Steiner
points can be obtained in $O(h_g \sqrt{\log h_g}2^{\sqrt{\log h_g}})$ time.

For each cut-segment $l$, to compute the edges between consecutive graph nodes
on $l$, it suffices to sort all Steiner points on
$l$. The sorting on all cut-segments takes $O(h_g\cdot \log^{3/2}h_g\cdot
2^{\sqrt{\log h_g}})$ time.

Hence, the total time for building the graph $G_E(g)$ is $O(n_B+h_g\cdot
\log^{3/2}h_g\cdot 2^{\sqrt{\log h_g}})$.
The lemma thus follows.
\end{proof}


We define a gateway set $V_g(z,G_E(g))$ for $z$ on $G_E(g)$ such that
for any point $p\in \calV(g)\cap \overline{z_1z_2}$,
there is a shortest path from $z$ to
$p$ using $G_E(g)$ containing a gateway of $z$.
$V_g(z,G_E(g))$ is defined similarly as $V^2_g(s,G_E)$ in Section \ref{sec:newgraph},
but only on the Steiner points in the triangle $\triangle zz_1z_2$
(because $\triangle zz_1z_2$ contains a shortest path
from $z$ to any point in $\calV(g)\cap \overline{z_1z_2}$). Specifically, for each
{\em relevant projection cut-segment} $l$  (defined similarly as the
relevant projection cut-lines in Section \ref{sec:newgraph})
of $z$ to the right of $z$, if $z$ is horizontally visible
to $l$, then the node of $G_E(g)$ on $l$ immediately below the horizontal projection
point of $z$ on $l$ is in $V_g(z,G_E(g))$. Thus,
$|V_g(z,G_E(g))|=O(\sqrt{\log h_g})$.

\begin{lemma}\label{lem:130}
For any point $p\in \calV(g)\cap \overline{z_1z_2}$,
there is a shortest path from $z$ to
$p$ in $B$ using $G_E(g)$ that contains a gateway of $z$ in $V_g(z,G_E(g))$.
\end{lemma}
\begin{proof}
Consider a point $p\in \calV(g)\cap \overline{z_1z_2}$. Note that
$p$ defines a node in $G_E(g)$. Let $l_p$ be the cut-segment
through $p$. Since the triangle $\triangle zz_1z_2\subseteq B$ and
$p\in \overline{z_1z_2}$, $z$ is horizontally visible to $l_p$.

If there is no other cut-segment of $T(g)$ strictly
between $z$ and $l_p$, then $l_p$ must be a relevant projection
cut-segment of $z$. Let $p'$ be the gateway of $z$ on $l_p$, i.e., the
graph node on $l_p$ immediately below the horizontal
projection $z_h(l_p)$ of $z$ on $l_p$.
Note that the path
$\overline{zz_h(l_p)}\cup \overline{z_h(l_p)p}$ is a shortest path from $z$
to $p$ since it is $xy$-monotone. Clearly, this path contains
the gateway $p'$.

If there is at least one cut-segment strictly between $z$ and $l_p$,
then if $l_p$ is a relevant cut-segment of $z$, we can prove the
lemma by a similar analysis as above; otherwise,
there is at least one node $u$ in $T(g)$ such that
$l(u)$ is a relevant projection cut-segment of $z$ between $z$ and $p$
and $p$ defines a Steiner point on $l(u)$ (this can be seen
from the definition of the graph $G_E(g)$; we omit the
details). Let $z_h(l(u))$ be the horizontal projection
of $z$ on $l(u)$ and $p_h(l(u))$ be the horizontal projection
of $p$ on $l(u)$. The path
$\overline{zz_h(l(u))}\cup\overline{z_h(l(u))p_h(l(u))}\cup
\overline{p_h(l(u))p}$ is a
shortest path from $z$ to $p$ since it is $xy$-monotone. Because
$p_h(l(u))$ is a Steiner point on $l(u)$, this
path must contain a gateway of $z$ on $l(u)$ (this gateway must be
on $\overline{z_h(l(u))p_h(l(u))}$).
The lemma thus follows.
\end{proof}

Since $\calV(g)\subseteq \calV(\calQ)$, each
point of $\calV(g)$ is also a node of $G_E(\calM)$. We merge
the two graphs $G_E(\calM)$ and $G_E(g)$ into one graph, denoted by $G_E(\calM,g)$,
by treating the two nodes in these two graphs
defined by the same point in $\calV(g)$ as a single node.
By Lemmas \ref{lem:60} and
\ref{lem:130}, we have the following result.

\begin{lemma}\label{lem:140}
If a shortest \st\ path contains a point in $\calV(g)\cap \overline{z_1z_2}$,
then there is a shortest \st\ path along $G_E(\calM,g)$ containing
a gateway of $z$ in $V_g(z,G_E(g))$ and a gateway of $t$ in
$V_g(t,G_E(\calM))$.
\end{lemma}
\begin{proof}
Let $p$ be a point of $\calV(g)\cap \overline{z_1z_2}$ that is
contained in a shortest \st\ path. By Lemma \ref{lem:110}, there is a
shortest path from $s$ to $p$ that contains $z$. By Lemma
\ref{lem:130}, there is a shortest path from $z$ to $p$ that contains
a gateway of $z$ in $V_g(z,G_E(g))$. On the other hand, since both $t$
and $p$ are in the ocean $\calM$ and $p\in \calV(g)\subseteq \calV(\calQ)$,
by Lemma \ref{lem:60}, there exists
a shortest path from $t$ to $p$ that contains a gateway of $t$ in
$V_g(t,G_E(\calM))$.  This proves the lemma.
\end{proof}

By Lemma \ref{lem:140}, if there is a shortest path from $z$ to $t$
that contains a
point of $\calV(g)\cap \overline{z_1z_2}$, then we can use the gateways
of both $z$ and $t$ to find a shortest path along the graph $G_E(\calM,g)$.
By using a similar algorithm as that for Lemma \ref{lem:30}, we can
compute the gateways of $z$ on $G_E(g)$.

\begin{lemma}\label{lem:150}
With a preprocessing of $O(h_g\log^{3/2}h_g2^{\sqrt{\log h_g}})$ time and
$O(h_g\sqrt{\log h_g}2^{\sqrt{\log h_g}})$ space,  we
can compute the gateway set $V_g(z,G_E(g))$ of $z$ in $O(\log h)$ time.
\end{lemma}
\begin{proof}
The algorithm is similar to that in Lemma \ref{lem:30} for computing
$V^2_g(s,G_E)$. One main difference is that here every two graph
nodes on any cut-segment of $T(g)$ are visible to each other.
As the preprocessing, we build a sorted list of the graph nodes on each
cut-segment of $T(g)$, and construct a fractional cascading data
structure \cite{ref:ChazelleFr86} along $T(g)$ for the sorted lists of all
cut-segments. Then for a point
$z$, $V_g(z,G_E(g))$ can be computed in $O(\log h)$ time.
\end{proof}

So far, we have shown how to find a shortest \st\ path if such a path
contains a point in $\{z_1,z_2\}\cup \{\calV(g)\cap \overline{z_1z_2}\}$.
It remains to handle the case when
no shortest \st\ path contains a point in
$\{z_1,z_2\}\cup \{\calV(g)\cap \overline{z_1z_2}\}$ (including the case of
$\calV(g)=\emptyset$), i.e.,
no shortest path from $z$ to $t$ contains a point in
$\{z_1,z_2\}\cup \{\calV(g)\cap \overline{z_1z_2}\}$.
Lemma \ref{lem:160} below shows that in this case, $t\in \calM$ must be horizontally
visible to $\overline{zz_2}$ and thus there is a trivial shortest
path from $z$ to $t$.

\begin{lemma}\label{lem:160}
If no shortest path $\pi(z,t)$ contains a point in
$\calV(g)\cap \overline{z_1z_2}$ (this includes the case of $\calV(g)=\emptyset$), then
$t$ must be horizontally visible to $\overline{zz_2}$.
\end{lemma}
\begin{proof}
Let the points of $\calV(g)\cap \overline{z_1z_2}$ be
$v_1,v_2,\ldots,v_m$ ordered along $\overline{z_1z_2}$ from $z_1$ to
$z_2$, and let $v_0=z_1$ and $v_{m+1}=z_2$. Under the condition of this lemma, since
$t \in \calM$, there exists a shortest path
$\pi$ from $z$ to $t$ that crosses $\overline{z_1z_2}$ once, say, at a point $p$ in
the interior of $\overline{v_iv_{i+1}}$,
for some $i$ with $0\leq i\leq m$ (see Fig.~\ref{fig:visibility}).
For any two points $q_1$ and $q_2$
on $\pi$, let $\pi(q_1,q_2)$ denote the subpath of $\pi$ between
$q_1$ and $q_2$. Hence, $\pi(z,p)$ is in $B$ and $\pi(p,t)$ is outside $B$. Then
$\pi(p,t)$ is in $\calM'$ (i.e., $\calM'$ is the union of $\calM$ and all corridor paths).

\begin{figure}[t]
\begin{minipage}[t]{\linewidth}
\begin{center}
\includegraphics[totalheight=1.5in]{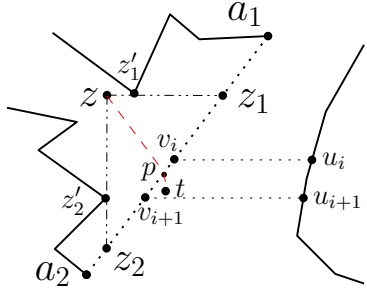}
\caption{\footnotesize Illustrating a shortest path (the red dashed curve) from
$z$ to $t$ crossing the interior of $\overline{v_iv_{i+1}}$ at $p$.}
\label{fig:visibility}
\end{center}
\end{minipage}
\vspace*{-0.15in}
\end{figure}

%
%
We extend a horizontal line segment from $v_i$ (resp., $v_{i+1}$) to the right until
hitting the first point on $\partial\calQ$, denoted by $u_i$ (resp., $u_{i+1}$);
if $u_i$ and $u_{i+1}$ are not on the same elementary curve of
$\calQ$ (in which case
one or both of $u_i$ and $u_{i+1}$ are extremes on different
elementary curves),
then we keep moving one or both of $u_i$ and $u_{i+1}$ horizontally
to the right until hitting the next point on $\partial\calQ$.
By the definitions of $\calV(\calQ)$ and $\calV(g)$, in this way,
we can always put both $u_i$ and $u_{i+1}$ on the same
elementary curve of $\calQ$, say $\beta$ (see Fig.~\ref{fig:visibility}); let
$\beta(u_i,u_{i+1})$ denote the
portion of $\beta$ between $u_i$ and $u_{i+1}$. Let $R$ denote the
region enclosed by $\overline{u_iv_i}$, $\overline{v_iv_{i+1}}$,
$\overline{v_{i+1}u_{i+1}}$, and $\beta(u_i,u_{i+1})$. Note that for
any point $q\in R$, $q$ is horizontally visible to
$\overline{v_iv_{i+1}}$ and thus is horizontally visible to
$\overline{zz_2}$.
In the following, we will show that $t$ must be in $R$, which proves the lemma.

Suppose to the contrary $t\not\in R$. We then show that
the path $\pi(p,t)$ must intersect $\overline{v_iu_{i}}$ or
$\overline{v_{i+1}u_{i+1}}$, which implies that
there is a shortest $z$-$t$ path containing a point in
$\{z_1,z_2\}\cup \{\calV(g)\cap \overline{z_1z_2}\}$,
a contradiction (recall that we have an assumption that no shortest \st\ paths cross $\overline{a_1z_1}\cup \overline{a_2z_2}$). Indeed, if $\pi(p,t)$ intersects $\overline{v_iu_{i}}$
(resp., $\overline{v_{i+1}u_{i+1}}$), say, at a point $q$, then we can obtain a new
$z$-$t$ path $\pi'$ by replacing $\pi(z,q)$
with an $xy$-monotone path $\overline{zv_i}\cup\overline{v_iq}$
(resp., $\overline{zv_{i+1}}\cup\overline{v_{i+1}q}$), and $\pi'$ is
a shortest $z$-$t$ path containing
a point in $\{z_1,z_2\}\cup \{\calV(g)\cap \overline{z_1z_2}\}$.
Below, we show that $\pi(p,t)$ must intersect $\overline{v_iu_{i}}$ or
$\overline{v_{i+1}u_{i+1}}$. Note that $\beta(u_i,u_{i+1})$ may
overlap with a gate of a canal. Depending on whether $\beta(u_i,u_{i+1})$
overlaps with any canal gate, there are two possible cases.

\begin{enumerate}
\item
If $\beta(u_i,u_{i+1})$ does not overlap with any canal gate, then
since $t\in \calM$, $t\not\in R$, $p\in R$, and $\pi(p,t)\subseteq \calM'$, if we go
from $t$ to $p$, we must enter $R$.  The only place on the boundary of $R$ we can cross to
enter $R$ is either $\overline{v_iu_i}$ or
$\overline{v_{i+1}u_{i+1}}$. Hence, $\pi(t,p)$ must intersect $\overline{v_iu_{i}}$ or
$\overline{v_{i+1}u_{i+1}}$.

%
%
%
\item
If $\beta(u_i,u_{i+1})$ overlaps with a canal gate, say $g_1$, then one may wonder
that $\pi(t,p)$ could enter the interior of $R$ through $g_1$ without
crossing any of $\overline{v_iu_i}$ and $\overline{v_{i+1}u_{i+1}}$.
Since $g_1$ is a canal gate, one of $g_1$'s endpoints, say, $x$, must be a corridor path
terminal, and $x$ may or may not be on $\beta(u_i,u_{i+1})$.
If $x$ is on $\beta(u_i,u_{i+1})$, then since $x$ is in $\calV(\calQ)$,
$x$ cannot be in the interior of $\beta(u_i,u_{i+1})$ and can only be at an endpoint
of $\beta(u_i,u_{i+1})$. Let $C$ be the canal that has $g_1$ as a
gate, and $\pi(C)$ be the corridor path of $C$. If $\pi(t,p)$
enters the interior of $R$ through $g_1$, then it must travel through the canal $C$,
implying that $\pi(t,p)\subseteq \calM'$ contains the corridor path $\pi(C)$.
Since $x$ is on $\pi(C)$, $\pi(t,p)$ contains $x$.
If $x$ is on $\beta(u_i,u_{i+1})$ (and thus is an endpoint of $\beta(u_i,u_{i+1})$),
then $x$ is one of $u_i$ or $u_{i+1}$; hence, $\pi(t,p)$ intersects $\overline{v_iu_{i}}$ or
$\overline{v_{i+1}u_{i+1}}$.
Suppose now $x$ is not on $\beta(u_i,u_{i+1})$.  Then an endpoint of $\beta(u_i,u_{i+1})$,
say, $u_i$, lies on $g_1$ (but $u_i\not=x$). Further, $\pi(t,p)$ goes through $x$, and then
enters $R$, but without intersecting any of $\overline{v_iu_i}$
and $\overline{v_{i+1}u_{i+1}}$.  Thus, $\pi(t,p)$ must cross some point $q$ of $g_1$
to enter $R$.  We can then replace the portion $\pi(x,q)$ of $\pi(t,p)$ by the segment
$\overline{xq}\subseteq g_1$ to obtain a new shortest $t$-$p$ path.
Since $u_i$ divides $g_1$ into two parts, one outside $R$ and
containing $x$ and the other intersecting $R$ and containing $q$, the segment
$\overline{xq}$ contains $u_i$.  Hence, the new shortest $t$-$p$ path
intersects $\overline{v_iu_{i}}$.

\end{enumerate}

The lemma thus follows.
\end{proof}

By Lemma \ref{lem:160}, if the condition of the lemma
holds, then we can always find a trivial shortest path from $z$ to
$t$ by shooting vertical and horizontal rays from $z$ and $t$, respectively.

We have finished all possible cases for finding a shortest \st\ path
when $s\in B$ and $t\in \calM$.
The next lemma is concerned with computing the special
points $z_1,z_2$, and $z$ for any point $s$ in $B$.

\begin{lemma}\label{lem:170}
With a preprocessing of $O(n_B)$ time and space,
the three special points $z_1$, $z_2$, and
$z$ can be found in $O(\log n)$ time for any query point $s$ in $B$, where $n_B=|B|$.
\end{lemma}
\begin{proof}
Consider any query point $s\in B$. To determine $z_1$, $z_2$, and
$z$, based on our previous discussions, it suffices to compute the two points $z_1'$ and $z_2'$.
We only show how to design a data structure for computing $z_1'$ since the solution for finding $z_2'$ is similar.
Note that $n_B\leq n$.

In the preprocessing, for each vertex $v$ of $B$, we find whether $v$ is horizontally
visible to $g$, and if yes, mark $v$ as an {\em h-vertex}. All h-vertices of $B$ can be marked
by computing the horizontal visibility of $B$ from $g$ in $O(n_B)$ time \cite{ref:JoeCo87,ref:LeeVi83}.
Also, in $O(n_B)$ time, we compute the Euclidean shortest path
tree $T_1$ from $a_1$ to all vertices of $B$ and the corresponding shortest path map
$M_1$ in $B$ \cite{ref:GuibasLi87}; similarly, we compute the shortest path
tree $T_2$ from $a_2$ and the corresponding shortest path map $M_2$.

For each vertex $v\in T_1$, we associate $v$ with two special vertices:
$\alpha_1(v)$ and $\alpha_2(v)$, defined as follows. The vertex $\alpha_1(v)$ is the first
h-vertex on the path in $T_1$ from $v$ to $a_1$ and $\alpha_2(v)$ is the child
vertex of $\alpha_1(v)$ on the path in $T_1$ from $v$ to $a_1$; if $\alpha_1(v)=v$,
then $\alpha_2(v)$ does not exist and we set $\alpha_2(v)=nil$. Note that $\alpha_2(v)$ is
not an h-vertex if it exists. The $\alpha$ vertices for all vertices in $T_1$ can be computed in $O(n_B)$ time by a depth-first search on $T_1$ starting at $a_1$.
For each vertex $v\in T_2$, we compute only one special vertex for $v$, $\beta_1(v)$, which
is the first h-vertex on the path in $T_2$ from $v$ to $a_2$. The $\beta$ vertices for all vertices of $T_2$ can also be computed in $O(n_B)$ time.

This finishes our preprocessing, which takes $O(n_B)$ time in total.

Below we find the point $z_1'$ in $O(\log n)$ time.
Let $\pi(s,a_1)$ and $\pi(s,a_2)$ be the Euclidean shortest paths in $B$
from $s$ to $a_1$ and $a_2$, respectively. For any point $p$, let $y(p)$ denote its $y$-coordinate.

By using the shortest path map $M_1$, we find the vertex, denoted by $v$, which directly connects
to $s$ on $\pi(s,a_1)$. Likewise, we find the vertex $u$ that directly connects to $s$ on
$\pi(s,a_2)$ using $M_2$. Both $v$ and $u$ are found in $O(\log n)$ time. Depending on whether $v=u$, there are two main cases.

\begin{enumerate}
\item
If $v=u$, then clearly $s\neq z'$. Let $v_1=\alpha_1(v)$ and $u_1=\beta_1(u)$. Note that $v_1$
and $u_1$ are available once we find $v$ and $u$. Depending on whether $v_1=u_1$, we further have two subcases.

\begin{enumerate}
\item If $v_1=u_1$, then we claim $z'=v_1=u_1$. Indeed, since $z'$ is the last common vertex
of $\pi(s,a_1)$ and $\pi(s,a_2)$ if we move on them from $s$, no vertex on
$\pi(s,a_1)\cap \pi(s,a_2)$ can be horizontally visible to $g$ except possibly $z'$.
Because $v_1=u_1$, $v_1=u_1$ must be on $\pi(s,a_1)\cap \pi(s,a_2)$. Since $v_1=u_1$ is
horizontally visible to $g$, $v_1=u_1=z'$ must hold.

    By the definition of $z_1'$, the above claim implies $z_1'=z'=u_1=v_1$.

\item If $v_1\neq u_1$, then an easy observation is $y(v_1)\geq y(u_1)$. Let $v_2=\alpha_2(v)$. Note that due to $u=v$ and $v_1\neq u_1$, $\alpha_2(v)$ exists.

    If $y(v_2)>y(v_1)$, then the horizontal visibility of $v_2$ to $g$ is ``blocked'' by the
path $\pi(v_1,a_1)$ (e.g., see Fig.~\ref{fig:funnel2}(a)). Thus we obtain $z_1'=v_1$.

    If $y(v_2)\leq y(v_1)$, then the horizontal visibility of $v_2$ to $g$ is ``blocked'' by
the path $\pi(u_1,a_2)$ (e.g., see Fig.~\ref{fig:funnel2}(b)). Thus we obtain that $z_1'$ is
the horizontal projection of $u_1$ on the line segment $\overline{v_1v_2}$, which can be computed in $O(1)$ time.
\end{enumerate}

\item
If $v\neq u$, then $s=z'$.
%
%
%
If $s$ is horizontally visible to $g$ (which can be determined in $O(\log n)$ time
using the horizontal visibility decomposition of $B$), then $z_1'=s=z'$.  Otherwise,
let $v_1=\alpha_1(v)$ and $u_1=\beta_1(u)$. Depending on whether $v=v_1$, we further have two subcases.

\begin{enumerate}
\item If $v\neq v_1$, then $\alpha_2(v)$ exists and we let $v_2=\alpha_2(v)$.
Note that $\pi(s,a_1)$ is a convex chain.

    Similar to the above discussion, if $y(v_2)>y(v_1)$, then we have $z_1'=v_1$; otherwise, $z_1'$ is the horizontal projection of $u_1$ on $\overline{v_1v_2}$.

\item
If $v=v_1$, then $s$ connects directly to $v_1$ on $\pi(s,a_1)$.
Similar to the above discussion, if $y(s)>y(v_1)$, then we have $z_1'=v_1$; otherwise, $z_1'$ is the horizontal projection of $u_1$ on $\overline{v_1s}$.
\end{enumerate}
\end{enumerate}

Therefore, we can find the point $z_1'$ in $O(\log n)$ time. The lemma thus follows.
\end{proof}

We have discussed all possible cases of finding a shortest \st\ path
when $s$ is in a bay $B$ and $t$ is in the ocean $\calM$, and in each case,
we can obtain a shortest path in $O(\log n)$ time.

\subsubsection{The Point $t$ is in Another Bay}

Let $B_s$ be the bay containing $s$ with gate $g_s$, and
$B_t$ be the bay containing $t$ with gate $g_t$.
In this case, any shortest \st\ path must cross both $g_s$ and $g_t$. The
algorithm for this case is similar to the one for the case of $t\in \calM$. Again, we need to
consider different cases of how a shortest \st\ path may cross
different portions of both the gates $g_s$ and $g_t$.

We define the points $z_1$, $z_2$, and $z$ in $B_s$ for $s$ in the
same way as
before, but denote them by $z_1(s)$, $z_2(s)$, and $z(s)$ instead.
Similarly, we define the corresponding three points $z_1(t)$,
$z_2(t)$, and $z(t)$ in $B_t$ for $t$. Based on our previous
discussions, we have the following cases.

\begin{enumerate}
\item
There is a shortest \st\ path containing a point $z_s$ in
$\{z_1(s),z_2(s)\}$ and a point $z_t$ in $\{z_1(t),z_2(t)\}$.
Note that both $z_s$ and $z_t$ are on their bay gates and thus are in $\calM$.

In this case, there must be a shortest \st\ path that is a
concatenation of a shortest path $\pi(s,z_s)$ from $s$ to $z_s$ in
$B_s$, a shortest path $\pi(z_s,z_t)$ from $z_s$ to $z_t$ in $\calM'$,
and a shortest path $\pi(z_t,t)$ from $z_t$ to $t$ in $B_t$. The path
$\pi(s,z_s)$ can be found by using $\calD(B_s)$, i.e., the Euclidean
two-point shortest path query data structure on $B_s$
\cite{ref:GuibasOp89}, and similarly, $\pi(z_t,t)$ can be found by
using $\calD(B_t)$. The path $\pi(z_s,z_t)$ can be found by using our data
structure for $\calM'$ in Lemma \ref{lem:80}.

\item
There is a shortest \st\ path that contains $z(s)$
and a point $z_t$ in $\{z_1(t),z_2(t)\}$.

In this case, there must be a shortest \st\ path that is a
concatenation of a shortest $s$-$z(s)$ path $\pi(s,z(s))$ in
$B_s$, a shortest $z(s)$-$z_t$ path $\pi(z(s),z_t)$, and a shortest $z_t$-$t$ path
$\pi(z_t,t)$ in $B_t$.
The path $\pi(s,z(s))$ (resp., $\pi(z_t,t)$) can be found by using $\calD(B_s)$
(resp., $\calD(B_t)$), and the path
$\pi(z(s),z_t)$ can be found by using similar algorithms
as discussed above since $z_t$ is in $\calM$.

\item
There is a shortest \st\ path that contains $z(t)$
and a point $z_s$ in $\{z_1(s),z_2(s)\}$.

This case is solved by using the similar approach as for Case 2 above.

\item
There is a shortest \st\ path that contains $z(s)$ and $z(t)$.

In this case, there must be a shortest \st\ path that is a concatenation of a shortest path
$\pi(s,z(s))$ from $s$ to $z(s)$ in
$B_s$, a shortest path $\pi(z(s),z(t))$ from $z(s)$ to $z(t)$,
and a shortest path $\pi(z(t),t)$ from $z(t)$ to $t$ in $B_t$.
The path $\pi(s,z(s))$ (resp., $\pi(z(t),t)$) can be found by using $\calD(B_s)$
(resp., $\calD(B_t)$).  It remains to show how to compute $\pi(z(s),z(t))$ below.
\end{enumerate}

Recall that we have defined a graph $G_E(g_s)$ in $B_s$ on the points of
$\calV(g_s)$, which consists of all points of $\calV(\calQ)$ lying on $g_s$. We
also find a gateway set $V_g(z(s),G_E(g_s))$ for $z(s)$ on $G_E(g_s)$. Similarly,
for $B_t$ and its gate $g_t$, we define $\calV(g_t)$, $G_E(g_t)$, and
$V_g(z(t),G_E(g_t))$. Let $G_E(\calM,g_s,g_t)$ be the graph formed by merging
$G_E(\calM)$, $G_E(g_s)$, and $G_E(g_t)$.
A shortest path from $z(s)$ to $z(t)$ can be found based on Lemmas \ref{lem:180} and \ref{lem:190}
below, which are similar to Lemmas \ref{lem:140} and \ref{lem:160}, respectively.

\begin{lemma}\label{lem:180}
If there is a shortest path from $z(s)$ to $z(t)$ containing a point in $\calV(g_s)\cap
\overline{z_1(s)z_2(s)}$ and a point in $\calV(g_t)\cap \overline{z_1(t)z_2(t)}$,
then there is a shortest path from $z(s)$ to $z(t)$ along $G_E(\calM,g_s,g_t)$ that contains
a gateway of $z(s)$ in $V_g(z(s),G_E(g_s))$ and a gateway of $z(t)$ in
$V_g(z(t),G_E(g_t))$.
\end{lemma}
\begin{proof}
Suppose there is a shortest $z(s)$-$z(t)$ path containing a point $p_s$ in $\calV(g_s)\cap
\overline{z_1(s)z_2(s)}$ and a point $p_t$ in $\calV(g_t)\cap \overline{z_1(t)z_2(t)}$.
Then by Lemma \ref{lem:130}, there is a shortest $z(s)$-$p_s$ path $\pi(z(s),p_s)$ along $G_E(g_s)$
containing a gateway of $z(s)$ in $V_g(z(s),G_E(g_s))$ and there is a shortest $p_t$-$z(t)$
path $\pi(p_t,z(t))$ along $G_E(g_t)$ containing a gateway of $z(t)$ in $V_g(z(t),G_E(g_t))$.
Since both $p_s$ and $p_t$ are in $\calV(\calQ)$, by Lemma \ref{lem:40}, there exists a shortest
$p_s$-$p_t$ path $\pi(p_s,p_t)$ along $G_E(\calM)$.

The concatenation of $\pi(z(s),p_s)$, $\pi(p_s,p_t)$, and $\pi(p_t,z(t))$ is a shortest $z(s)$-$z(t)$ path,
which is along the graph $G_E(\calM,g_s,g_t)$ and contains a gateway of $z(s)$ and a gateway of $z(t)$.
\end{proof}

\begin{lemma}\label{lem:190}
If no shortest $z(s)$-$z(t)$ path contains any point of
$\{z_1(s),z_2(s)\}\cup \{\calV(g_s)\cap \overline{z_1(s)z_2(s)}\}$, then
$z(t)$ must be horizontally visible to $\overline{z(s)z_2(s)}$;
similarly, if no shortest $z(s)$-$z(t)$ path contains any
point of $\{z_1(t),z_2(t)\}\cup \{\calV(g_t)\cap
\overline{z_1(t)z_2(t)}\}$, then $z(s)$ must be horizontally visible
to $\overline{z(t)z_2(t)}$.
\end{lemma}
\begin{proof}
We prove only the case when no shortest $z(s)$-$z(t)$
path contains any point of $\{z_1(s),z_2(s)\}\cup \{\calV(g_s)\cap \overline{z_1(s)z_2(s)}\}$,
$z(t)$ must be horizontally visible to $\overline{z(s)z_2(s)}$
(the other case is similar).

Let $\pi$ be a shortest $z(s)$-$z(t)$ path that intersects
$g_s$ at a point $p_s$ and intersects $g_t$ at a point $p_t$ (see
Fig.~\ref{fig:vispath}).  Let $\pi(p_1,p_2)$ denote
the subpath of $\pi$ between any two points $p_1$ and $p_2$ on $\pi$. We assume
$\pi(z(s),p_s)\subseteq B_s$, $\pi(p_t,z(t))\subseteq B_t$, and
$\pi(p_s,p_t)\subseteq \calM'$, since such a path $\pi$ always exists.

Let the points of $\calV(g_s)$ on $\overline{z_1(s)z_2(s)}$ be
$v_1,v_2,\ldots,v_m$ ordered along $\overline{z_1(s)z_2(s)}$ from $z_1(s)$ to
$z_2(s)$, and let $v_0=z_1(s)$ and $v_{m+1}=z_2(s)$.
Suppose $p_s$ is in the interior of
$\overline{v_iv_{i+1}}$, for some $i$ with $0\leq i\leq m$.
We define $u_i$, $u_{i+1}$, $\beta(u_i,u_{i+1})$, and $R$ in the same way
as in the proof of Lemma \ref{lem:160}.

Since $p_t$ is in $\calM$, by the proof of Lemma \ref{lem:160}, $p_t$
must be in the region $R$. Further, since $p_t$ is on $g_t\subseteq
\partial\calQ$, $p_t$ is on $\beta(u_i,u_{i+1})$. Thus,
$\beta(u_i,u_{i+1})\cap g_t$ is not empty. Since $g_t$ is a line segment,
$\beta(u_i,u_{i+1})\cap g_t$ is also a line segment.
Let $\overline{q_1q_2}=\beta(u_i,u_{i+1})\cap g_t$.
Thus, $p_t\in \overline{q_1q_2}$.

Recall that $z(t)$ is visible to $z_1(t)\in g_t$ and
$\overline{z(t)z_1(t)}$ is horizontal. Hence, to prove that
$z(t)$ is horizontally visible to $\overline{z(s)z_2(s)}$, it
suffices to prove that $z_1(t)$ is horizontally visible to
$\overline{z(s)z_2(s)}$. For this, it suffices to prove that $z_1(t)$ must
be on $\overline{q_1q_2}$ since every point on
$\overline{q_1q_2}\subseteq \beta(u_i,u_{i+1})$ is horizontally visible to
$\overline{z(s)z_2(s)}$. In the following, we prove $z_1(t)\in \overline{q_1q_2}$.

\begin{figure}[t]
\begin{minipage}[t]{\linewidth}
\begin{center}
\includegraphics[totalheight=1.5in]{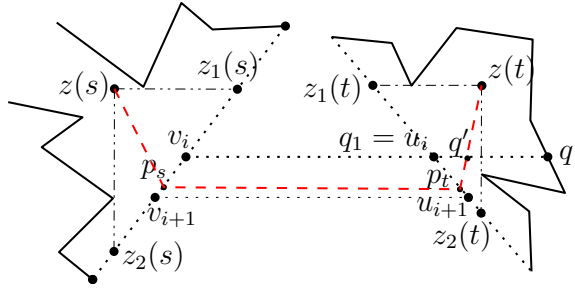}
\caption{\footnotesize Illustrating a possible shortest path (the red dashed curve) from
$z(s)$ to $z(t)$ crossing $\overline{u_iq}$. If this happens, we can
find another shortest path $\pi(z(t),q')\cup \overline{q'v_i}\cup
\overline{v_iz(s)}$, which contains $v_i$. In this example, $q_1=u_i$
and $q_2=u_{i+1}$.}
\label{fig:vispath}
\end{center}
\end{minipage}
\vspace*{-0.15in}
\end{figure}

Suppose to the contrary $z_1(t)\not\in \overline{q_1q_2}$ (see
Fig.~\ref{fig:vispath}).
Without loss of generality, we assume $q_1$ is closer to $z_1(t)$
than $q_2$. Since $p_t\in \overline{q_1q_2}$, $q_1\in \overline{p_tz_1(t)}\subseteq g_t$.
This implies that $q_1$ is not an endpoint of $g_t$, and thus
$q_1$ must be an endpoint of $\beta(u_i,u_{i+1})$ (i.e., one of
$u_i$ or $u_{i+1}$) since
$\overline{q_1q_2}=\beta(u_i,u_{i+1})\cap g_t$; assume $q_1=u_i$.
We extend $\overline{v_iu_i}$ horizontally into the bay $B_t$ until
hitting a point, say $q$, on the boundary of $B_t$ (see
Fig.~\ref{fig:vispath}). The horizontal segment
$\overline{u_iq}$ partitions $B_t$ into two sub-polygons such that $p_t$ and
$z_1(t)$ are in different sub-polygons. Since $\overline{z(t)z_1(t)}$ is
horizontal, $p_t$ and $z(t)$ are also in different sides of
$\overline{u_iq}$, implying
that the path $\pi(z(t),p_t)$ must intersect
$\overline{u_iq}$ since $\pi(z(t),p_t)$ is in $B_t$.
Let $q'$ be the intersection of $\pi(z(t),p_t)$ and $\overline{u_iq}$ (see
Fig.~\ref{fig:vispath}).
Then, the concatenation of  $\pi(z(t),q')$, $\overline{q'v_i}$,
and $\overline{v_iz(s)}$ is also a shortest path from $z(t)$ to $z(s)$ since
$ \overline{q'v_i}\cup \overline{v_iz(s)}$ is $xy$-monotone.
But this means that there is a shortest
$z(s)$-$z(t)$ path containing $v_i$, contradicting
with the lemma condition that no shortest $z(s)$-$z(t)$ path contains any point of
$\{z_1(s),z_2(s)\}\cup \{\calV(g_s)\cap \overline{z_1(s)z_2(s)}\}$.

The above arguments prove that $z_1(t)$ is on $\overline{q_1q_2}$.
The lemma thus follows.
\end{proof}

By Lemmas \ref{lem:180} and \ref{lem:190}, we can find a shortest
$z(s)$-$z(t)$ path by either using the gateways of $z(s)$ and
$z(t)$ in the merged graph
$G_E(\calM,g_s,g_t)$ or shooting horizontal and vertical rays from
$z(s)$ and $z(t)$. We have finished all possible cases for finding a
shortest \st\ path when the two query points are in different bays.
For each case, we compute a ``candidate" shortest \st\ path, and take
the one with the smallest length among all these cases (there are only a
constant number of them).

It remains to solve the canal case, i.e., when the query points are
in canals. The algorithm is similar to that for the bay case; the only
difference is that we have to take care of two gates for each
canal. Specifically, suppose $s$ is in a canal $C_s$ and $t$ is
in a canal $C_t$. If $C_s\neq C_t$, then there must be a shortest \st\ path
$\pi$ that intersects a gate of $C_s$ at a point $p_s$ and intersects a gate of
$C_t$ at a point $p_t$ such that the subpath $\pi(s,p_s)$ is in $C_s$,
the subpath $\pi(p_t,t)$ is in $C_t$, and the subpath $\pi(p_s,p_t)$
is in $\calM'$. Hence, we can use a similar approach as for the bay case to
find a shortest \st\ path by considering all four gate pairs of $C_s$ and
$C_t$. If $C_s=C_t$, while we can treat this case in the same way as for
the case of $C_s\neq C_t$, we need to consider one more possible situation when
a shortest \st\ path may be contained entirely in $C_s$, which is easy since $C_s$
is a simple polygon. If one of $C_s$ or $C_t$ is a bay, the case can be handled
in a similar fashion.

We summarize the whole algorithm in the proof of the following theorem.

\begin{theorem}\label{theo:20}
We can build a data structure of size $O(n+h^2\log h4^{\sqrt{\log
h}})$ in $O(n+h^2\log^{2}h4^{\sqrt{\log h}})$ time that can
answer each two-point $L_1$ shortest path query in $O(\log n)$ time (i.e., for
any two query points $s$ and $t$, the length of a shortest \st\ path can be found in $O(\log n)$
time and an actual path can be reported in additional time linear to the number of edges of the output path).
\end{theorem}
\begin{proof}
Our preprocessing algorithm consists of the following major steps.

\begin{enumerate}
\item
Compute a triangulation of the free space $\calM$ in
$O(n+h\log^{1+\epsilon} h)$ time
\cite{ref:Bar-YehudaTr94,ref:ChazelleTr91}. Then produce
all bays, canals, corridor paths, $\calM$, and $\calV(\calQ)$ in
$O(n+h\log h)$ time
\cite{ref:ChenA11ESA,ref:ChenCo12arXiv,ref:ChenL113STACS}.
\item
Compute the vertical and horizontal visibility decompositions of $\calP$
in $O(n+h\log^{1+\epsilon} h)$ time \cite{ref:Bar-YehudaTr94,ref:ChazelleTr91}.
Build a point location data structure \cite{ref:EdelsbrunnerOp86,ref:KirkpatrickOp83}
for each of the two decompositions in $O(n)$ time, which is used for
performing any vertical or horizontal ray-shooting in $O(\log n)$
time.
\item
Construct the graph $G_E(\calM)$ of size $O(n+h\sqrt{\log
h}2^{\sqrt{\log h}})$ in $O(n+h\log^{3/2} h2^{\sqrt{\log h}})$  time
by Lemma \ref{lem:50}.

\item
Perform the preprocessing of Lemma \ref{lem:70} in
$O(n+h\cdot \log^{3/2} h \cdot 2^{\sqrt{\log h}})$ time and
$O(n+h\cdot \sqrt{\log h} \cdot 2^{\sqrt{\log h}})$ space.

\item
Perform the preprocessing of Lemma \ref{lem:80} in
$O(n+h^2\log^{2} h 4^{\sqrt{\log h}})$ time and
$O(n+h^2\log h4^{\sqrt{\log h}})$ space.

\item
Compute a two-point Euclidean shortest path
query data structure $\calD(B)$ in each bay or canal $B$. Since the
total number of vertices of all bays and canals is $O(n)$, this
step takes $O(n)$ time.

\item
Construct the graph $G_E(g)$ for the gate $g$ of every bay or canal by Lemma
\ref{lem:120}. The total space for all such graphs is $O(h\sqrt{\log
h}2^{\sqrt{\log h}})$ and the total time for building all these graphs
is $O(n+h\log^{3/2} h2^{\sqrt{\log h}})$, as proved below.
First, each point of $\calV(\calQ)$ can be on at most one
bay or canal gate. Thus, the sum of $h_g$'s in Lemma \ref{lem:120} over all
gates $g$ is $O(|\calV(\calQ)|)$, which is $O(h)$. Second, the total number
of obstacle vertices of all bays and canals  is $O(n)$, and each canal has two
gates. Hence, the sum of $n_B$'s in Lemma \ref{lem:120} over all
bay and canals $B$ is $O(n)$.

\item
Perform the preprocessing of Lemma \ref{lem:150} for the graphs $G_E(g)$ of
all gates $g$, which can be done in totally $O(h\log^{3/2} h 2^{\sqrt{\log h}})$ time and
$O(h \sqrt{\log h}2^{\sqrt{\log h}})$ space.

\item
Merge the graph $G_E(\calM)$ and the graphs $G_E(g)$ for all gates
$g$ into a single graph $G_E(\calP)$, which takes $O(h)$ time since there are $O(h)$ points in
$\calV(\calQ)$.
Thus, the size of $G_E(\calP)$ is $O(h\sqrt{\log h} 2^{\sqrt{\log h}})$.

\item
For each node $v$ of $G_E(\calP)$, compute a shortest path tree rooted at $v$ in
$G_E(\calP)$. Maintain a shortest path length
table such that for any two nodes $u$ and
$v$ of $G_E(\calP)$, the length of a shortest path between $u$ and
$v$ in $G_E(\calP)$ can be obtained in $O(1)$ time. This step takes
$O(h^2\log h4^{\sqrt{\log h}})$ space and $O(h^2\log^{2}h4^{\sqrt{\log h}})$ time.

\item
Perform the preprocessing of Lemma \ref{lem:170} for each bay and canal,
which takes $O(n)$ space and $O(n)$ time in total.
\end{enumerate}

In summary, the total preprocessing space and time are
$O(n+h^2\log h4^{\sqrt{\log
h}})$ and $O(n+h^2\log^{2}h4^{\sqrt{\log h}})$, respectively.

Consider any two query points $s$ and $t$.
Next, we discuss our query algorithm that computes the
length of a shortest \st\ path in $O(\log n)$ time and reports an
actual path in additional time linear to the number of edges of the output
path. We will not explicitly discuss how to report an
actual path (which is similar to that in Lemma \ref{lem:80} and is easy).

First of all, as discussed in Section \ref{sec:pre},
we determine whether there exists a trivial shortest
\st\ path by shooting horizontal and vertical rays from $s$ and $t$,
which can be done in $O(\log n)$ time. In the following, we assume
that there is no trivial shortest \st\ path.
Depending on whether the query points are in the bays, canals, or the
ocean $\calM$, there are several possible cases.

\begin{description}
\item[Both query points are in $\calM$.]

In this case, we use the
algorithm for Lemma \ref{lem:80} to find a shortest \st\ path in
$O(\log n)$ time.

\item[Only one query point is in $\calM$.]

Without loss of generality, we assume that $s$ is in a bay or a canal $B$ and
$t$ is in $\calM$. Further, we assume that $B$ is a canal since the
case that $B$ is a bay can be considered as a special case.

Let $g^1$ and $g^2$ be the two gates of $B$. We define three points
$z(s,g^1)$, $z_1(s,g^1)$,
and $z_2(s,g^1)$ for $s$ in $B$ with respect to the gate $g^1$ in the same way as we
defined $z$, $z_1$, and $z_2$ before. Similarly, we define $z(s,g^2)$, $z_1(s,g^2)$,
and $z_2(s,g^2)$ for $s$ in $B$ with respect to the gate $g^2$.  These points
can be computed in $O(\log n)$ time by Lemma \ref{lem:170}. Then, we
compute the lengths of the following ``candidate'' shortest \st\ paths and
return the one with the smallest length.

\begin{enumerate}
\item
For each point $p\in \{z_1(s,g^1),z_2(s,g^1),z_1(s,g^2),z_2(s,g^2)\}$,
the path which is a concatenation of a shortest path $\pi(s,p)$ from $s$ to $p$
in $B$ and a shortest path $\pi(p,t)$ from $p$ to $t$ in $\calM'$.

The path $\pi(s,p)$ can be found in $O(\log n)$ time by using the data structure
$\calD(B)$ on $B$, and the path $\pi(p,t)$ can be found in
$O(\log n)$ time by Lemma \ref{lem:80}.

\item
For each point $p\in \{z(s,g^1),z(s,g^2)\}$, the path which is a
concatenation of a shortest path $\pi(s,p)$ from $s$ to $p$
in $B$ and a particular path $\pi(p,t)$ from $p$ to $t$.

The path $\pi(s,p)$ can be found in $O(\log n)$ time by using the data structure
$\calD(B)$ on $B$. The path $\pi(p,t)$ is determined as follows.
First, based on Lemma \ref{lem:160} (although $B$ is a bay in
Lemma \ref{lem:160}, the result also holds for canals because the lemma
was proved with respect to a gate regardless of whether it is a gate of
a bay or a canal), we check whether there exists a
path from $p$ to $t$ consisting of only two line segments, by performing
horizontal and vertical ray-shootings. If yes, then such a path is
$\pi(p,t)$. Otherwise, by Lemmas \ref{lem:140} and \ref{lem:160}, we
find a shortest path from $p$ to $t$ along the merged graph $G_E(\calP)$
by using the gateways of $p$ and the gateways of $t$, which can be obtained in $O(\log n)$
time by Lemmas \ref{lem:150} and \ref{lem:70}, respectively.
Since both $p$ and $t$ have
$O(\sqrt{\log h})$ gateways, a shortest $p$-$t$ path can be determined
in $O(\log n)$ time using the gateway graph as discussed at the end of
Section \ref{sec:newgraph}.
\end{enumerate}

\item[Neither query point is in $\calM$.] Let $B_s$ be the
bay or canal that contains $s$ and $B_t$ be the bay or canal that contains $t$.

If $B_s=B_t$ and $B_s$ is a bay, then by Lemma \ref{lem:90}, we can find a shortest \st\ path
by using the data structure $\calD(B_s)$ in $O(\log n)$ time.

Suppose $B_s\neq B_t$. Then we assume both $B_s$ and $B_t$ are canals since
the other cases are just special cases of this case. Let $g^1_s$ and $g^2_s$ be
the two gates of $B_s$ and $g^1_t$ and $g^2_t$ be
the two gates of $B_t$. Similarly as before, we define the points $z(s,g_s^i)$,
$z_1(s,g_s^i)$, and $z_2(s,g_s^i)$ for $s$ with respect to $g_s^i$,
and $z(t,g_t^i)$, $z_1(t,g_t^i)$, and $z_2(t,g_t^i)$ for $t$ with respect to $g_t^i$, for
$i=1,2$. These points
can all be determined in $O(\log n)$ time by Lemma \ref{lem:170}. Then we
compute the lengths of the following ``candidate'' shortest \st\ paths and
return the one with the smallest length.

\begin{enumerate}
\item
For each pair of points $p_s$ and $p_t$ such that $p_s\in \{z_1(s,g_s^1),z_2(s,g_s^1),z_1(s,g_s^2),z_2(s,g_s^2)\}$ and $p_t\in \{z_1(t,g_t^1),z_2(t,g_t^1),z_1(t,g_t^2),z_2(t,g_t^2)\}$,
the path which is a concatenation of a shortest path $\pi(s,p_s)$ from $s$ to $p_s$
in $B_s$, a shortest path from $p_s$ to $p_t$ in $\calM'$, and a shortest path $\pi(p_t,t)$ from $p_t$ to $t$ in $B_t$.

The paths $\pi(s,p_s)$ and $\pi(p_t,t)$ can be found in
$O(\log n)$ time by using $\calD(B_s)$ and $\calD(B_t)$, respectively.
The path $\pi(p_s,p_t)$ can be obtained in
$O(\log n)$ time by Lemma \ref{lem:80}.

\item
For each point $p_s\in \{z(s,g_s^1),z(s,g_s^2)\}$ and each point $p_t\in
\{z_1(t,g_t^1),z_2(t,g_t^1),z_1(t,g_t^2),z_2(t,g_t^2)\}$, the path
which is a concatenation of a shortest path from $s$ to $p_s$ in $B_s$,
a particular path $\pi(p_s,p_t)$ from $p_s$ to $p_t$, and a shortest
path from $p_t$ to $t$ in $B_t$.

The paths $\pi(s,p_s)$ and $\pi(p_t,t)$ can be found in
$O(\log n)$ time by using $\calD(B_s)$ and $\calD(B_t)$, respectively.
Since $p_t$ is in $\calM$,
the particular path $\pi(p_s,p_t)$ is defined similarly as
the path $\pi(p,t)$ in the second subcase of
the above case when only one query point $t$ is in $\calM$
and thus can be obtained by the similar approach.

\item
For each point $p_s\in
\{z_1(s,g_s^1),z_2(s,g_s^1),z_1(s,g_s^2),z_2(s,g_s^2)\}$ and each point
$p_t\in \{z(t,g_t^1),z(t,g_t^2)\}$, the path which is a concatenation
of a shortest path from $s$ to $p_s$ in $B_s$, a particular path
$\pi(p_s,p_t)$ from $p_s$ to $p_t$, and a shortest path from $p_t$ to
$t$ in $B_t$.

This subcase is symmetric to the subcase immediately above and can be handled similarly.

\item
For each point $p_s\in \{z(s,g_s^1),z(s,g_s^2)\}$ and each point $p_t\in \{z(t,g_t^1),z(t,g_t^2)\}$, the path
which is a concatenation of a shortest path from $s$ to $p_s$ in $B_s$, a particular path $\pi(p_s,p_t)$ from $p_s$ to $p_t$, and a shortest path from $p_t$ to $t$ in $B_t$.

The paths $\pi(s,p_s)$ and $\pi(p_t,t)$ can be found in
$O(\log n)$ time by using $\calD(B_s)$ and $\calD(B_t)$, respectively.
The particular path $\pi(p_s,p_t)$ is determined similarly as the path
$\pi(p,t)$ in the second subcase of the above case when only
one query point $t$ is in $\calM$, but based on Lemmas \ref{lem:180} and
\ref{lem:190} instead. Note that although $B_s$ and $B_t$ are bays in
these lemmas, the results also hold for canals (actually, they
are proved with respect to two gates regardless of whether they are gates of
bays or canals). Specifically, we determine $\pi(p_s,p_t)$ as follows.
Based on Lemma \ref{lem:190}, we first check whether there exists a
path from $p_s$ to $p_t$ consisting of only two line segments, by
horizontal and vertical ray-shootings. If yes, then such a path is
$\pi(p_s,p_t)$. Otherwise, by Lemmas \ref{lem:180} and \ref{lem:190}, we
find a shortest $p_s$-$p_t$ path along the merged graph $G_E(\calP)$
by using the gateways of $p_s$ and the gateways of $p_t$, which
can be computed in $O(\log n)$ time by Lemma \ref{lem:150}.
Since both $p_s$ and $p_t$ have
$O(\sqrt{\log h})$ gateways, a shortest $p_s$-$p_t$ path can be obtained
in $O(\log n)$ time using the gateway graph as discussed in Section \ref{sec:newgraph}.
\end{enumerate}

\end{description}

Finally, if $B_s=B_t$ and $B_s$ is a canal, then the algorithm is similar as for
the above case with the difference that we must consider an additional
``candidate'' path that is a shortest \st\ path inside $B_s$, which
can be found in $O(\log n)$ time by using the data structure $\calD(B_s)$.

Hence, in any case, we find a shortest \st\ path in $O(\log n)$ time.
The theorem thus follows.
\end{proof}

If we replace all enhanced graphs, e.g., $G_E(\calM)$ and $G_E(g)$ for
every gate $g$, by the corresponding graphs similar to $G_{old}$ in \cite{ref:ChenSh00}
as discussed in Section \ref{sec:pre}, then we obtain the following results.

%
%
%
\begin{corollary}
We can build a data structure in $O(n+h^2\log^2h)$ time and space,
such that each two-point shortest path query is answered in
$O(\log n+\log^2 h)$ time; alternatively,
we can build a data structure in $O(nh\log h+h^2\log^2 h)$ time and
$O(nh\log h)$ space, such that each two-point shortest path query is
answered in $O(\log n\log h)$ time.
\end{corollary}
\begin{proof}
If we replace all the enhanced graphs $G_E(\calM)$ and $G_E(g)$ for every gate $g$ of the
bays and canals by the graphs similar to $G_{old}$ in \cite{ref:ChenSh00}
as discussed in Section \ref{sec:pre}, then the size of the new merged graph,
denoted by $G_{old}(\calP)$, becomes $O(h\log h)$ instead of
$O(h\sqrt{\log h}2^{\sqrt{\log h}})$.
Hence, the data structure for Theorem \ref{theo:20}
needs $O(n+h^2\log^2 h)$ space and can be built in $O(n+h^2\log^2 h)$
time by using the approach in \cite{ref:ChenSh00}.
However, using the new graph $G_{old}(\calP)$, each query for any two points in
$\calM$ can be answered in $O(\log^2 h)$ time because there are $O(\log
h)$ gateways for each query point. Therefore, any general two-point
shortest path query can be answered in $O(\log^2 h+\log n)$ time, by
using a similar query algorithm as in Theorem \ref{theo:20}. We omit the details.

In the result above,
we compute a shortest path tree rooted at each node in the merged graph $G_{old}(\calP)$.
Alternatively, we can compute a shortest path map in the free space $\calF$ for
each node $v$ of $G_{old}(\calP)$, such that given any query point
$t$, the length of a shortest path from $v$ to $t$ can be found in
$O(\log n)$ time and an actual path can be reported in additional
time linear to the number of edges of the output path.
Each such shortest path map is of size $O(n)$ and can be computed in
$O(n+h\log h)$ time
\cite{ref:ChenA11ESA,ref:ChenCo12arXiv,ref:ChenL113STACS} (after
the free space $\calF$ is triangulated). Since the size
of $G_{old}(\calP)$ is $O(h\log h)$, the overall preprocessing time
and space are $O(nh\log
h+h^2\log^2 h)$ and $O(nh\log h)$, respectively. For querying, since
a query point may have
$O(\log h)$ gateways and for each gateway $v$, we can determine the
shortest path from $v$ to the other query point in $O(\log n)$ time,
the total query time is $O(\log h\log n)$. We omit the details.
\end{proof}

%
%
%

\section{The Weighted Rectilinear Case}
\label{sec:weighted}

In this section, we extend our techniques in Section
\ref{sec:newgraph} to the weighted rectilinear case. In the
weighted rectilinear case, every polygonal obstacle $P\in \calP$ is
{\it rectilinear} and {\it weighted}, i.e., each edge of $P$ is
either horizontal or vertical and $P$ has a weight
$w(P)\geq 0$ ($w(P)=+\infty$ is possible). If a line segment $e$ is in $P$,
then the {\em weighted length} of $e$ is $x\cdot (1+w(P))$, where $x$ is the $L_1$ length of
$e$. Any polygonal path $\pi$ can be divided into a sequence of maximal
line segments such that each segment is contained in the same obstacle
or in the free space $\calF$; the {\em weighted length} of $\pi$ is the
sum of the weighted lengths of all maximal line segments of $\pi$.

Consider a vertex $v$ of any rectilinear obstacle $P$ such that the interior angle
of $P$ at $v$ is $3\pi/2$. We define the {\em
internal projections} of $v$ on the boundary $\partial P$ of $P$ as follows.
Suppose $\overline{u_1v}$ and $\overline{u_2v}$ are the two edges of
$P$ incident to $v$. We extend $\overline{u_1v}$ into the interior of
$P$ along the direction from $u_1$ to $v$ until we hit $\partial P$
at the first point, which is an {\em internal projection} of $v$; similarly,
we define another interval projection of $v$  by extending
$\overline{u_2v}$. Internal projections are used to control shortest
paths that pass through the interior of obstacles.

The ``visibility'' in the weighted case is defined in a slightly
different way: Two points $p$ and $q$ are {\em visible} to each other
if $\overline{pq}$ is entirely in either $\calF$ or an obstacle.

Let $\mathcal{V}$ be the set of all obstacle vertices of $\calP$, their internal
projections, and all type-1 Steiner points. Then $|\calV|=O(n)$.
We build a graph $G_E(\calV)$ on $\calV$ similar to the one presented in
Section \ref{sec:newgraph}, with the following differences.
(1) 
The visibility here is based on the new definition above.
(2) 
Since a path can travel through
the interior of any obstacle, for each cut-line $l$, an edge
in $G_E(\calV)$ connects every two consecutive Steiner points on
$l$, whose weight is the weighted length of the line segment
connecting the two points.
(3) 
In addition to the vertical
cut-lines, there are also horizontal cut-lines, which are defined
similarly and have type-2 and type-3 Steiner points defined on them similarly
to those on the vertical cut-lines.
Thus, $G_E(\calV)$ has $O(n\sqrt{\log n}2^{\sqrt{\log n}})$ nodes and edges.


\begin{lemma}\label{lem:200}
The graph of $G_E(\calV)$ can be built in $O(n\log^{3/2}n2^{\sqrt{\log n}})$
time.
\end{lemma}
\begin{proof}
We obtain all internal projections of $\calV$ by computing the
horizontal and vertical visibility decompositions of every
obstacle in $\calP$. We find the four projection points on $\partial\calP$ (i.e., $p^r,p^l,p^u$,
and $p^d$) for all obstacle vertices $p$ of $\calP$ in $O(n\log n)$ time by
computing the horizontal and vertical visibility decompositions of $\calF$.
These can be all done in totally $O(n\log n)$ time.

Then we compute the vertical and horizontal cut-line trees, which takes
$O(n\log n)$ time since $|\calV|=O(n)$. Next, we compute
the Steiner points and the graph edges. Below, we only show how to compute those
related to the vertical cut-lines; those related to the horizontal
cut-lines can be computed in a similar way.
Let $T^v(\calV)$ denote the vertical cut-line tree.

As in Lemma \ref{lem:10}, we can compute the type-2 and type-3 Steiner
points on all cut-lines of $T^v(\calV)$ by traversing $T^v(\calV)$ in a top-down
manner. Since the internal projections and $\{p^r,p^l,p^u,p^d\}$ for
each obstacle vertex $p$ have been obtained, we can compute all
$O(n\sqrt{\log n}2^{\sqrt{\log n}})$ such Steiner points in $O(n\sqrt{\log
n}2^{\sqrt{\log n}})$ time; the corresponding horizontal graph
edges connecting these Steiner points and the points of $\calV$ can
also be computed.

It remains to compute the graph edges connecting every pair of consecutive Steiner points
on each cut-line of $T^v(\calV)$, which takes $O(n\log^{3/2} n2^{\sqrt{\log n}})$ time by a
plane sweeping algorithm, as follows. We first sort all Steiner points on each cut-line.
We then sweep a vertical line $L$ from left to
right and use a balanced binary search tree $T$ to maintain the
intervals between the obstacle edges of $\calP$ intersecting $L$. By standard
techniques, we augment $T$ to also maintain the weighted length
information along $L$ such that for any two points $p$ and $q$ on
$L$, the weighted length of $\overline{pq}$ can
be obtained in $O(\log n)$ time using $T$. During the sweeping,
when $L$ encounters a cut-line $l$, for every two consecutive Steiner
points $p$ and $q$ on $l$, we use $T$ to determine in $O(\log n)$ time the weighted
length of the edge connecting $p$ and $q$.
Since there are $O(n\sqrt{\log n} 2^{\sqrt{\log
n}})$ pairs of consecutive Steiner points on all cut-lines, it takes
$O(n\log^{3/2} n2^{\sqrt{\log n}})$ time to compute all these graph edges.

Hence, we can build the graph $G_E(\calV)$ in $O(n\log^{3/2} n2^{\sqrt{\log n}})$ time.
\end{proof}

Consider any two query points $s$ and $t$. For simplicity of discussion, we assume that both $s$ and $t$ are
in $\calF$ (the general case can also be handled similarly). With a preprocessing of $O(n^2)$
time and space, a shortest \st\ path that does not
contain any vertex of $\calV$ can be found in $O(\log n)$ time \cite{ref:ChenSh00}. Thus in the
following, we focus on finding a shortest \st\ path containing at lease one
vertex of $\calV$.

Let $Y(s)$ be the set of $s$ and
the four projections of $s$ on $\partial\calP$, i.e., $Y(s)=\{s,s^l,s^r,s^u,s^d\}$;
similarly, let $Y(t)=\{t,t^l,t^r,t^u,t^d\}$. It was shown in
\cite{ref:ChenSh00} that it suffices to find a shortest path
from $p$ to $q$ containing a vertex of $\calV$ for every $p\in Y(s)$ and every $q\in Y(t)$.
With a little abuse of notation,
we let $s$ be any point in $Y(s)$ and $t$ be any point in $Y(t)$. Our
goal is to find a shortest \st\ path that contains at lease one vertex of $\calV$.
Unless otherwise indicated, any shortest \st\ path
mentioned below refers to a shortest \st\ path that contains a vertex of $\calV$.

In \cite{ref:ChenSh00}, similar to the discussions in Section \ref{sec:pre},
$O(\log n)$ gateways for $s$ and $O(\log n)$ gateways for $t$ were defined,
such that any shortest \st\ path must contain a gateway of $s$ and a
gateway of $t$. Hence by using the gateway graph, a shortest \st\ path can be
found in $O(\log^2 n)$ time.

Based on our enhanced graph $G_E(\calV)$, as in Section \ref{sec:newgraph},
we define a new gateway set $V_g(s,G_E(\calV))$ of size $O(\sqrt{\log n})$ for
$s$ and a new gateway set $V_g(t,G_E(\calV))$ of size $O(\sqrt{\log n})$ for $t$.
The gateway set $V_g(s,G_E(\calV))$ contains $O(\sqrt{\log n})$  Steiner points on the
vertical cut-lines defined in the same way as those in $V^2_g(s,G_E)$ in Section \ref{sec:newgraph}; similarly, $V_g(s,G_E(\calV))$ also contains $O(\sqrt{\log n})$  Steiner points on the
horizontal cut-lines.  The gateway set $V_g(t,G_E(\calV))$ is defined similarly.
Using a similar proof as for Lemma \ref{lem:20}, we can show that
there exists a shortest \st\ path containing a gateway of $s$ in $V_g(s,G_E(\calV))$ and a
gateway of $t$ in $V_g(t,G_E(\calV))$. Next, we show how to compute the two gateway sets and
(the weights of) their gateway edges. Below, we discuss only the case for $s$.

The fractional cascading approach \cite{ref:ChazelleFr86} used in
Section \ref{sec:newgraph} can still compute the gateway
set $V_g(s,G_E(\calV))$ in $O(\log n)$ time, but it cannot compute the weights of the
gateway edges in $O(\log n)$ time for the following reasons.
Consider a gateway $v\in V_g(s,G_E(\calV))$, say on a vertical cut-line $l$. Then
there is a gateway edge $(s,v)$ that consists of two line segments
$\overline{ss_h(l)}$ and $\overline{s_h(l)v}$ (recall that $s_h(l)$ is the horizontal projection of $s$ on $l$). Hence, the weighted length
of the edge $(s,v)$ is the sum of the weighted lengths of these two line segments.
It was shown in \cite{ref:ChenSh00} that $\overline{ss_h(l)}$ must be in the free
space (since $s$ is in $\calF$); thus, the weighted length of
$\overline{ss_h(l)}$ is easy to compute. However, the vertical segment
$\overline{s_h(l)v}$ may intersect multiple
obstacles \cite{ref:ChenSh00}. We give an algorithm to compute in
$O(\log n)$ time the gateways and the weights of the
gateway edges for $s$ in the next lemma.


\begin{lemma}\label{lem:210}
With a preprocessing of $O(n^2\log n)$ time and $O(n^2)$ space,
the gateways of $V_g(s,G_E(\calV))$ for $s$ and their weighted edges
can be computed in $O(\log n)$ time.
\end{lemma}
\begin{proof}
We discuss only how to
compute the gateways of $V_g(s,G_E(\calV))$ that are on the vertical cut-lines
since those on the horizontal cut-lines can be computed similarly.
Further, for simplicity of discussion, we only compute the gateways of
$V_g(s,G_E(\calV))$ above $s$ (i.e., above the horizontal line through $s$) since those
below $s$ can be computed similarly.
Below, with a little abuse of notation, we let $V_g(s,G_E(\calV))$ refer to the
set of its gateways on the vertical cut-lines and above $s$.

We follow the terminology in Section \ref{sec:newgraph}.
Recall that $s$ has $O(\log n)$ projection cut-lines in the vertical cut-line tree $T^v(\calV)$.
Let $S_l$ be the set of all projection cut-lines of $s$ in $T^v(\calV)$.
For each projection cut-line $l\in S_l$, let $v(l)$ be the Steiner
point on $l$ immediately above the horizontal projection $s_h(l)$ of $s$ on $l$.
Let $S_v=\{v(l)\ |\ l\in S_l\}$. By their definitions,
$V_g(s,G_E(\calV))$ is a subset of $S_v$ (since each gateway of $V_g(s,G_E(\calV))$
is on a {\it relevant} projection cut-line of $s$ in $T^v(\calV)$).
Hence, to compute $V_g(s,G_E(\calV))$ and their gateway
edges, it suffices to compute the set $S_v$ and the weighted lengths of
$\overline{ss_h(l)}\cup \overline{s_h(l)v(l)}$ for all projection cut-lines $l\in S_l$.
Since $\overline{ss_h(l)}$ is in
$\calF$ for any projection cut-line $l$ of $s$ \cite{ref:ChenSh00} (because $s\in
\calF$), it suffices to compute the weighted length of
$\overline{s_h(l)v(l)}$. Below, for any line segment $\overline{ab}$, let $d_w(\overline{ab})$
denote the weighted length of $\overline{ab}$. Let $S_w=\{\overline{s_h(l)v(l)}\ |\ l\in S_l\}$.

We use fractional cascading \cite{ref:ChazelleFr86}
to obtain $S_v$ in $O(\log n)$ time,
with a similar approach as for Lemma \ref{lem:30}.
To compute the weighted lengths of the segments in $S_w$, we need to build another fractional
cascading data structure in the preprocessing.

For every cut-line $l$ of $T^v(\calV)$, we compute the intersections of $l$ with all
obstacle edges of $\calP$; let $I(l)$ be the set of such intersections.
Clearly, $|I(l)|=O(n)$. We sort these intersections and the Steiner points on $l$ to obtain a
sorted list $I'(l)$. For all $n$ cut-lines of $T^v(\calV)$,
this takes totally $O(n^2\log n)$ time,
because the total number of Steiner points is $O(n\sqrt{\log
n}2^{\sqrt{\log n}})$ (which is $O(n^2)$) and the total number of
intersections between the cut-lines and the obstacle edges is $O(n^2)$.

Consider the sorted set $I'(l)$ for any cut-line $l$ of $T^v(\calV)$. For
any two consecutive points $p_1$ and $p_2$ in $I'(l)$, the entire segment
$\overline{p_1p_2}$ is either in $\calF$ or in the same
obstacle. From top to bottom in $I'(l)$, for each point $p\in I'(l)$,
we compute the weighted length $d_w(\overline{pp^*})$ and associate it with
$p$, where $p^*$ is the highest point in $I'(l)$.  Further, for each point $p\in I'(l)$,
we maintain a weight $w_p$, defined as follows: Suppose $p'$ is the point in
$I'(l)$ immediately below $p$; if the interior of $\overline{pp'}$ is
contained in an obstacle, then $w_p$ is the weight of that obstacle, and
$w_p=0$ otherwise.
Since $I'(l)$ is sorted, computing such information in $I'(l)$ takes $O(|I'(l)|)$ time.
With such information, for any query point
$q$ on $l$, suppose $p$ is the point in $I'(l)$ that is
immediately above $q$; then we have
$d_w(\overline{qp^*})=d_w(\overline{pp^*})+(1+w_p)\cdot d_w(\overline{pq})$.
Hence, once we know the point $p$ for $q$,
 $d_w(\overline{qp^*})$ can be computed in $O(1)$ time; further,
for any point $p'$ in $I'(l)$ above $q$, we have $d_w(\overline{qp'})=d_w(\overline{qp^*})-d_w(\overline{p'p^*})$, which
is computed in $O(1)$ time since the value $d_w(\overline{p'p^*})$ is already stored at $p'$.

In the preprocessing, we build another fractional cascading data structure on $T^v(\calV)$
and the sorted lists $I'(l)$ for all cut-lines $l$ of $T^v(\calV)$, which takes $O(n^2)$ space and $O(n^2\log n)$ time.

For any query point $s$, we first use a similar approach as for Lemma \ref{lem:30} to compute
the set $S_v$ in $O(\log n)$ time. For each projection cut-line $l\in S_l$, let $v'(l)$
be the point in $I'(l)$ immediately above $s_h(l)$. Note that $v'(l)$ is between $v(l)$
and $s_h(l)$. We can use the above fractional cascading data structure to compute the points
$v'(l)$ for all $l\in S_l$ in $O(\log n)$ time (since the cut-lines of $S_l$ are at the nodes
of a path from the root to a leaf in $T^v(\calV)$). Then for each $l\in S_l$, to compute
$d_w(\overline{s_h(l)v(l)})$, as discussed above, we have  $d_w(\overline{s_h(l)v(l)})=d_w(\overline{s_h(l)p^*})-d_w(\overline{v(l)p^*})$, where $p^*$ is the highest
point in $I'(l)$ and $d_w(\overline{s_h(l)p^*})=d_w(\overline{v'(l)p^*})+(1+w_{v'(l)})\cdot d_w(\overline{s_h(l)v'(l)})$.
Since both $v(l)$ and $v'(l)$ have been computed, $d_w(\overline{s_h(l)v(l)})$ is obtained
in $O(1)$ time. Hence, the weighted lengths of all segments in $S_w$ are computed in $O(\log n)$ time.

The lemma thus follows.
\end{proof}

The following theorem summarizes our algorithm for the weighted rectilinear case.

%
%
%
%
\begin{theorem}\label{theo:30}
For the weighted rectilinear case, we can build a data structure of
size $O(n^2\log n4^{\sqrt{\log n}})$ in $O(n^2\log^{2}n4^{\sqrt{\log
n}})$ time that can answer each query in $O(\log n)$ time (i.e., for
any two query points $s$ and $t$, the weighted length of a shortest \st\ path can be found in $O(\log n)$
time and an actual path can be reported in additional time linear to the number of edges of the output path).
\end{theorem}
\begin{proof}
In the preprocessing, we compute the graph $G_E(\calV)$ by Lemma \ref{lem:200}. For each node
$v$ of $G_E(\calV)$, we compute a shortest path tree rooted at $v$ in $G_E(\calV)$. We maintain
a shortest path length table such that for any two nodes $u$ and $v$ in $G_E(\calV)$, the
(weighted) length of the shortest path from $u$ to $v$ in $G_E(\calV)$ is obtained in $O(1)$ time.
Computing all shortest path trees in $G_E(\calV)$ takes
$O(n^2\log n4^{\sqrt{\log n}})$ space and $O(n^2\log^{2}n4^{\sqrt{\log n}})$ time. We also
perform the preprocessing for Lemma \ref{lem:210}. Hence, the preprocessing takes
$O(n^2\log n4^{\sqrt{\log n}})$ space and $O(n^2\log^{2}n4^{\sqrt{\log n}})$ time in total.

Consider any two query points $s$ and $t$. First, we use the approach in \cite{ref:ChenSh00}
to find a shortest \st\ path that does not contain any obstacle vertex of $\calP$ (if any),
after a preprocessing of $O(n^2)$ time and space. Below, we focus on finding a shortest \st\
path containing an obstacle vertex of $\calP$, which must contain a gateway of $s$ in
$V_g(s,G_E(\calV))$ and a gateway of $t$ in $V_g(t,G_E(\calV))$. By Lemma \ref{lem:210}, we
can compute both $V_g(s,G_E(\calV))$ and $V_g(t,G_E(\calV))$ in $O(\log n)$ time. Then, a
shortest \st\ path can be found by building a gateway graph (as discussed in
Section \ref{sec:newgraph}) in $O(\log n)$ time since the sizes of both $V_g(s,G_E(\calV))$
and $V_g(t,G_E(\calV))$ are $O(\sqrt{\log n})$. As in \cite{ref:ChenSh00}, after the shortest
\st\ path length is computed, an actual shortest \st\ path can be reported by using the shortest
path trees of the nodes in $G_E(\calV)$, in time linear to the number of edges of the output path.

The theorem thus follows.
\end{proof}


\bibliographystyle{plain}

\end{document}